\documentclass[12pt]{article} 
\usepackage{eurosym}
\usepackage[leqno]{amsmath}
\usepackage[amsmath,thmmarks,thref,framed]{ntheorem}
\usepackage{xcolor}
\usepackage[colorinlistoftodos]{todonotes}

\usepackage{authblk}

\usepackage{ascii,amsfonts,amssymb,amstext,bbm,calrsfs,color,dsfont,epsfig,keystroke,latexsym,lmodern,psfrag,textcomp,units,url,yfonts}

\usepackage{hyperref}
\hypersetup{colorlinks=true}
\hypersetup{urlcolor=black}
\hypersetup{linkcolor=blue}
\hypersetup{citecolor=blue}

\definecolor{black}{cmyk}{1,1,1,1}
\definecolor{blue}{cmyk}{1,1,0,0}

\theoremstyle{break}
\theoremheaderfont{\small\normalfont\sffamily}
\theorembodyfont{\itshape}
\theoremsymbol{\hfill\EOT}
\theoremseparator{:}
\newtheorem{theorem}{Theorem}
\theoremclass{Theorem}
\theoremstyle{plain}
\theoremheaderfont{\sc}
\theoremsymbol{\hfill\EOT}
\theoremseparator{.}
\newshadedtheorem{conjecture}{Conjecture}
\theoremstyle{break}
\theorembodyfont{\slshape}
\theoremsymbol{\hfill\EOT}
\newtheorem{lemma}{Lemma}
\theoremsymbol{\hfill\LF}
\theoremnumbering{arabic}
\newtheorem{proposition}{Proposition}
\theorembodyfont{\slshape}
\theoremsymbol{\hfill\LF}
\theoremnumbering{arabic}
\newtheorem{corollary}{Corollary}
\newshadedtheorem{corshaded}{Corollary}
\theoremstyle{plain}
\theoremheaderfont{\bfseries}
\theorembodyfont{\small\fontfamily{cmbr}\selectfont}
\theoremsymbol{\hfill\SI}
\theoremseparator{.}
\newtheorem{remark}{Remark}
\theoremstyle{break}
\theoremheaderfont{\bfseries}
\theorembodyfont{\small\fontfamily{cmbr}\selectfont}
\theoremsymbol{\hfill\SI}
\theoremseparator{.}

\theoremstyle{break}
\theoremheaderfont{\ttfamily}
\theorembodyfont{\small\sffamily}
\theoremsymbol{\ensuremath{\ast}}
\theoremseparator{.}

\theoremheaderfont{\itshape}
\theorembodyfont{\normalfont}
\theoremsymbol{\hfill\BS}
\theoremseparator{.}
\newtheorem{definition}{Definition}
\theoremstyle{break}
\theoremheaderfont{\normalfont\ttfamily}
\theorembodyfont{\small\fontfamily{cmbr}\selectfont}
\theoremsymbol{\hfill\ENQ}
\newtheorem{notation}{Notation}
\theoremheaderfont{\normalfont\ttfamily}
\theorembodyfont{\small\fontfamily{cmbr}\selectfont}
\theoremsymbol{\hfill\ENQ}

\theoremheaderfont{\normalfont\ttfamily}
\theorembodyfont{\small\fontfamily{cmbr}\selectfont}
\theoremsymbol{\hfill\ENQ}
\newtheorem{assumption}{Assumption}
\theoremheaderfont{\itshape}
\theorembodyfont{\upshape}
\theoremstyle{nonumberplain}
\theoremseparator{.}
\theoremsymbol{\End}
\newtheorem{proof}{Proof}

\def\1{\mathbf 1}
\def\a{\Gamma}
\def\A{\mathfrak{A}}

\def\B{\mathrm B}
\def\BL{\mathcal B}
\def\Bil#1{\left<\B,\,#1\B\right>}
\def\C{C^{*}}
\def\cal{\mathcal{C}}
\def\CAR{\mathrm{CAR}}
\def\CP{\mathbb{C}}
\def\d{\mathrm d}
\def\e{\mathrm e}
\def\eqt#1{\texorpdfstring{#1}{}}
\def\E{\mathfrak e}
\def\fp{\mathfrak{p}(\H,\a)}

\def\Fs{\mathcal F}
\def\g{\mathfrak g}

\def\H{\mathcal H}
\def\h{\mathfrak h}

\def\ii{\mathrm i}
\def\Imm{\Im\mathfrak{m}}
\def\inner#1{\left< #1 \right>}

\def\L{\mathfrak L}
\def\LO{\mathcal L}

\def\N{\mathbb N}

\def\p{\partial}

\def\Pfz{\mathcal{P}_{\text{f}}(\Z^{d})}
\def\q{\mathfrak q}

\def\R{\mathbb R}

\def\Ree{\Re\mathfrak{e}}

\def\sCAR{\mathrm{sCAR}(\H,\a)}
\def\spec{\mathrm{spec}}
\def\spin{\mathfrak S}
\def\states{\mathfrak E}
\def\supp{\mathrm{supp}}

\def\tr{\mathrm{tr}}
\def\Tr{\mathrm{Tr}}

\def\x{\mathbf x}
\def\X{\mathcal X}

\def\y{\mathbf y}
\def\Z{\mathbb Z}

\usepackage{times}

\input{tcilatex}

\usepackage{anysize}

\marginsize{2cm}{2cm}{1cm}{1cm}

\begin{document}

\title{A $\Z_{2}$--Topological Index for Quasi--free Fermions}
\author[1]{N. J. B. Aza\thanks{Corresponding author: njavierbuitragoa@gmail.com}}
\author[1]{A. F. Reyes-Lega\thanks{anreyes@uniandes.edu.co}}
\author[1]{L. A. Sequera M.\thanks{la.sequera@uniandes.edu.co}}
\affil[1]{Departamento de F\'{i}sica, Universidad de Los Andes}

\date{\today}
\maketitle

\begin{abstract}
We use infinite dimensional self--dual $\CAR$ $\C$--algebras to study a $\Z_{2}$--index, which classifies free--fermion systems embedded on $\Z^{d}$ disordered lattices. Combes--Thomas estimates are pivotal to show that the $\Z_{2}$--index is uniform with respect to the size of the system. We additionally deal with the set of ground states to completely describe the mathematical structure of the underlying system. Furthermore, the weak$^{*}$--topology of the set of linear functionals is used to analyze paths connecting different sets of ground states.

\bigskip

\end{abstract}

%\section*{Conflict of interest statement}
%Authors declare no conflict of interests for this article.

%TCIMACRO{\TeXButton{\tableofcontents }{\tableofcontents}}%
%BeginExpansion
\tableofcontents
%EndExpansion

\section{Introduction}
A considerable number of mathematical results concerning gapped Hamiltonians of fermions has been achieved in recent years. Among the most important ones are topological protection under small perturbations and the persistence of the spectral gap for interacting fermions \cite{hastings2019stability,deRoecK_Salm}. We study a \emph{$\mathbb{Z}_{2}$--projection index} ($\Z_{2}$--PI), that is the one introduced long ago by Araki--Evans in their work where the two possible thermodynamic phases of the classical two-dimensional Ising model are characterized using operator algebras technologies \cite{araEva}. Here we deal with disordered \emph{free--fermion} systems on the lattice within the  mathematical framework  of \emph{self--dual $\CAR$ $\C$--algebras}. In particular, their structure is useful to study interacting fermion systems, even with superconducting terms \cite{LD1}. To be precise, the $\Z_{2}$--PI is defined in terms of well--defined \emph{basis} projections related to a self--adjoint operator, which typically is the \emph{Hamiltonian} of the system acting on a separable Hilbert space $\H$. See Definition \ref{def basis projection} below. Thus, the $\Z_{2}$--PI can be used to discriminate parity sectors in the set of quasi--free ground states of fermionic systems \cite{EK98,Varilly}.\\
A very important problem in this context is the classification of topological matter in general. The current classification scheme can be traced back to Dyson's~\cite{Dyson} classical work from 1962.  Of course  that work did not contemplate topological aspects for such systems, but it provided the setting on which more recent work has been based. Indeed, a completion of this early work was made by Altland and Zirnbauer \cite{Altland}, leading to the identification of new symmetry classes. These ideas were generalized by Kitaev \cite{Kitaev} and led to  a ``periodic table'' of topological insulators and superconductors. In that work, Kitaev showed how the classification can be achieved in terms of \emph{Bott periodicity} and \emph{$K$--theory}. More recently an exhaustive and complete version of the classification was made by Ryu et al. \cite{Ryu}. They explore arbitrary dimensions making use again of classifying spaces given by the \emph{Cartan symmetric spaces} along  with Bott periodicity in a more strong way. For further details we refer to the reader to see \cite{Hastings2011, Loring2010,Fiorenza2016,Katsura2018}). This allows them to consider disordered systems and shows the explicit relation between gapped Hamiltonians and Anderson localization phenomena, a very important result for this kind of problem.\\
The first iconic example of a topological fermionic system is the \emph{quantum Hall effect}. The observed quantization of the conductivity was explained by Thouless et al. \cite{TKNN} and led to the recognition of the important role played by the \emph{Chern number}. The restrictions on the validity of this result where eventually overcome by Bellissard \cite{Bellisard} and collaborators, in what was to become one of the main examples of applications of \emph{noncommutative geometry} to physics. This was a big step to deal with more realistic models that consider disordered media. In this line of ideas there are more recent works, due to Carey et al \cite{Carey1,Prodan,Bourne}, where Bellissard's techniques are generalized to deal with a wider class of  systems. \\
For interacting systems rigorous proofs of quantization of conductivity were provided in \cite{GMP, bachRoeckFrass}. These studies rely on the study  of families of \emph{gapped Hamiltonians}, such that any two elements on these families can be continuously deformed into one another. The latter was demonstrated rigorously by Bachmann, Michalakis, Nachtergaele and Sims \cite{bachmann2012automorphic} by studying spectral flow of quantum spin systems under a ``quasi--adiabatic'' evolution. They proved that such related systems verify the same \emph{Lieb--Robinson bounds} and in its thermodynamic limit the spectral flow has a \emph{cocycle structure} for the automorphism in the algebra of observables. By using the dual space of the underlying algebras considered they also studied the \emph{ground states} associated.\\
From the point of view of physics, fundamental properties of such systems are deduced from the study of the set of ground states in the thermodynamic limit and zero temperature. Relevant examples include electronic conduction problems (e.g. quantum Hall effect), or the study of different phases of matter. Nevertheless, knowledge of ground s-tates for concrete models is a huge challenge in general. This is due to the fact that there is no  general procedure to find the full set of ground states for specific systems. As far as we know, there are very few mathematical physics results about the existence of ground states, in contrast to the theoretical point of view, see \cite{araki1985ground,cha2018complete}. Instead, one generally verifies the existence of the \emph{ground state energy} for specific physical systems\footnote{Ground state energy can be understood as the states associated to the lowest energy of a physical system. For example, Giuliani and Jauslin use rigorous renormalization methods to prove the existence of the ground state energy for the \emph{bilayer graphene} \cite{giuliani2016ground}.}. \par
In this paper we focus on the study of $\Z_{2}$--PI for non--interacting fermion systems. We specifically deal with \emph{unique} ground states associated to families of gapped Hamiltonians. Note that there is an alternative form of the $\Z_{2}$--PI \eqref{eq:top_index} in terms of \emph{orthogonal complex structures} \cite{Varilly,EK98}. There, the index appears naturally in the proof of the \emph{Shale--Stinespring Theorem} and is related to the \emph{parity} of the \emph{fermionic Fock} ground states. In \cite{CalReyes}, this approach to the $\Z_{2}$--PI was used to study ground states for finite \emph{Kitaev chains} with different boundary conditions. More recently, for infinite \emph{translationally invariant} fermionic chains, Bourne and Schulz-Baldes classify ground states using orthogonal complex structures \cite{BouSBal}. Furthermore, Matsui \cite{matsui20} uses \emph{split--property} of infinite chains and its connection with the $\Z_{2}$--PI. Theorem \ref{theorem:states_algebras2} is related, to a certain extent, to the mentioned results. Notice, however, that  we do not require translational invariance nor are we making any assumptions about the spatial dimension of the system, as  required, for example, in \cite{matsui20}. Although there are many similarities with the $\Z_{2}$--index introduced in  \cite{BouSBal}, they rely on \emph{different} notions of spectral flow (see the remark in the introduction to \cite{BouSBal}). It would be interesting to relate the different approaches to the $\Z_{2}$--index (particularly in view of their generalizations to interacting systems).   An important point on which we focus in the present paper is in an analysis of the uniformity of the $\Z_{2}$--PI with respect to the size of the systems. Physically, $\Z_{2}$--PI will distinguish if a pair of Hamiltonians are in the same \emph{phase of matter} or not. Moreover, the $\Z_{2}$--PI is closely related with the one proposed by Kitaev \cite{kitaev2001}, which was introduced to distinguish the parity of states in \emph{quantum wires}. However, as already mentioned, our results consider \emph{any} physical dimension, and has the potential to be studied in the \emph{interacting} case. In fact, in \cite{AMR2} we will report on results about the \emph{stability} of the $\Z_{2}$--PI for \emph{weakly} interacting fermions. Observe that the technical tools in that case differ from the current study and other technologies such as Lieb--Robinson bounds and Renormalization Group Methods will be required. For example, we use similar techniques as in \cite{ogata2020,ogata2021} where are studied indexes on one--dimensional interacting quantum spin systems and \cite{deRoecK_Salm} where is proven the stability of the \emph{spectral gap} for weakly interacting fermions.\\
The scope of the current paper is to consider bounded Hamiltonians on separable Hilbert spaces. However, most of the statements can be extended for unbounded Hamiltonians as is done, for example, in \cite{brupedraLR,nachtergaele2018quasi}. Specifically, in \cite{nachtergaele2018quasi}, the authors consider \emph{quantum lattice systems} such that their results hold for well--defined interactions as well as unbounded \emph{zero--range} interactions (also known as \emph{one--site} interactions or \emph{on--site} terms). In \cite{AMR2} we will take into account unbounded zero--range interactions such that we can invoke general Theorems for fermions on lattices as the ones presented in \cite{brupedraLR}. Moreover, we will tackle all (even unbounded cases, not being limited to semibounded) one--particle Hamiltonians, similar to \cite{universal}. Our setting will be closely related to the one of \cite[Sect. 5]{LD1}. There, we efficiently estimated \emph{determinant bounds} and \emph{summabilities of two--point correlation functions} associated with covariances, which depend only on one--particle Hamiltonians. These will be pivotal in the study of analytic properties of \emph{generating functions}, which in turn describe \emph{all} the statistical properties of weakly interacting fermion systems on the lattice. Thus, our results will consider any $\beta\in\R^{+}\cup\{\infty\}$, crystal lattice $\Z^{d}$ with $d\in\N$, sufficiently weak interactions, and physical systems \emph{not} necessarily translationally invariant.\par
To conclude, our main results are Theorems \ref{theorem:states_algebras2} and \ref{theorem:groun_states}, as well as the set of Lemmata \ref{lemma:impor_res}--\ref{lemma:unit_oper2}: In Theorem \ref{theorem:states_algebras2} we prove that for a differentiable path of gapped self--dual Hamiltonians, a spectral flow automorphism can be constructed, which will allow to show that the relative $\Z_{2}$--PI of any two Hamiltonians along the path must be equal to one. In Theorem \ref{theorem:groun_states} we show how, for a given differentiable  path of Hamiltonians connecting gapped Hamiltonians with opposite index, the gap must close somewhere along the path. We furthermore analyse the structure of a family of ground states which is induced by the gap closing.\\
From the mathematical point of view, the first part of Theorem \ref{theorem:states_algebras2} is reminiscent of the interacting case \cite{bachmann2012automorphic}, however, we additionally state the $\Z_{2}$--PI result, discriminating if a pair of Hamiltonians are equivalent or not on infinite self--dual $\CAR$ $\C$--algebras. In particular, we have in mind differentiable families of operators $\{H_{s}\}_{s\in[0,1]}$, e. g., given by the differentiable operator $H_{s}\doteq (1-s)H_{0}+sH_{1}$, for $s\in[0,1]$, with $H_{0},H_{1}\in\BL(\H)$ bounded operators with the same spectral gap and acting on a separable Hilbert space $\H$. On the other hand, Theorem \ref{theorem:groun_states} deals with subsets of the \emph{ground states} set. For instance, open spectral gap ground states are considered. As a particular case of the general Theorem \ref{theorem:groun_states}, we prove that in the weak$^{*}$--topology, \emph{paths} connecting states in different topological components implies the existence of a Hamiltonian having $0$ as an eigenvalue.  \smallskip \\
The paper is organized as follows:
\begin{itemize}
 \item Section \ref{Section Gen unc as Grassmann int} presents the mathematical framework of $\CAR$ $\C$--algebras. We introduce self--dual $\CAR$ $\C$--algebras, which were introduced long ago by Araki in his elegant study of \emph{non}--interacting but \emph{non}--gauge invariant fermion systems. We recall pivotal properties of general $\CAR$ $\C$--algebras.
 \item In Section \ref{sec: main results} we state the main Theorems, as well as some relevant definitions concerning the $\Z_{2}$--PI and comment on the weak$^{*}$--topology of the set of states. In particular we discuss the conditions for a system to have pure or mixed states.
 \item Section \ref{sec: tech_proofs} is devoted to all technical proofs. We prove the existence of a \emph{spectral flow automorphism} for self--dual Hilbert spaces, for families of differentiable Hamiltonians. Then, the existence of strong limits for the dynamics, the spectral flow automorphism and the weak$^{*}$--convergence of ground states are proven. Well--known Combes--Thomas estimates are invoked for families of gapped Hamiltonians, which will permit to analyze two--point correlation functions such that we obtain the \emph{trace class} properties for relevant unitary operators.
 \item We finally include Appendix \ref{appendix: GAM}, providing a general framework of graphs with special attention to disordered models. Appendix \ref{appendix: fock} regards on basic statements of the Fock representation of $\CAR$. Appendix \ref{appendix: CAR} shows the explicit relation between the canonical $\CAR$ and self--dual $\CAR$ $\C$--algebras, and how their dynamics are related.
\end{itemize}
\begin{notation}\label{remark constant}
A norm on the generic vector space $\X$ is denoted by $\Vert \cdot \Vert _{\X}$ and the identity map of $\X$ by $\1_{\X}$. The space of all bounded linear operators on $(\X,\Vert \cdot \Vert _{\X}\mathcal{)}$ is denoted by $\BL(\X)$. The unit element of any algebra $\X$ is always denoted by $\mathfrak{1}$, provided it exists of course. The scalar product of any Hilbert space $\X$ is denoted by $\langle \cdot,\cdot\rangle_{\X}$ and $\tr_{\X}$ represents the usual trace on $\BL(\X)$. \end{notation}
\section{Mathematical Framework and Physical Setting}\label{Section Gen unc as Grassmann int}
We introduce the mathematical framework based on Araki's self--dual formalism \cite{A68,A70}. Our setting considers disorder effects, which come as is usual in physics, i.e., from impurities, crystal lattice defects, etc. Thus, disorder can modeled by (a) a random external potential, like in the celebrated Anderson model, (b) a random Laplacian, i.e., a self--adjoint operator defined by a next--nearest neighbor hopping term with random complex--valued amplitudes. In particular, random vector potentials can also be implemented.
\subsection{Self--dual \eqt{$\CAR$} Algebra}\label{Self--dual CAR Algebras}
If not otherwise stated, $\H$ always stands for a  (complex, separable) Hilbert space. If $\H$ is finite--dimensional, we will assume it is even--dimensional, i.e., $\dim\H\in2\N$. Let $\a\colon\H\to\H$ be a \emph{conjugation} or \emph{antiunitary involution} on $\H$, i.e., an antilinear operator such that $\a^{2}=\1_{\H}$ and\footnote{We will assume that the inner product $\inner{\cdot,\cdot}_{\H}\colon\H\times\H\to\CP$ associated to some Hilbert space $\H$ is a sesquilinear form on $\H$ such that is antilinear in its first component while is linear in the second one.}
\begin{equation*}
\inner{\a\varphi_{1},\a\varphi_{2}}_{\H}=\inner{\varphi_{2},\varphi_{1}}_{\H},\qquad \varphi_{1},\varphi_{2}\in\H.
\end{equation*}
The space $\H$ endowed with the involution $\a$ is named a \emph{self--dual Hilbert space}, $(\H,\a)$, and yields \emph{self--dual $\CAR$ algebra}:
\begin{definition}[Self--dual $\CAR$ algebra]
\label{def Self--dual CAR Algebras}
A self--dual $\CAR$ algebra $\sCAR\equiv (\sCAR,+,\cdot ,*)$ is a $\C$--algebra generated by a unit $\mathfrak{1}$ and a family $\{\B(\varphi )\}_{\varphi \in \H}$ of elements satisfying Conditions \ref{enum.scar1}.--\ref{enum.scar3}.:
\begin{enumerate}
 \item\label{enum.scar1} The map $\varphi \mapsto \B\left(\varphi\right) ^{*} $ is (complex) linear.
 \item\label{enum.scar2} $\B(\varphi)^{*}=\B(\a(\varphi))$ for any $\varphi \in \H$.
 \item\label{enum.scar3} The family $\{\B(\varphi )\}_{\varphi\in\H}$ satisfies the $\CAR$: For any $\varphi_{1},\varphi _{2}\in\H$,
\begin{equation}
\B(\varphi _{1})\B(\varphi_{2})^{*}+\B(\varphi_{2})^{* }\B(\varphi_{1})=\left\langle \varphi _{1},\varphi_{2}\right\rangle_{\H}\,\mathfrak{1}.  \label{CAR Grassmann III}
\end{equation}
\end{enumerate}
\end{definition}
For a historic overview on self--dual $\CAR$ algebras and some of their basic properties see \cite{A68,A70,A87,A88,EK98}. Note that by the $\CAR$ \eqref{CAR Grassmann III}, the antilinear map $\varphi \mapsto\B\left( \varphi \right) $ is necessarily injective and contractive. Therefore, $\H$ can be embedded in $\sCAR$. \par
Conditions \ref{enum.scar1}.--\ref{enum.scar3}. of Definition \ref{def Self--dual CAR Algebras} only define self--dual $\CAR$ algebras up to Bogoliubov $^{*}$-- automorphisms\footnote{An analogous result for $\CAR$ algebra is, for instance, given by \cite[Theorem 5.2.5]{BratteliRobinson}.} (see \eqref{Bogoliubov automorphism}). In \cite{LD1}, an explicit construction of $^{*}$--isomorphic self--dual $\CAR$ algebras from $\H$ and $\a$ is presented. This is done via \emph{basis projections} \cite[Definition 3.5]{A68}, which highlight the relationship between $\CAR$ algebras and their self--dual counterparts.
\begin{definition}[Basis projections]\label{def basis projection}
A basis projection associated with $(\H,\a)$ is an orthogonal projection $P\in\BL(\H)$ satisfying $\a P\a=P^{\bot}\equiv\1_{\H}-P$. We denote by $\h_{P}$ the range $\mathrm{ran}(P)$ of the basis projection $P$. The set of all basis projections associated with $(\H,\a)$ will be denoted by $\fp$.
\end{definition}
For simplicity, in the rest of this section, we will assume that $\H$ is finite--dimensional with even size: $\dim\H\in 2\N$. For any $P\in\fp$ a few remarks are in order:\\
$\h_{P}$ must satisfy the conditions
\begin{equation}
\a(\h_{P})=\h_{P}^{\bot }\qquad\text{and}\qquad\a(\h_{P}^{\bot })=\h_{P}.
\label{frakA and perp}
\end{equation}
Then, by \cite[Lemma 3.3]{A68}, an explicit $P\in\fp$ can always be constructed. Moreover, $\varphi\mapsto (\a\varphi)^{*}$ is a unitary map from $\h_{P}^{\bot }$ to the dual space $\h_{P}^{* }$. In this case we can identify $\H$ with
\begin{equation}
\H\equiv\h_{P}\oplus\h_{P}^{*}\label{definition H bar}
\end{equation}
and
\begin{equation}
\B\left(\varphi\right)\equiv\B_{P}(\varphi)\doteq\B\left(P\varphi\right) +\B\left(\a P^{\bot}\varphi\right)^{*}.  \label{map iodiote}
\end{equation}
Therefore, there is a natural isomorphism of $\C$--algebras from $\sCAR$ to the $\CAR$ algebra $\CAR(\h_{P})$ generated by the unit $\mathfrak{1}$ and $\{\B_{P}(\varphi )\}_{\varphi \in \h_{P}}$. In other words, a basis projection $P$ can be used to \emph{fix} so--called \emph{annihilation} and \emph{creation} operators. For each basis projection $P$ associated with $(\H,\a)$, by \eqref{definition H bar}, $\h_{P}$ can be seen as a one--particle Hilbert space.\\
Self--dual $\CAR$ algebras naturally arise in the diagonalization of quadratic fermionic Hamiltonians (Definition \ref{def trace state copy(1)}), via Bogoliubov transformations defined as follows \cite{A68,A70}:\\
For any unitary operator $U\in\BL(\H)$ such that $U\a=\a U$, the family of elements $\B(U\varphi)_{\varphi\in\H}$ satisfies Conditions (a)--(c) of Definition \ref{def Self--dual CAR Algebras} and, together with the unit $\mathfrak{1}$, generates $\sCAR$. Like in \cite[Section 2]{A70}, such a unitary operator $U\in \BL(\H)$ commuting with the antiunitary map $\a$ is named a \emph{Bogoliubov transformation}, and the unique $^{*}$--automorphism $\mathbf{\chi}_{U}$ such that
\begin{equation}
\mathbf{\chi}_{U}\left( \B(\varphi )\right) =\B(U\varphi ),\qquad \varphi \in \H,
\label{Bogoliubov  automorphism}
\end{equation}
is called in this case a \emph{Bogoliubov} $^{*}$\emph{--automorphism}. Note that a Bogoliubov transformation $U\in\BL(\H)$ always satisfies
\begin{equation}
\det\left(U\right)=\det\left(\a U\a\right)=\overline{\det\left(U\right)}=\pm1
\label{orientation}
\end{equation}
If $\det\left(U\right)=1$, we say that $U$ is in the \emph{positive} connected set $\mathfrak{U}_{+}$. Otherwise $U$ is said to be in the \emph{negative} connected set $\mathfrak{U}_{-}$. $\mathbf{\chi}_{U}\left( \B(\varphi )\right)$ is said to be \emph{even} (respectively \emph{odd}) if and only if $U\in\mathfrak{U}_{+}$ (respectively $U\in\mathfrak{U}_{-}$). \\
Clearly, if $P\in\fp$, see Definition \ref{def basis projection}, and $U\in \BL(\H)$ is a Bogoliubov transformation, then $P_{U}\doteq U^{*}PU$ is another basis projection. Conversely, for any pair $P_{1},P_{2}\in\fp$ there is a (generally not unique) Bogoliubov transformation $U$ such that $P_{2}=U^{*}P_{1}U$. See \cite[Lemma 3.6]{A68}. In particular, Bogoliubov transformations map one--particle Hilbert spaces onto one another.\par
Considering the Bogoliubov $^{*}$--automorphism \eqref{Bogoliubov  automorphism} with $U=-\1_{\H}$, an element $A\in\sCAR$, satisfying
\begin{equation}\label{eq:even odd}
\mathbf{\chi}_{-\1_{\H}}(A)=
\begin{cases}
   \quad A&\text{is called \emph{even}},\\
   -A&\text{is called \emph{odd}},
\end{cases}
\end{equation}
Note that the subspace $\sCAR^{+}$ of even elements is a sub--$\C$--algebra of $\sCAR$.\\
It is well--known that in quantum mechanics the even elements are the ones suitable for the description of fermion systems. For example, self--adjoint (even) elements of the $\CAR$ algebra which are quadratic in the creation and annihilation operators are used, for instance, in the Bogoliubov approximation of the
celebrated (reduced) BCS model. In the context of self--dual $\CAR$ algebra, those elements are called \emph{bilinear Hamiltonians} and are self--adjoint bilinear elements:
\begin{definition}[Bilinear elements of self--dual CAR algebra]
\label{def trace state copy(1)}
Given an orthonormal basis $\{\psi_{i}\}_{i\in I}$ of $\H$, we define the bilinear element associated with $H\in \BL(\H)$ to be
\begin{align*}
\Bil{H} \doteq \sum\limits_{i,j\in I}\left\langle \psi _{i},H\psi _{j}\right\rangle_{\H}\B\left(\psi_{j}\right)\B\left(\psi_{i}\right)^{*}.
\end{align*}
\end{definition}
Note that $\Bil{H}$ \emph{does not depend} on the particular choice of the orthonormal basis, but does depend on the choice of generators $\{\B(\varphi )\}_{\varphi \in \H}$ of the self--dual $\CAR$ algebra $\sCAR$, and by \eqref{CAR Grassmann III}, bilinear elements of $\sCAR$ have adjoints equal to
\begin{equation}
\Bil{H}^{*}=\langle \B,H^{*}\B\rangle,\qquad H\in \BL(\H).\label{selfadjoint bilinear}
\end{equation}
\emph{Bilinear Hamiltonians} are then defined as bilinear elements associated with \emph{self--adjoint} operators $H=H^{*}\in \BL(\H)$. They include all second quantizations of one--particle Hamiltonians, but also models that are \emph{not gauge invariant}. Important models in condensed matter physics, like in the BCS theory of superconductivity, are bilinear Hamiltonians that are \emph{not} gauge invariant.\\
Without loss of generality (w.l.o.g.), our analysis of bilinear elements can be restricted to operators $H\in \BL(\H)$ satisfying $H^{* }=-\a H\a$, which, in particular, have zero trace, i.e., $\tr_{\H}\left( H\right)=0$\footnote{Recall that for any separable Hilbert space $\H$, $A\in\BL(\H)$ and any orthonormal basis $\{\psi_{i}\}_{i\in I}$ of $\H$ the trace of $A$, $\tr_{\H}(A)\doteq\sum\limits_{i\in I}\inner{\psi_{i},A\psi_{i}}_{\H}$, does not depend of the choice of the orthonormal basis.}. We call such operators \emph{self--dual operators}:
\begin{definition}[Self--dual operators]\label{def one particle hamiltinian}
A self--dual operator on $(\H,\a)$ is an operator $H\in\BL(\H)$ satisfying the equality $H^{*}=-\a H\a$. If, additionally, $H$ is self--adjoint, then we say that it is a self--dual Hamiltonian on $(\H,\a)$.
\end{definition}
We say that the basis projection $P$ (Definition \ref{def basis projection}) (block--) ``diagonalizes'' the self--dual operator $H\in \BL(\H)$ whenever
\begin{equation}
H=\frac{1}{2}\left( PH_{P}P-P^{\bot }\a H_{P}^{* }\a P^{\bot }\right),\qquad \text{with}\qquad H_{P}\doteq 2PHP\in \BL(\h_{P}).  \label{kappabisbiskappabisbis}
\end{equation}
In this situation, we also say that the basis projection $P$\ diagonalizes $\Bil{H} $, similarly to \cite[Definition 5.1]{A68}.\\
By the spectral theorem, for any self-dual \emph{Hamiltonian} $H$ on $(\H,\a)$, there is always a basis projection $P$ diagonalizing $H$. In quantum physics, as discussed in Section \ref{Self--dual CAR Algebras}, $\h_{P}$ is in this case the \emph{one--particle Hilbert space} and $H_{P}$ the \emph{one--particle Hamiltonian}.
\subsection{Quasi--Free Dynamics}\label{subsec: dyna_states}
Bilinear Hamiltonians are used to define so-called \emph{quasi--free} dynamics: For any $H=H^{*}\in\BL(\H)$, we define the continuous group $\{\tau _{t}\}_{t\in {\R}}$ of $^{*}$--automorphisms of $\sCAR$ by
\begin{equation}\label{eq:autoSCAR}
\tau _{t}(A)\doteq \e^{-\ii t\Bil{H}}A\e^{\ii t\Bil{H}}\ ,\qquad A\in \mathrm{%
sCAR}(\H,\a),\ t\in\R.
\end{equation}
Provided $H$ is a self--dual Hamiltonian on $(\H,\a)$ (Definition \ref{def one particle hamiltinian}), this group is a quasi--free dynamics, that is, a strongly continuous group of Bogoliubov $^{*}$--automorphisms, as defined in Equation \eqref{Bogoliubov automorphism}. Straightforward computations using Definitions \ref{def Self--dual CAR Algebras} and \ref{def trace state copy(1)}, together with the properties of the antiunitary involution $\a$, lead to show that
\begin{equation}\label{eq:dyn_bogou}
\exp \left( -\frac{z}{2}\Bil{H}\right)\B\left(\varphi\right) ^{* }\exp \left( \frac{z}{2}\langle\B,H\B\rangle \right)=\B\left( \e^{zH}\varphi \right)^{* },
\end{equation}
even for any self--dual operator $H$ on $(\H,\a)$, all $z\in \CP$ and $\varphi \in \H$.\\
Moreover, for $\{\tau_{t}\}_{t\in\R}$, we define the linear subspace
\begin{equation}\label{eq:domain_sCAR}
\mathcal{D}(\delta)\doteq\{A\in\sCAR\colon t\mapsto\tau_{t}(A)\text{ is differentiable at }t=0\}\subset\sCAR
\end{equation}
and the linear operator (unique, generally unbounded) $\delta\colon\mathcal{D}\to\sCAR$ by
\begin{equation}\label{eq:gener_sCAR}
\delta(A)\doteq\frac{\d\tau_{t}(A)}{\d t}\Big|_{t=0}.
\end{equation}
The operator $\delta$ is called the generator of $\tau$ and $\mathcal{D}(\delta)$ is the (dense) domain of definition of $\delta$. Here we will assume that $\delta$ is a symmetric unbounded derivation, i.e., the domain $\mathcal{D}(\delta)$ of $\delta$ is a dense $^{*}$--subalgebra of $\A$ and, for all $A,B\in\mathcal{D}(\delta)$,
\begin{align*}
\delta(A)^{*}=\delta(A^{*}),\qquad\delta(AB)=\delta(A)B+A\delta(B).
\end{align*}
Note that the set of all symmetric derivations on $\mathcal{D}(\delta)$ can be endowed with a real vector space structure. In fact, for any symmetric derivations $\delta _{1}$ and $\delta _{2}$ and all real numbers $\alpha _{1},\alpha _{2}$, the expression
$$
\left(\alpha_{1}\delta_{1}+\alpha_{2}\delta_{2}\right)\left(A\right)\doteq\alpha_{1}\delta_{1}\left(A\right)+\alpha_{2}\delta_{2}\left(A\right),\qquad A\in\mathcal{D}(\delta),
$$
gives rise to another symmetric derivation $\alpha_{1}\delta_{1}+\alpha_{2}\delta _{2}$ on $\mathcal{D}(\delta)$.
\subsection{States}\label{sec:states}
A linear functional $\omega\in\sCAR^{*}$ is a ``state'' if it is positive and normalized, i.e., if for all $A\in\sCAR,\omega(A^{*}A)\geq0$ and $\omega(\mathfrak{1})=1$. In the sequel, $\states\subset\sCAR^{*}$ will denote the set of all states on $\sCAR$. Note that any $\omega\in\states$ is \emph{Hermitian}, i.e., for all $A\in\sCAR,\omega(A^{*})=\overline{\omega(A)}$. $\omega\in\states$ is said to be ``faithful'' if $A=0$ whenever $A\geq 0$ and $\omega(A)=0$. Since $\sCAR$ is a unital $\C$--algebra, $\states$ is a \emph{weak$^{*}$--compact} convex set, such that its \emph{extremal} points coincide with the \emph{pure} states \cite[Theorem 2.3.15]{BratteliRobinsonI}. The latter, combined with the fact that $\sCAR$ is separable allows to claim that the set of states $\states$ is metrizable in the weak$^{*}$--topology \cite[Theorem 3.16]{rudin}. Note that the existence of extremal points is a consequence of the \emph{Krein--Milman Theorem}. More specifically, if $\mathtt{E}(\states)$ denotes the set of extremal points of $\states$,
$$
\states=\mathrm{cch}\left(\mathtt{E}\left(\states\right)\right),
$$
where, for $\X$ a Topological Vector Space and $A\subset\X$, $\mathrm{cch}(A)$ refers to the \emph{closed convex hull} of $A$.\\
A state $\omega\in \states$ is said to be a \emph{pure state} if $\omega\in\mathtt{E}(\states)$. By definition, a \emph{mixed state} is a state that is not pure. Notice that a pure state can be characterized as a state which cannot be written as a convex linear combination of two different states. Notice in particular that if  $\omega\in\states$ is a state of the form
\begin{equation}\label{eq:convex_state}
\omega=\sum_{j=1}^{m}\lambda_{j}\omega_{j},
\end{equation}
where $\{\omega_{j}\}_{j=1}^{m}\in\mathtt{E}(\states)$, $m\in\N$, and $\lambda_{j}\in[0,1]$ for $j\in\{1,\ldots,m\}$, with $\sum\limits_{j=1}^{m}\lambda_{j}=1$. Then, if $\omega$ is pure, it necessarily follows that $\omega=\omega_{1}=\cdots=\omega_{m}$.
As is usual, for the state $\omega\in\states$ on $\sCAR$, $(\H_{\omega},\pi_{\omega},\Omega_{\omega})$ denotes its associated \emph{cyclic} representation: $H_{\omega}$ is the Hilbert space associated to $\omega$, and is given by the closure of (the linear span) of the set $\left\{ \pi _{\omega}(A)\Omega _{\omega}\colon A\in \sCAR\right\}$\footnote{For the Topological Vector Space $\X$, $\overline{\X}$ denotes its closure.},
$$
\H_{\omega}=\overline{\pi_{\omega}\left(\sCAR\right)\Omega_{\omega}},
$$
i.e., $\H_{\omega}$ is a Hilbert space with scalar product $\inner{\cdot,\cdot}_{\H_{\omega}}$, $\pi _{\omega}$ a representation from $\sCAR$ into $\BL(\H_{\omega})$, the set of bounded operators acting on $\H_{\omega}$, and $\Omega _{\omega}\in\H_{\omega}$ is a \emph{unit} cyclic vector with respect to $\pi _{\omega}(\sCAR)$. More specifically, for all $A\in\sCAR$ we write
\begin{equation}\label{eq:GNS_state}
\omega(A)=\inner{\Omega_{\omega},\pi_{\omega}(A)\Omega_{\omega}}_{\H_{\omega}}.
\end{equation}
$(\H_{\omega},\pi_{\omega},\Omega_{\omega})$ is the so--called \emph{GNS construction}, which is unique up to unitary equivalence.\\
If the state $\omega\in\states$ is mixed, see Expression \eqref{eq:convex_state}, its associated representation $(\H_{\omega},\pi_{\omega})$ is reducible, that is, it can be decomposed as a direct sum $\pi_{\omega}=\bigoplus\limits_{j\in J}
\pi_{\omega_{j}}$ on $\H_{\omega}=\bigoplus\limits_{j\in J}\H_{\omega_{j}}$. Here, $\{\H_{j}\}_{j\in J}$ is a countable family of \emph{orthogonal} Hilbert spaces, by meaning that for two different Hilbert spaces $\H_{1}$ and $\H_{2}$ of $\{\H_{j}\}_{j\in J}$, $\inner{\varphi_{1},\varphi_{2}}_{\H}=0$ for all $\varphi_{1}\in\H_{1}$ and all $\varphi_{2}\in\H_{2}$. The set $\{\pi_{\omega_{j}}\}_{j\in J}$ are representations of $\sCAR$ on proper subspaces of $\H_{\omega}$. In particular if $\omega$ is pure, its representation $(\H_{\omega},\pi_{\omega})$ is irreducible and $\omega$ is an extremal point $\mathtt{E}(\states)$ of the set of states on $\sCAR$.\\
States $\omega\in\states$ are said to be \emph{quasi--free} when, for all $N\in\N_{0}$ and $\varphi_{0},\ldots ,\varphi _{2N}\in \H$,
\begin{equation}
\omega\left(\B\left(\varphi_{0}\right)\cdots\B\left(\varphi_{2N}\right)\right)=0,  \label{ass O0-00}
\end{equation}
while, for all $N\in\N$ and $\varphi_{1},\ldots,\varphi_{2N}\in\H$,
\begin{equation}
\omega \left(\B\left(\varphi _{1}\right) \cdots \B\left(\varphi _{2N}\right) \right)=\mathrm{Pf}\left[\omega \left( \mathbb{O}_{k,l}\left( \B(\varphi _{k}),\B(\varphi_{l})\right)\right)\right]_{k,l=1}^{2N},  \label{ass O0-00bis}
\end{equation}
where
\begin{equation*}
\mathbb{O}_{k,l}\left(A_{1},A_{2}\right)\doteq
\left\{
\begin{array}{ccc}
A_{1}A_{2} & \text{for} & k<l,\\
-A_{2}A_{1} & \text{for} & k>l,\\
0 & \text{for} & k=l.
\end{array}\right.
\end{equation*}
In Equation \eqref{ass O0-00bis}, $\mathrm{Pf}$ is the usual Pfaffian defined by
\begin{equation}
\mathrm{Pf}\left[M_{k,l}\right]_{k,l=1}^{2N}\doteq \frac{1}{2^{N}N!}\sum_{\pi \in \mathcal{S}_{2N}}\left(-1\right)^{\pi}\prod\limits_{j=1}^{N}M_{\pi \left( 2j-1\right) ,\pi\left(2j\right)}\label{Pfaffian}
\end{equation}
for any $2N\times 2N$ skew--symmetric matrix $M\in \mathrm{Mat}\left(2N,\CP\right)$. Note that \eqref{ass O0-00bis} is equivalent to the definition given either in \cite[Definition 3.1]{A70} or in \cite[Equation (6.6.9)]{EK98}.\\
Quasi--free states are therefore particular states that are uniquely defined by two--point correlation functions, via \eqref{ass O0-00}--\eqref{ass O0-00bis}. In fact, a quasi--free state $\omega$ is uniquely defined by its so--called \emph{symbol}, that is, a positive operator $S_{\omega}\in\BL(\H)$ such that
\begin{equation}
0\leq S_{\omega }\leq \1_{\H}\qquad\text{and}\qquad S_{\omega}+\a S_{\omega}\a=\1_{\H},
\label{symbol}
\end{equation}
through the conditions
\begin{equation}
\left\langle \varphi _{1},S_{\omega}\varphi_{2}\right\rangle_{\H}=\omega\left(\B(\varphi _{1})\B(\varphi_{2})^{*}\right),\qquad\varphi_{1},\varphi_{2}\in\H.\label{symbolbis}
\end{equation}
Conversely, any self--adjoint operator satisfying \eqref{symbol} uniquely defines a quasi--free state through Equation \eqref{symbolbis}. In physics, $S_{\omega} $ is related to the \emph{one--particle density matrix} of the system. Note that any basis projection associated with $(\H,\a)$ can be seen as a symbol of a quasi--free state on $\sCAR$. Such state is pure and called a \emph{Fock state} \cite[Lemma 4.3]{A70}. Araki shows in \cite[Lemmata 4.5--4.6]{A70} that any quasi--free state can be seen as the restriction of a quasi--free state on $\mathrm{sCAR}(\H\oplus\H,\a\oplus (-\a))$, the symbol of which is a basis projection associated with $(\H\oplus \H,\a\oplus (-\a))$. This procedure is called \emph{purification} of the quasi--free state.\\
Quasi--free states obviously depend on the choice of generators of the self--dual CAR algebra. Another example of a quasi--free state is provided by the tracial state:
\begin{definition}[Tracial state]\label{def trace state}
The tracial state $\tr_{\A}\in\states$ is the quasi--free state with symbol $S_{\tr}\doteq \frac{1}{2}\1_{\H}$.
\end{definition}
The tracial state can be used to highlight the relationship between quasi--free states and bilinear Hamiltonians. In fact, one can show, c.f. \cite{LD1}, that for any $\beta \in (0,\infty)$ and any self--dual Hamiltonian $H$ on $(\H,\a)$ the positive operator $(1+\e^{-\beta H})^{-1}$ satisfies Condition \eqref{symbol} and is the symbol of a quasi--free state $\omega_{H}^{(\beta)}$ satisfying
\begin{equation}
\omega_{H}^{(\beta)}(A)=\frac{\tr_{\A}\left(A\exp\left(\frac{\beta }{2}\langle\B,H\B\rangle\right) \right)}{\tr_{\A}\left( \exp\left(\frac{\beta}{2}\Bil{H}\right)\right)},\qquad A\in \sCAR.\label{Gibbs states}
\end{equation}
The state $\omega_{H}^{(\beta)}\in\states$ is named the $(\tau_{t},\beta)$--\emph{Gibbs state}, thermal equilibrium state, or \emph{KMS}--state, associated with the self--dual (one--particle) Hamiltonian $H$ on $(\H,\a)$ at fixed $\beta\in(0,\infty)$. As is usual, we call to the parameter $\beta\in(0,\infty)$ the \emph{inverse (non--negative) temperature} of a physical system. Note that, given $H\in\BL(\H)$, we also can define two particular quasi--free states $\omega_{H}^{(0)}$ and $\omega_{H}^{(\infty)}$, which satisfy \eqref{Gibbs states} for the convergent sequence $\left\{\beta_{n}\right\}_{n\in\N}\subset\R_{0}\cup\{\infty\}$ to a $\beta\subset\R_{0}\cup\{\infty\}$. The former case is closely related with the tracial state in Definition \ref{def trace state}, and corresponds to the infinite temperature. Namely, the state at $\beta=\lim\limits_{n\to\infty}\beta_{n}=0$ is known as trace state or chaotic state. This particular name comes from the fact that physically it corresponds to the state of maximal entropy which occurs at infinite temperature. Its uniqueness is a well--known property. On the other hand, states at $\beta=\lim\limits_{n\to\infty}\beta_{n}=\infty$ are also thermal equilibrium states. More generally, these are defined by:
\begin{definition}[Ground state]\label{def: ground state}
Let $\omega\in\states$ be a state on $\sCAR$ and let $H\in\BL(\H)$ be a self--dual Hamiltonian on $(\H,\a)$. We say that $\omega\equiv\omega_{H}^{(\infty)}$ is a \emph{ground state} if it satisfies
$$
\ii\omega(A^{*}\delta(A))\geq0,
$$
for all $A\in\mathcal{D}(\delta)$. Here $\delta$ is the generator with domain $\mathcal{D}(\delta)$, of the continuous group $\{\tau_{t}\}_{t\in\R}$ of $^{*}$--automorphisms of $\sCAR$ given by \eqref{eq:autoSCAR}.
\end{definition}
From now on, we will denote by $\states^{(\beta)}\in\states$ the set of all KMS states at inverse temperature $\beta\in\R_{0}^{+}\cup\{\infty\}$ associated to the self--dual Hamiltonian $H$ on $(\H,\a)$. A few of remarks regarding $\states^{(\beta)}$ are discussed:\\
To lighten the notation, in the sequel when we refer to the KMS state $\omega_{H}^{(\beta)}$ we will omit any mention of the dependence on $H$, i.e., $\omega_{H}^{(\beta)}\equiv\omega^{(\beta)}$. For $\beta\in\R^{+}\cup\{\infty\}$, $\omega^{(\beta)}\in\states^{(\beta)}$ is $\tau$ invariant or stationary, i.e., $\omega^{(\beta)}\circ\tau=\omega^{(\beta)}$. See \cite[Propositions 5.3.3 and 5.3.19]{BratteliRobinson}. In contrast, the tracial case $\beta=0$ not necessarily is. Then, for $\beta\in\R^{+}\cup\{\infty\}$, $\omega\equiv\omega^{(\beta)}$, there is a strongly continuous one--parameter unitary group $\left(\e^{\ii t\mathcal{L}_{\omega}}\right)_{t\in\R}$ with generator $\mathcal{L}_{\omega}=\mathcal{L}_{\omega}^{*}$ satisfying $\e^{\ii t\mathcal{L}_{\omega}}\Omega_{\omega }=\Omega _{\omega}$ such that for any $t\in\R$
\begin{equation*}
\pi _{\omega }\left( \tau _{t}\left( A\right) \right)=\e^{-\ii t\mathcal{L}_{\omega}}\pi _{\omega}\left( A\right) \e^{\ii t\mathcal{L}_{\omega}}\qquad\text{and}\qquad\e^{\ii t\mathcal{L}_{\omega}}\in\pi_{\omega}(\sCAR)''.
\end{equation*}
If any $A\in\mathcal{D}(\delta)\subseteq\sCAR$,
\begin{equation*}
\pi_{\omega}(A)\Omega _{\omega }\in \mathcal{D}(\mathcal{L})\qquad \text{and}\qquad \mathcal{L}\left(\pi_{\omega }(A)\Omega _{\omega }\right)=\pi_{\omega}(\delta \left(A\right))\Omega_{\omega}.
\end{equation*}
If $\omega$ is a ground state, then the generator satisfies $\mathcal{L}_{\omega}\geq0$.\\
For $\beta\in\R^{+}$, the set $\states^{(\beta)}\in\states$ forms a weak$^{*}$--compact convex set that also is a simplex\footnote{This is true because one can show that the set of KMS $\states^{(\beta)}\subset\states$ forms a base of the \emph{cone}  which is also a \emph{lattice} \cite[Chapter 4]{BratteliRobinsonI}.}, while the set of ground states or KMS states at inverse temperature $\infty$, $\states^{(\infty)}\subset\states$, forms a face $\mathcal{F}$, i.e., a subset of a compact convex set $\mathcal{K}$ such that if there are finite linear combinations
$$
\omega=\sum_{j=1}^{n}\lambda_{i}\omega_{i}\qquad\text{with}\qquad\sum_{j=1}^{n}\lambda_{j}=1
$$
of elements $\{\omega_{j}\}_{j=1}^{n}\in\mathcal{K}$ and $\omega\in\mathcal{F}$, then $ \{\omega_{j}\}_{j=1}^{n}\in\mathcal{F}$.\\
Let $A$ be a self--dual operator on $(\H,\a)$, such that $E_{\Sigma}(A)\doteq\chi_{\Sigma}(A)$ defines the \emph{spectral projection} of $A$ on the Borel set $\Sigma\subset\R$. Here, $\chi_{\Sigma}\colon\Sigma\to\{0,1\}$ is the so--called characteristic function on $\Sigma\subset\R$, with $\chi_{\Sigma}^{2}=\chi_{\Sigma}$. For $H$, a self--adjoint Hamiltonian on $(\H,\a)$, i.e., $H=-\a H\a$, we denote by $E_{0},\,E_{-}$ and $E_{+}$, the restrictions of the spectral projections of $H$ on $\{0\}$, the negative real numbers $\R^{-}$ and the positive real numbers $\R^{+}$, respectively. Using functional calculus we note that
\begin{equation*}
%\label{eq:Ham_spec}
H=\int_{\spec(H)}\lambda\d E_{\lambda}=\int_{\R}\lambda\d E_{\lambda},
\end{equation*}
where $\spec(H)$ denotes the spectrum of $H$. Thus, one verifies that
\begin{equation}\label{eq:sum_basis_proj}
\a E_{\lambda}\a=E_{-\lambda}\qquad\text{for all}\qquad \lambda\in\R\qquad\text{and}\qquad E_{0}+E_{-}+E_{+}=\1_{\H}.
\end{equation}
In particular, we have $\a E_{0}\a=E_{0}$. However, we strongly will assume throughout this paper that $E_{0}=0$ so that the ground state is \emph{unique}. For details see \cite{araki1985ground}[Theorems 3 and 4]. By \eqref{eq:sum_basis_proj}, both $E_{+}$ and $E_{-}$ are basis projections in $\fp$: $\a E_{\pm}\a=\1_{\H}-E_{\pm}$, i.e., ground states can be uniquely characterized by their spectral projections $E_{\pm}$.
In particular, the symbol $S_{\omega}$ in \eqref{symbolbis} corresponds to the \emph{spectral projection} $E_{+}$ on the positive real numbers, associated to the self--dual Hamiltonian $H$ on $(H,\a)$ in such a way that ground states are uniquely determined by the two--point correlation function defined by:
\begin{equation}\label{eq:symbolbis2}
\omega_{E_{+}}\left(\B(\varphi _{1})\B(\varphi_{2})^{*}\right)=\left\langle \varphi _{1},E_{+}\varphi_{2}\right\rangle_{\H},\qquad\varphi_{1},\varphi_{2}\in\H.
\end{equation}
Thus, for a quasi--free system associated to some self--dual Hamiltonian $H$, the set of all ground states $\states_{H}^{(\infty)}\equiv\states^{(\infty)}$, is studied via (positive) spectral projections of $H$. Additionally, straightforward calculations show the uniqueness of ground states, even under \emph{small} perturbations. See \cite[Chapter 5]{BratteliRobinson} and \cite{hastings2019stability} for recent results on the stability of free fermion systems. More generally, for a unital $\C$--algebra and $\beta\in\R^{+}$, the quasi--free state is unique, and this also holds for \emph{gapped systems}, as defined below, for $\beta\to\infty$. We now define:
\begin{definition}[Quasi--free ground states]\label{def:qfgs}
The state $\omega\in\sCAR^{*}$ satisfying \eqref{symbol}, \eqref{symbolbis} and \eqref{eq:symbolbis2} it will be called \emph{quasi--free ground state}. The set of all quasi--free ground states it will denoted by $\q\states^{(\infty)}\subset\states^{(\infty)}$.
\end{definition}
\subsection{Gapped Systems}
We consider the (possibly unbounded) self--adjoint operator $\mathtt{h}=\mathtt{h}^{*}
\in\LO(\mathtt{H})$ (the linear operators on $\mathtt{H}$), for some separable Hilbert space $\mathtt{H}$, whose spectrum is denoted by $\spec(\mathtt{h})\subset\R\cup\{\infty\}$. Physically, we say that the system described by $\mathtt{H}$ has a \emph{gap} if whenever we \emph{measure} the spectrum of the associated Hamiltonian there exists a strictly positive \emph{distance} $\gamma\in\R^{+}$ between the two lowest eigenvalues $\mathcal{E}_{1},\mathcal{E}_{2}\in\R$ such that $\mathcal{E}_{2}-\mathcal{E}_{1}>\gamma$, with $\mathcal{E}_{1}\doteq\inf\spec(\mathtt{h})$. The parameter $\gamma$, also called \emph{spectral gap}, is known to be the difference between the lowest energy of the system and the energy of its first \emph{excited} state. In Definition \ref{def: spectral gap} below, we formally express this. On the other hand, in the context of fermion systems, Definition \ref{def: spectral gap fer} is suitable for our interests. Then, introducing the notation $\mathfrak{d}(X,Y)$ to denote the distance between the sets $X,Y\subset\R$:
\begin{align*}
\mathfrak{d}(X,Y)\doteq\inf\left\{d(x,y)\colon x\in X,\,y\in Y\right\},
\end{align*}
with $d(x,y)\doteq|x-y|$ for $x,y\in\R$, we define:
\begin{definition}[Gapped Hamiltonians]\label{def: spectral gap}
Let $\mathtt{H}$ be a (one--particle) Hilbert space and consider $\mathtt{h}\in\LO(\mathtt{H})$ the (one--particle) Hamiltonian, that is, a self--adjoint operator $\mathtt{h}=\mathtt{h}^{*}$, whose spectrum is denoted by $\spec(\mathtt{h})\subset\R$. We will say that $\mathtt{h}$ is a \emph{gapped Hamiltonian} if there are $\Sigma$ and $\widetilde{\Sigma}$, nonempty and disjoint subsets of $\spec(\mathtt{h})$, such that $\Sigma\cup\widetilde{\Sigma}=\spec(\mathtt{h})$ and such that  $\gamma\doteq\mathfrak{d}(\Sigma,\widetilde{\Sigma})>0$.
\end{definition}
\begin{remark}
In the latter definition $\Sigma\in\mathbb{R}$ can be a Borel set containing the isolated eigenvalue
 $\mathcal{E}_{1}$, which carries the information of the lowest energy associated to the physical system to consider. Note that if $\spec(\mathtt{h})$ is a purely \emph{point spectrum} (the set of all the eigenvalues associated to $\mathtt{h}$) we can define the family of elements of $\widetilde{\Sigma}$ with indices on $\N\setminus\{1\}$ as the map $\mathcal{E}\colon\N\setminus\{1\}\to\widetilde{\Sigma}$, such that $\mathcal{E}\doteq\{\mathcal{E}_{n}\}_{n\in\N\setminus\{1\}}$, the rest of eigenvalues of $\mathtt{H}$, given $\mathcal{E}$, belong to $\widetilde{\Sigma}$.
\end{remark}
Definition \ref{def: spectral gap} is completely general and is usually used to study spectrum related to physical systems. Nevertheless, our primary interest is the fermionic case and then we need to consider an alternative expression. In order to find such an expression recall Definition \ref{def one particle hamiltinian} of a self--dual operator $H\in\BL(\H)$\footnote{Following \cite{nachtergaele2018quasi}, we could consider \emph{unbounded} one--site potentials. Thus, we would need to define \emph{Hamiltonians} on well--defined dense sets on Hilbert spaces. However, for the sake of simplicity, we will omit any mention on densely defined self--adjoint operators. Note that in \cite{brupedraLR}, Bru and Pedra consider unbounded one--site potentials on $\C$--algebras. In \cite{AMR2} we will deal with unbounded one--site potentials.}, where one considers a self--dual Hilbert space $(\H,\a)$, with $\H$ a finite--dimensional Hilbert space with orthonormal basis given by $\{\psi_{i}\}_{i\in I}$. Hence, for any $H\in\BL(\H)$ satisfying $H^{*}=-\a H\a$ we have:
\begin{enumerate}
 \item[(i)] $\tr_{\H}(H)=0$.
 \item[(ii)] $\spec(\lambda\1_{\H}-H)=\lambda-\spec(H)$ for $\lambda\in\CP$.
\end{enumerate}
Both (i) and (ii) are fundamental to study the underlying systems we are considering. Based on Definition \ref{def: spectral gap}, items (i) and (ii), and above comments we can define the following:
\begin{definition}[Fermionic Gapped Hamiltonians]\label{def: spectral gap fer}
Let $(\H,\a)$ be a self--dual Hilbert space and consider $H\in\BL(\H)$ be a self--dual Hamiltonian with spectrum denoted by $\spec(H)\subset\R$. We will say that $H$ is a \emph{gapped Hamiltonian} if there exists $\g\in\R^{+}$ satisfying the \emph{gap assumption}
\begin{align*}
\g\doteq\inf\left\{\epsilon>0\colon\left[-\epsilon,\epsilon\right]\cap\spec(H)\neq\emptyset\right\}.
\end{align*}
\end{definition}
Observe that for fermionic systems Definitions \ref{def: spectral gap} and \ref{def: spectral gap fer} are equivalent. In fact, in Definition \ref{def: spectral gap fer}, $\Sigma\in\R$ is a finite interval with $a\doteq\inf\{\Sigma\}$ and $b\doteq\sup\{\Sigma\}$, $\widetilde{\Sigma}$ is nothing but $-\Sigma$, so that $-a\doteq\sup\{\widetilde{\Sigma}\}$ and $-b\doteq\inf\{\widetilde{\Sigma}\}$. Then, the self--dual formalism permits to consider a \emph{symmetric} decomposition of the spectrum. Therefore $\Sigma$ can be understood as a Borel set on $\R^{+}$ related to the positive part of the energy while $\widetilde{\Sigma}\equiv-\Sigma$ its symmetric negative part: the gap $\g$ centered at zero separates these. We finally stress following Definition \ref{def: spectral gap fer} that denoting by $\Sigma_{0}$ and $-\Sigma_{0}$ the remaining two open sets, their closures respectively are $\Sigma$ and $-\Sigma$.\par
Due to the above reasons, from now on we will only consider fermion systems. Thus, let us now consider the family of self--dual Hamiltonians
%[R++ $\{H_{s}\}_{s\in\cal}\in\BL(\H)$ ++R]
%[A++ $\{H_{s}\}_{s\in\cal}\subset\BL(\H)$ ++A]
$\{H_{s}\}_{s\in\cal}\subset\BL(\H)$ on $(\H,\a)$, where $\cal$ is the compact set $[0,1]$. In particular, $\{H_{s}\}_{s\in\cal}$ will define a differentiable family of self--adjoint operators on $\BL(\H)$. More specifically, for any $s\in\cal$ we will consider that the map $s\mapsto H_{s}$ is strongly differentiable so that $\p_{s}H_{s}\in\BL(\H)$. For example, we are particularly interested in the family of differentiable operators $H_{s}\doteq (1-s)H_{0}+sH_{1}$, for any $s\in\cal$. Other models we are taking into account, is the Anderson model, as discussed in Appendix \ref{appendix: GAM}. See \cite{bru2014heat} and \cite{ABdS3}.  Following Definition \ref{def: spectral gap fer} we now define:
\begin{definition}[Phase of Matter]\label{def: phase of the matter}
Let $\cal\equiv[0,1]$ and $\{H_{s}\}_{s\in\cal}\in\BL(\H)$ be a continuous family of self--dual Hamiltonians on $(\H,\a)$. We will say that $H_{s}$ is a $s$--gapped Hamiltonian if the \emph{gap assumption} in Definition \ref{def: spectral gap fer} is satisfied for any $s\in\cal$.
Define now an equivalence relation on the set of gapped self-dual Hamiltonians, as follows. If $H^{(0)},H^{(1)}\in \BL(\H)$ are two  gapped self-dual Hamiltonians, then $H^{(0)}\sim H^{(1)}$ if and only if the following conditions are satisfied:
\begin{itemize}
\item[(i)] There is a continuous family of
 self-dual Hamiltonians on $(\H,\a)$, where   $H_s$ is $s$--gapped for all $s\in\cal$;
 \item[(ii)] there is a lower bound  $\g\in\R^{+}$, independent of $s$, in the sense that $\inf\limits_{s\in\cal}\g_{s}\geq\g>0$; and
\item[(iii)]  $H_0=H^{(0)}$ and $H_1=H^{(1)}$.
\end{itemize}
We will refer to such equivalence classes as \emph{phases of matter}, and will use the notation $\mathbf \Omega$ to denote a given phase (with an appropriate label, if needed).
\end{definition}
Observe that a difference between ground states associated to family of Hamiltonians $\{H_{s}\}_{s\in\cal}$ satisfying above definition and the general definition of ground states (Definition \ref{def: ground state}) is necessary. In fact, one can prove that if the family of Hamiltonians is gapped, then its associated ground states $\{\omega_{s}\}_{s\in\cal}$ satisfy:
\begin{equation}\label{eq:GGS}
\ii\omega_{s}(A^{*}\delta(A))\geq\g_{s}(\omega_{s}(A^{*}A)-|\omega_{s}(A)|^{2}),\quad\text{for any}\quad s\in\cal\quad\text{and}\quad A\in\mathcal{D}(\delta),
\end{equation}
with $\g_{s}\in\R^{+}, s\in\cal$, and $\inf\limits_{s\in\cal}\g_{s}\geq\g>0$. For details see \cite{matsui2013}. In the sequel we will say that states satisfying the above inequality are \emph{gapped ground states}.
\section{Main Results}\label{sec: main results}
We study gapped Hamiltonians satisfying the following Assumption:
\begin{assumption}\label{assump:hamil}
Take $\cal\equiv[0,1]$. (a) $\mathbf{H}\doteq\{H_{s}\}_{s\in\cal}\subset\BL(\H)$ is a differentiable family of self--dual Hamiltonians such that $\p\mathbf{H}\doteq\{\p_{s}H_{s}\}_{s\in\cal}\subset\BL(\H)$. (b) For the infinite volume we assume that the sequences of self--dual Hamiltonians $H_{s,\,L}\colon\cal\to\BL(\H_{\infty})$ and $\p_{s}H_{s,\,L}\colon\cal\to\BL(\H_{\infty})$ convergent in norm and pointwise, that is, $\lim\limits_{L\to\infty}H_{s,L}=H_{s,\infty}$ and $\lim\limits_{L\to\infty}\p_{s}H_{s,L}=\p_{s}H_{s,\infty}$ in the norm sense.
\end{assumption}
Now, for any self--dual Hilbert space $(\H,\a)$, take $P_{1}\in\fp$ and $P_{2}\in\fp$ basis projections, the ``$\Z_{2}$--projection index'' ($\Z_{2}$--PI) $\sigma\colon\fp\times\fp\to\Z_{2}$ is the map defined by:
\begin{equation}\label{eq:top_index}
\sigma(P_{1},P_{2})\doteq(-1)^{\dim(P_{1}\land P_{2}^{\perp})}.
\end{equation}
Here, $\land$ symbolizes the \emph{lower bound} or \emph{intersection} of the basis projections $P_{1}$ and $P_{2}$ in $\fp$. Note that the $\Z_{2}$--PI defines a \emph{topological group} with two components. In particular, $\sigma(P_{1},P_{2})$ gives an equivalence criterion for their associated quasi--free states $\omega_{P_{1}}$ and $\omega_{P_{2}}$ restricted to the even part $\sCAR^{+}$ of the self--dual $\C$--algebra $\sCAR$. See Expression \eqref{eq:even odd}. More generally, we know by the Shale--Stinespring Theorem that two Fock representations $\pi_{P_{1}}$ and $\pi_{P_{2}}$ associated to $P_{1},P_{2}\in\fp$ are unitarily equivalent if and only if $P_{1}-P_{2}\in\mathcal{I}_{2}$, i.e., a \emph{Hilbert--Schmidt class} operator \cite{Varilly}. See Appendix \ref{appendix: fock}, in special Equation \eqref{eq:fock_rep} and \cite[Theo. 6.14]{A87}. Then, we analyze the class of Hamiltonians described by last assumption and their connection with topological indexes.
We formally state one of the main results of the paper:
\begin{theorem}[$\Z_{2}$--projection Index]\label{theorem:states_algebras2}
Take $\cal\equiv[0,1]$ and let $\mathbf{H}\doteq\{H_{s,\infty}\}_{s\in\cal}\subset\BL(\H_{\infty})$ be a differentiable family of gapped self--dual Hamiltonians on $(\H_{\infty},\a_{\infty})$, with $\p\mathbf{H}\doteq\{\p_{s}H_{s,\infty}\}_{s\in\cal}\subset\BL(\H_{\infty})$, see Definition \ref{def: phase of the matter} and Assumption \ref{assump:hamil} (b).
For any $s\in\cal$, $E_{+,s,\infty}$ denotes the spectral projection associated to the positive part of $\spec(H_{s,\infty})$ and consider the $\Z_{2}$--PI given by \eqref{eq:top_index}. Then:
\begin{enumerate}
\item[(1)] For any $s\in\cal$, $H_{0,\infty}$ is unitarily equivalent to $H_{s,\infty}$ via the unitary operator $V_{s}^{(\infty)}$ on $(\H_{\infty},\Gamma_{\infty})$ satisfying the differential equation \eqref{eq:spectral proof2}.
\item[(2)] The Bogoliubov $^{*}$--automorphism $\chi_{V_{s}^{(\infty)}}$ is inner and maintains its parity, \emph{even} $V_{s}^{(\infty)}\in\mathfrak{U}_{+}^{\infty}$ or \emph{odd} $V_{s}^{(\infty)}\in\mathfrak{U}_{-}^{\infty}$, over the family $\mathbf{H}$\footnote{Here $U\in\mathfrak{U}_{\pm}^{\infty}$ is a Bogoliubov transformation as defined in Expression \eqref{orientation} associated to the self--dual Hilbert space $(\H_{\infty},\a_{\infty})$.}.
 \item[(3)] For $r,s\in\cal$, the $\Z_{2}$--PI $\sigma(H_{r,\infty},H_{s,\infty})\equiv\sigma(E_{+,r,\infty},E_{+,s,\infty})$ satisfies $\sigma(H_{r,\infty},H_{s,\infty})=1$.
\end{enumerate}
\end{theorem}
In regard to this Theorem some remarks are in order:\\
Consider the pair, $P\in\fp$ a basis projection and $U\in\mathfrak{U}_{\pm}$ a Bogoliubov transformation as defined in Expression \eqref{orientation}. Then, $\dim\left(\mathrm{ker}(PUP)\right)\in\N_{0}$ and $U\in\mathfrak{U}_{+}$ or $U\in\mathfrak{U}_{-}$ if $\dim\left(\mathrm{ker}(PUP)\right)$ is respectively even or odd \cite[Theo. 6.3]{A87}. According to Theorem \ref{theorem:states_algebras2} (2), it follows that the number
$$
\dim\left(\mathrm{ker}(E_{+,s,\infty}V_{s}^{(\infty)}E_{+,s,\infty})\right)\in\N_{0}
$$
is uniform for the family $\mathbf{H}$. Physically, this is in close relation with the  number of the particles of the systems described by the family of Hamiltonians $\mathbf{H}$. In fact, consider the even and odd parts $\mathfrak{A}_{\infty}^{\pm}\subset \mathfrak{A}_{\infty}$ of the self--dual $\CAR$ $\C$--algebra associated to the self--dual Hilbert space $(\H_{\infty},\a_{\infty})$, see Expressions \eqref{eq:even odd} and \eqref{local elements}, and consider $\pi_{E_{+,s,\infty}}$,  the fermionic Fock representation associated to $E_{+,s,\infty}$ such that can be decomposed: $\pi_{E_{+,s,\infty}}=\pi_{E_{+,s,\infty}}^{+}\oplus\pi_{E_{+,s,\infty}}^{-}$, where $\pi_{E_{+,s,\infty}}^{\pm}$ is the restriction of $\omega_{E_{+,s,\infty}}$ to $\mathfrak{A}_{\infty}^{\pm}$. Then, the GNS representation associated to the vacuum vector $\Omega_{E_{+,s,\infty}}$ given by \eqref{eq:vacuum_P}, namely, $\pi_{\Omega_{E_{+,s,\infty}}}$ is identified with $\pi_{E_{+,s,\infty}}^{\pm}$ or $\pi_{E_{+,s,\infty}}^{\mp}$ depending if $|J|$ in Expression \eqref{eq:vacuum_P} is even or odd \cite{EK98}. Then, the physical meaning of Theorem \ref{theorem:states_algebras2}--(3) can be understood by saying that the $\Z_{2}$--PI is $1$ for any two self--dual Hamiltonians $H_{r},H_{s}\in\mathbf{H}$, with $r,s\in\cal$ (same parity).

\begin{proof}

(1) For any $A\in\BL(\H_{\infty})$ and all $s\in\cal$ define the \emph{spectral flow automorphism} $\kappa_{s}\colon\BL(\H_{\infty})\to\BL(\H_{\infty})$ by
$$
\kappa_{s}(A)\doteq\left(V_{s}^{(\infty)}\right)^{*}AV_{s}^{(\infty)},
$$
where $V_{s}^{(\infty)}$ is a linear operator satisfying $V_{0}^{(\infty)}=\pm\1_{\H_{\infty}}$, and the differential equation \eqref{eq:spectral proof2}. See Lemmata \ref{lemma:impor_res}--\ref{theor:unit_oper} and Corollary \ref{corol:V_cauchy}. In particular, since any Hamiltonian $H_{s,\infty}$ in $\mathbf{H}$ can be written as
$$
H_{s,\infty}=\int_{\R}\lambda\d E_{\lambda,s,\infty},
$$
with $\a E_{\lambda,s,\infty}\a=E_{-\lambda,s,\infty}$ for all $\lambda\in\R$ and $E_{-,s,\infty}+E_{+,s,\infty}=\1_{\H_{\infty}}$, by Lemmata \ref{lemma:impor_res}--\ref{theor:unit_oper}, (1) follows.\\
By comments around Expression \ref{orientation}, a Bogoliubov $^{*}$--automorphism $\chi_{U}$ on a self--dual $\CAR$--algebra is even or odd if and only if $\det(U)=1$ or $\det(U)=-1$, respectively. Then, part (2) follows from Corollary \ref{corol:sign} and Lemmata \ref{lemma:VI_cauchy}--\ref{lemma:unit_oper2}.\\
(3) Concerning the $\Z_{2}$--PI $\sigma(P_{1},P_{2})$ we first invoke \cite[Theo. 6.30 and Lemma 7.17]{EK98}: (a) $\sigma(P_{1},P_{2})=\sigma(P_{2},P_{1})$, (b) If $P_{1}-P_{2}\in\mathcal{I}_{2}$, a Hilbert--Schmidt class operator, then $\sigma(P_{1},P_{2})$ is continuous in $P_{1}$ and $P_{2}$ with respect to the norm topology in $\fp$ (c) If $U\in\BL(\H)$ is a unitary operator such that $U\a=\a U$ and $\1_{\H}-U$ is a trace class operator, then $\sigma(P,UPU^{*})=\det U$. Then we proceed to verify these statements for the family of positive spectral projections $\{E_{+,s,\infty}\}_{s\in\cal}\in\BL(\H_{\infty})$. By \eqref{eq:sum_basis_proj}--\eqref{eq:symbolbis2} and comments around it, any positive spectral projection in $\{E_{+,s,\infty}\}_{s\in\cal}$ is a basis projection and thus $\{E_{+,s,\infty}\}_{s\in\cal}\subset\mathfrak{p}(\H_{\infty},\a_{\infty})$. W.l.o.g. take $E_{+,r,\infty}$ and $E_{+,s,\infty}$ with $r,s\in\cal$. We then verify (a)--(c) as follows:
\begin{enumerate}
\item[(a)] Note that $E_{+,r,\infty}\land E_{+,s,\infty}^{\perp}=\a_{\infty}\left(E_{+,r,\infty}^{\perp}\land E_{+,s,\infty}\right)\a_{\infty}$.
\item[(b)] For $L\in\R_{0}^{+}$, we need to verify that $E_{+,r,L}-E_{+,s,L}\in\mathcal{I}_{2}$ as $L\to\infty$. Here, $(\H_{L},\a_{L})$ is the Hilbert space given by the canonical orthonormal basis $\{\E_{\x}\}_{\x\in\mathbb{X}_{L}}$ defined by \eqref{eq: X_L} below. Since $E_{+,r,L}$ and $E_{+,s,L}$ are self--adjoint operators on $\BL(\H_{L})$ by Lemma \ref{lemma:impor_res}, there are unitary operators $V_{r}^{(L)},V_{s}^{(L)}$ such that $E_{+,s,L}=U_{L}(s,r)E_{+,r,L}U_{L}(r,s)$, with $U_{L}(s,r)\doteq V_{s}^{(L)}\left(V_{r}^{(L)}\right)^{*}$ satisfying \eqref{eq:unit_ele_sr2} and $U_{L}(0,0)=\1_{\H_{L}}$. Now, for $s,r\in\cal$, note from \eqref{eq:unit_ele_sr} that
\begin{align*}
U_{L_{2}}(s,r)-U_{L_{1}}(s,r)&=V_{s}^{(L_{2})}\left(V_{r}^{(L_{2})}\right)^{*}-V_{s}^{(L_{1})}\left(V_{r}^{(L_{1})}\right)^{*}\\
&=\left(V_{s}^{(L_{2})}-V_{s}^{(L_{1})}\right)\left(V_{r}^{(L_{2})}\right)^{*}-V_{s}^{L_{1}}\left(\left(V_{r}^{(L_{1})}\right)^{*}-\left(V_{r}^{(L_{2})}\right)^{*}\right).
\end{align*}
By Corollary \ref{corol:V_cauchy}, the sequence of linear operators $V_{s}^{(L)}$ converges in norm to some linear operator $V_{s}$ as $L\to\infty$. Thus by above expression for fixed $s,r\in\cal$, $U_{L}(s,r)$ converges in norm to some $U(s,r)$ as $L\to\infty$. It follows that $E_{+,r,\infty}$ and $E_{+,s,\infty}$ are unitarily equivalent. Similar to Theorem \ref{lemma: supporting states2} and using Expression \eqref{eq:symbolbis2}, the quasi--free ground states $\omega_{E_{+,s,\infty}}$ and $\omega_{E_{+,r,\infty}}$ are unitarily equivalent via a Bogoliubov $^{*}$--automorphism $\mathbf{\chi}_{U(s,r)}$, see \eqref{Bogoliubov  automorphism}, such that
$$
\omega_{E_{+,s,\infty}}=\omega_{E_{+,r,\infty}}\circ\mathbf{\chi}_{U(s,r)}.
$$
In particular, by the Shale--Stinespring Theorem, the Fock representations $\pi_{E_{+,r,\infty}}$ and $\pi_{E_{+,s,\infty}}$ are unitarily equivalent and therefore the difference $E_{+,r,\infty}-E_{+,s,\infty}$ is Hilbert--Schmidt class. See \cite[Sec. 6.3]{Varilly}, \cite[Theo. 6]{A70} and \cite[Theo. 6.14]{A87}.
\item[(c)] By Lemma \ref{lemma:unit_oper3}, for any $r,s\in\cal$ and $L\in\R_{0}^{+}\cup\{\infty\}$, the operator $U_{L}(s,r)$ already defined commutes with $\a_{L}$, is a trace--class per unit volume operator and $\det\left(U_{L}(s,r)\right)=1$. Then, by the continuity of the $\Z_{2}$--PI index on the norm topology on $\mathfrak{p}(\H_{\infty},\a_{\infty})$ note that \cite[Theo. 3]{araEva}--\cite[Lemma 7.17]{EK98}:
\begin{align*}
\sigma\left(E_{+,r,\infty},E_{+,s,\infty}\right) &= \sigma\left(\left(V_{r}^{(\infty)}\right)^{*}E_{+,0,\infty}V_{r}^{(\infty)},\left(V_{s}^{(\infty)}\right)^{*}E_{+,0,\infty}V_{s}^{(\infty)}\right)\\
&= \sigma\left(E_{+,0,\infty},U_{\infty}(r,s)E_{+,0,\infty}U_{\infty}(s,r)\right)\\
&= \det\left(U_{\infty}(r,s)\right)=1.
\end{align*}
\end{enumerate}
\end{proof}
\par
Hitherto in this paper we have been interested in physical systems with \emph{open} gap, which is the case of systems of last Theorem. In fact, Theorem \ref{theorem:states_algebras2}--(1) claims that two self--dual Hamiltonians that belong to the same phase of matter, say $H_{0,\infty},H_{1,\infty}\in\mathbf{\Omega}$, and that are connected by a \emph{differentiable path},  can in fact  be connected  through the spectral flow automorphism $\kappa_{s}$ on $\BL(\H_{\infty})$. As is usual, a path is nothing but a continuous map $\gamma\colon\cal\to\BL(\H_{\infty})$ connecting the \emph{initial point} $\gamma_{0}=H_{0,\infty}$ and the \emph{terminal point} $\gamma_{1}=H_{1,\infty}$. Equivalently, we say that $H_{0,\infty}$ and $H_{1,\infty}$ are the  \emph{extremal points} of the path. Observe that for any $s,r\in\cal$ with $H_{s,\infty},H_{r,\infty}\in \mathbf{H}\subset \mathbf \Omega$ we can write
$$
H_{s,\infty}=\kappa_{s,r}(H_{r,\infty}), \quad\text{with}\quad\kappa_{s,r}\doteq\kappa_{s}^{-1}\circ\kappa_{r}
$$
 In this case, the index is the same for any pair of Hamiltonians along the path.
Notice that, by definition, each phase of matter $\mathbf \Omega$ is \emph{arcwise connected}. But different phases correspond to disjoint sets, so that
we are also interested in the possibility of having two self--dual Hamiltonians, $H_{0,\infty}$ and $H_{2,\infty}$, acting on $\H_{\infty}$ but belonging to \emph{different} phases of matter, in the sense of Definition \ref{def: phase of the matter}. In this case, if $H_{0,\infty}\in\mathbf{\Omega}$ while $H_{2,\infty}\notin \mathbf{\Omega}$, Theorem \ref{theorem:groun_states} (see Corollary \ref{cor:gap} too) below shows that the path $\widetilde{\kappa}$ connecting both Hamiltonians \emph{closes} the gap,  meaning that there is a Hamiltonian $\widetilde{H}\in\BL(\H_{\infty})$ on $\widetilde{\kappa}$ such that $0$ is an eigenvalue of $\widetilde{H}$. Concerning the latter, observe that one can study the gap closing in terms of the self--dual $\CAR$, $\C$--algebra $\A_{\infty}\doteq\mathrm{sCAR}(\H_{\infty},\a_{\infty})$, in such a way that we associate to $H_{0,\infty}$ and $H_{2,\infty}$ the bilinear elements $\Bil{H_{0,\infty}}$ and $\Bil{H_{2,\infty}}$ on $\A_{\infty}$ (see Expression \eqref{local elements} below).\\
We state the second main result of the current paper:
\begin{theorem}\label{theorem:groun_states}
Take $\cal\equiv[0,1]$ and let $\mathbf{H}\doteq\{H_{s,\infty}\}_{s\in\cal}\subset\BL(\H_{\infty})$ be a differentiable path of self--dual Hamiltonians on $(\H_{\infty},\a_{\infty})$, with $\p\mathbf{H}\doteq\{\p_{s}H_{s,\infty}\}_{s\in\cal}\subset\BL(\H_{\infty})$.
Suppose that $H_{0,\infty}$ and $H_{1,\infty}$ are gapped, and  their corresponding basis projections satisfy $\sigma (E_{+,0,\infty}, E_{+,1,\infty})=-1$. Then there exists   $\tilde s\in \mathring{\cal}$  such that $0\in\spec(H_{\tilde s,\infty})$, i.e., such that $E_{0,\tilde s, \infty}\neq 0$. Furthermore, the path $\mathbf{H}$
induces, in a natural way, a family of ground states for $H_{\tilde s,\infty}$ which are generically mixed states.
More specifically, we have:
\begin{enumerate}
\item[(1)]  The gap closing at $\tilde s$ induces a splitting of the one--particle  Hilbert space as the direct sum of two $\Gamma_\infty$--invariant subspaces,  $\H_{\infty}=\mathcal K_0\oplus \mathcal K_1$, as well as a factorization $\A_{\infty}\cong\A_{0}\otimes \A_{1}$, where $\A_{0}\doteq \mathrm{sCAR}(\mathcal K_0,\a_{\infty})$ and $\A_{1}\doteq \mathrm{sCAR}(\mathcal K_1,\a_{\infty})$.
\item[(2)] The set of ground states associated with the self--dual Hamiltonian $H_{\tilde s,\infty}$ have the form
$$
\widetilde\omega(a_1a_0)=\omega_{\tilde P_+}(a_1)\phi_0(a_0),
$$
where $a_1\in \A_{1}$, $\omega_{\tilde P_+}\in\states^{(\infty)}$ is a quasi--free ground state on $\mathfrak A_1$ associated  with a basis projection $\tilde P_+\in \mathfrak{p}(\mathcal K_1,\Gamma_\infty)$, and $a_0\in \A_{0}$ with  $\phi_0\in\states^{(\infty)}$ an arbitrary state on $\mathfrak A_0$. In particular, the right and left limits $\lim\limits_{s\to \tilde s^\pm} E_{+,s,\infty}$ single out states $\omega_0^{+}$ and $\omega_0^{-}$ on $\mathfrak A_0$, with $\omega_0^{+}\neq\omega_0^{-}$. In this sense, the path $\mathbf{H}$ induces, in a natural way, the following family of ground states for $H_{\tilde s,\infty}$:
\begin{equation}\label{eq:state-thm2-mixed}
\widetilde\omega_\lambda(a_1a_0)=\omega_{\tilde P_+}(a_1)(\lambda \omega_0^{+}(a_0) +(1-\lambda)\omega_0^{-}(a_0)),\qquad  \lambda\in[0,1].
\end{equation}
\end{enumerate}
\end{theorem}
\begin{proof}
In view of Theorem \ref{theorem:states_algebras2}, and taking into account the hypothesis  $\sigma (E_{+,0,\infty}, E_{+,1,\infty})=-1$, it is clear that the gap assumption cannot hold for all $s \in\cal$. Therefore, the set $\mathcal N\doteq \lbrace s\in \mathring{\cal}\colon E_{0,s,\infty}\neq 0\rbrace$ is nonempty. For the sake of simplicity, we assume that $|\mathcal N|=1$, i. e., $\mathcal{N}$ contains one element, say $\tilde s$. Define now the upper/lower limit (in norm)
$$
E_{\tilde s^\pm}\doteq \lim_{s\to \tilde s^\pm} E_{+,s,\infty}.
$$
$E_{\tilde s^+}$ and $E_{\tilde s^-}$ are basis projections, but they are \emph{not} equal and do \emph{not} correspond to the positive spectral projection $E_{+,\tilde s,\infty}$. Let $\omega_\pm\in\q\states^{(\infty)}$  denote the  quasi--free state with symbol $E_{\tilde s^\pm}$ (cf. Equation (\ref{symbolbis})). By \cite[Lemma 3.6]{A68}, there exists a  Bogoliubov transformation   $U_{\tilde{s}^{+},\tilde{s}^{-}}\in\BL(\H_{\infty})$ such that
$E_{\tilde s^+}=U_{\tilde{s}^{+},\tilde{s}^{-}}^{*}E_{\tilde s^-}U_{\tilde{s}^{+},\tilde{s}^{-}}$.  See Expression \eqref{Bogoliubov  automorphism} and comments around it.  Since the gap condition is satisfied for all $s\in \cal\setminus \lbrace \tilde s\rbrace$, it is clear that
$\det U_{\tilde{s}^{+},\tilde{s}^{-}}=-1$ (in particular, $\1_{\H_{\infty}}-U_{\tilde{s}^{+},\tilde{s}^{-}}$ is trace--class)  and so $\sigma (E_{\tilde s^+},E_{\tilde s^-})=-1$, by \cite[Lemma 7.17]{EK98}. This implies, by \cite[Theorem 6.30]{EK98} that the restrictions of $\omega_+$ and $\omega_-$ to the even algebra   $\A_{\infty}^{+}$ lead to inequivalent GNS representations and, therefore, we must have $\omega_+\neq\omega_-$. We will use $\omega_{\pm}$ to construct a family $\lbrace\widetilde\omega_\lambda\rbrace_{\lambda\in [0,1]}$ of \emph{ground states} for  $H_{\tilde s}\in\mathbf{H}$. We proceed in steps, as follows:
\begin{enumerate}
\item\label{proof:step1} Define the following $\Gamma_\infty$--invariant subspaces of $\H_{\infty}$:
$$
\mathcal K_0 \doteq \left(\left(E_{\tilde s^+}\wedge(\1_{\H_{\infty}}-E_{\tilde s^-})\right)\H_{\infty}\right)\oplus(\left(E_{\tilde s^-}\wedge(\1_{\H_{\infty}}-E_{\tilde s^+}\right))\H_{\infty})
$$
and
$$
\mathcal K_1 \doteq (\left(E_{\tilde s^+}\wedge E_{\tilde s^-}\right)\H_{\infty})\oplus(\left((\1_{\H_{\infty}}-E_{\tilde s^+})\wedge(\1_{\H_{\infty}}-E_{\tilde s^-})\right)\H_{\infty}).
$$
\item\label{proof:step2} Define the   self--dual $\CAR$ $\C$--algebras  $\mathfrak A_0\doteq \mathrm{sCAR}(\mathcal K_0,\a_{\infty})$ and  $\mathfrak A_1\doteq \mathrm{sCAR}(\mathcal K_1,\a_{\infty})$. Notice that the restrictions of $\omega_+$ and $\omega_-$ to $\mathfrak A_1$ do coincide and correspond to the quasi--free  state defined by the basis projection $\tilde P_+\doteq E_{\tilde s^+}\wedge E_{\tilde s^-}$. We denote this state as $\omega_{\tilde P_+}$. On the other hand, the restrictions of $\omega_{\pm}$  to  $\mathfrak A_0$ are different (this follows from
      $\det U_{\tilde{s}^{+},\tilde{s}^{-}}=-1$). We therefore define $\omega_0^\pm \doteq \omega_{\pm}|_{\mathfrak A_0}$. Notice that $\tilde P_+$ is a basis projection only if regarded as a projection on $\mathcal K_1$, for if we define $\tilde P_-\doteq (\1_{\H_{\infty}}-E_{\tilde s^+})\wedge (\1_{\H_{\infty}}-E_{\tilde s^-})$ and   $\tilde P_0\doteq E_{\tilde s^+}\wedge (\1_{\H_{\infty}}-E_{\tilde s^-})+E_{\tilde s^-}\wedge (\1_{\H_{\infty}}-E_{\tilde s^+})$, we obtain  $\tilde P_{+} +\tilde P_{-} + \tilde P_{0}=\1_{\H_{\infty}}$.
\item\label{proof:step3} By \cite[Prop. 5.3.19]{BratteliRobinson}, if $\varpi\in\states^{(\infty)}$ is a ground state and $f\in C^{\infty}$ has Fourier transform $\hat{f}$ satisfying $\supp\hat{f}\subset\R_{-}$, we have
\begin{equation}\label{eq:ground_four}
\varpi(\tau_{f}(A)^{*}\tau_{f}(A))=0,\quad{\text{for}}\quad A\in\A_{\infty},
\end{equation}
with $\tau_{f}(\B(\varphi))=\B(\hat{f}(H)\varphi)$, $\varphi\in\H_{\infty}$, where $H\in\BL(\H_{\infty})$ is a self--dual Hamiltonian on $(\H_{\infty},\Gamma_{\infty})$ associated with $\varpi$.
\item\label{proof:mix_state1} Take $\tilde P_0$ and $\tilde P_\pm$ as defined in step \ref{proof:step2}. We are now in the situation of \cite[Proposition 6.37]{EK98} (see also \cite{araki1985ground}). Following these references, one can show that all ground states for $H_{\tilde s}$ are of the form
\begin{equation}\label{eq:mix_state1}
\widetilde\omega(a_1a_0) = \omega_{\tilde P_+}(a_1) \phi_0 (a_0),
\end{equation}
where $\omega_{\tilde P_+}$   is the  quasi--free state defined in step \ref{proof:step2} and  $\phi_0$  an arbitrary state on $\mathfrak A_0$. We will provide  some  details for completeness, and will show how the path $\mathbf{H}\doteq\{H_{s,\infty}\}_{s\in\cal}$ induces a particular family of (mixed) ground states for  $H_{\tilde s}$.
\begin{enumerate}
\item We first suppose that $\widetilde\omega\in\states^{(\infty)}$ is a ground state for $H_{\tilde s}$. It then follows from step
     \ref{proof:step3}, Expression \eqref{eq:ground_four},
      that
      $\varphi(\B(\varphi)^{*}\B(\varphi))=0$, whenever $\varphi\in \tilde P_{-}\H_{\infty}$.
      Thus, by the \emph{Schwarz inequality} one notes that for any $C\in\A_{\infty}$ and $\varphi\in \tilde P_{-}\H_{\infty}$
$$
|\widetilde{\omega}(\C\B(\varphi))|^{2}\leq \widetilde{\omega}(\C C)\widetilde{\omega}(\B(\varphi)^{*}\B(\varphi))=0.
$$
Consequently, for any $C\in\A_{\infty}$ and $\varphi\in \tilde P_{-}\H_{\infty}$,
\begin{equation}
\label{eq:B-vanishes}
\widetilde{\omega}(\C\B(\varphi))=\widetilde{\omega}(\B(\varphi)^{*}C)=\widetilde{\omega}(C\B(\varphi))=\widetilde{\omega}(\B(\varphi)^{*}\C)=0.
\end{equation}
For $\varphi\in \mathcal K_1$ we write $\varphi=\varphi_{+}+\varphi_{-}$, with $\varphi_\pm =\tilde P_\pm\varphi$.
Notice  that $\Gamma_{\infty}\varphi_{+},\varphi_{-}\in \tilde P_{-}\H_{\infty}$ and $\B(\varphi)=\B(\Gamma_{\infty}\varphi_{+})^{*}+\B(\varphi_{-})\in\A_1$. Hence, in order to verify \eqref{eq:mix_state1}, by linearity it suffices to consider for $m,n\in\N_{0}$, with $m+n\in\N$,
$$
a_1=\B\left(\varphi_{1}\right)^{*}\cdots\B\left(\varphi_{m}\right)^{*}\B
\left(\widetilde{\varphi}_{1}\right)\cdots\B\left(\widetilde{\varphi}_{n}\right)
\in\A_{1}
$$
for any $\varphi_{i},\widetilde{\varphi}_{j}\in \tilde P_{-}\H_{\infty}$, with $i\in\{1,\ldots,m\}$ and $j\in\{1,\ldots,n\}$.
It then follows from (\ref{eq:B-vanishes}) above, on using the CAR relations, that $\widetilde \omega(a_1)= \omega_{\tilde P_+}(a_1)$
(cf.\cite[Section 5.2]{BratteliRobinson}). Equality in (\ref{eq:mix_state1}) follows upon defining $\phi_0\doteq \widetilde \omega|_{\mathfrak A_0}$.
\item From step \ref{proof:step3} and the definitions given in step \ref{proof:step2}, it follows in particular that $\omega_\pm(a_1a_0)=\omega_{\tilde P_+}(a_1)\omega^\pm_{0}(a_0)$. We can explicitly show that $\omega_+$ and $\omega_-$ are ground states for $H_{\tilde s}$, by considering the GNS representations of the states $\omega_{\tilde P_+}$, on one hand, and of $\omega_0^\pm$, on the other, the crucial points being the facts that $H_{\tilde s}$ is a positive operator when restricted  to $\tilde P_+ \mathcal H_\infty$ and that we have an isomorphism  $\A_{\infty}\cong\A_{0}\otimes \A_{1}$. The computations are the same as in the  proof of Theorem 3 in \cite{araki1985ground}  and we omit the details. Finally, since we already know that $\omega_+\neq\omega_-$,  for $\lambda\in[0,1]$,  we can consider a convex linear combination   $\lambda \omega_0^++(1-\lambda)\omega_0^-$  and define $\widetilde \omega_\lambda$ as in Equation (\ref{eq:state-thm2-mixed}).
\end{enumerate}
\end{enumerate}
\end{proof}
\begin{remark}\label{rem:2} The gap closing discussed in the theorem above entails the breakdown of the bijective correspondence between the ground state of each Hamiltonian $H_{s,\infty}$ along the path $\mathbf{H}\doteq\{H_{s,\infty}\}_{s\in\cal}$ and its corresponding spectral projection $E_{+,s,\infty}$. Indeed, if the gap closes at $\tilde s\in \mathring{\cal}$, then the fact that $E_{0,\tilde s, \infty}\neq 0$ implies that there are \emph{many different} ground states associated to the Hamiltonian $H_{\tilde s,\infty}$. Let $\mathbf{\omega}\subset\states^{(\infty)}$ be the set of ground states associated with the path $\mathbf{H}\subset\BL(\H_{\infty})$. In particular, let $\omega_{0}\in\q\states^{(\infty)}$ and $\omega_{1}\in\q\states^{(\infty)}$ be the quasi--free ground states associated with the gapped Hamiltonians $H_{0,\infty}$ and $H_{1,\infty}$, respectively, such that their corresponding basis projections satisfy $\sigma (E_{+,0,\infty}, E_{+,1,\infty})=-1$. Assuming that the gap closing occurs only  at $s=\tilde s$, one might ask whether it is possible to appropriately \emph{choose} one of the ground states associated to $H_{\tilde s,\infty}$ in order to define a \emph{continuous} path $\gamma\colon \cal \rightarrow \states^{(\infty)}$ such that, for each $s\in \cal$, $\gamma(s)$ is a ground state for $H_{s,\infty}$. As shown below, a straightforward consequence of  Theorem \ref{theorem:groun_states} is that the answer to this question turns out to be negative.
\end{remark}
\begin{corollary}\label{cor:gap}
Take same assumptions of Theorem \ref{theorem:groun_states} and Remark \ref{rem:2}. Then, the path  $\gamma\colon \cal \rightarrow \states^{(\infty)}$ from Remark \ref{rem:2} is not continuous.
\end{corollary}
\begin{proof}
Under the stated assumptions, it is clear that, for $s\in \cal \setminus \lbrace \tilde s\rbrace$, $\gamma(s)=\omega_{E_{+,s,\infty}}$, the quasi--free ground state associated to the (basis) projection $E_{+,s,\infty}$. We will show that, irrespective of how $\gamma(\tilde s)$ is chosen, $\gamma$ has a discontinuity at that point.
 This is done as follows:\\
First, recall that since $\A_{\infty}$ is separable, the set $\states^{(\infty)}\subset\A_{\infty}^{*}$ is metrizable in the weak$^{*}$--topology. In fact, if $\{A_{n}\}_{n\in\N}$ is a countable set of operators on $\A_{\infty}$ separating points on $\A_{\infty}^{*}$, so that $\Vert A_{n}\Vert_{\A_{\infty}}\leq1$ for all $n\in\N$, the metric
\begin{equation}\label{eq:metric}
\mathfrak{d}_{\states^{(\infty)}}(\omega_{1},\omega_{2})\doteq\sum_{n\in\N}\frac{1}{2^{n}}\left|\omega_{1}(A_{n})-\omega_{2}(A_{n})\right|,\qquad\omega_{1},\omega_{2}\in\states^{(\infty)},
\end{equation}
induces the weak$^{*}$--topology on the set of states $\states^{(\infty)}$. For the sake of simplicity, and without confusion, the metrizable space $(\states^{(\infty)}, \mathfrak{d}_{\states^{(\infty)}})$ is denoted by $\states^{(\infty)}$. In particular, the open ball with center on $\omega\in\states^{(\infty)}$ and radius $\varepsilon\in\R^{+}$ is defined by
$$
\mathfrak{B}\left(\omega,\varepsilon\right)\doteq\left\{\omega'\in\states^{(\infty)}\colon\mathfrak{d}_{\states^{(\infty)}}(\omega,\omega')<\varepsilon\right\}\subset\states^{(\infty)}.
$$
Now, in the proof of Theorem \ref{theorem:groun_states} it was shown that $\omega_{+}\neq \omega_{-}$, where $\omega_\pm \doteq \lim\limits_{s\to \tilde s^{\pm}} \gamma(s)$.
Hence, since $\{A_{n}\}_{n\in\N}$ separates points on $\A_{\infty}^{*}$, there exists $m\in \N$ such that $\omega_+(A_m)\neq\omega_-(A_m)$, with $A_{m}\in\left\{A_{n}\right\}_{n\in\N}$. Thus, define
$$
w\doteq |\omega_+(A_m)-\omega_-(A_m)|\in\R^{+},
$$
and observe that by Expression \eqref{eq:metric} one has
\begin{align*}
\mathfrak{d}_{\states^{(\infty)}}(\omega_{+},\omega_{-})&=\sum_{n\in\N}\frac{1}{2^{n}}\left|\omega_{+}(A_{n})-\omega_{-}(A_{n})\right|\\
&\geq \frac{1}{2^m}w.
\end{align*}
Hence, one can take $\varepsilon\doteq\frac{1}{2^{m+1}}w$ so that $\mathfrak{B}\left(\omega_-,\varepsilon\right)\cap\mathfrak{B}\left(\omega_{+},\varepsilon\right) = \emptyset$.  Furthermore, given that $\omega_-$ is the left limit of $\gamma(s)$ at $\tilde s$ and $\omega_+$ its right limit, observe that it is enough to consider $\varepsilon'\in\R^{+}$, with $\varepsilon'<\varepsilon$, such that $\omega_+\not\in\mathfrak{B}\left(\omega_-,\varepsilon'\right)$ and $\mathfrak{B}\left(\omega_-,\varepsilon'\right)\cap \lbrace \gamma(s)\colon s >\tilde s\rbrace = \emptyset$. Therefore the path $\gamma\colon\cal \rightarrow \states^{(\infty)}$ is not continuous at $s=\tilde{s}$. More precisely, at that point, $\gamma$ has a discontinuity of the \emph{first kind}. See \cite[Def. 13.1]{Choquet1966}.
\end{proof}
\section{Technical Proofs}\label{sec: tech_proofs}
\subsection{Existence of the spectral flow automorphism}
\begin{lemma}\label{lemma:impor_res}
Take $\cal\equiv[0,1]$ and let $\mathbf{H}$ be a family of Hamiltonians as in Assumption \ref{assump:hamil}. For any $s\in\cal$, $E_{+,s}$ will denote the spectral projection associated to the positive part of $\spec(H_{s})$. Then, for the family of spectral projections $\{E_{+,s}\}_{s\in\cal}$, there exists a family of automorphisms $\{\kappa_{s}\}_{s\in\cal}$ on $\BL(\H)$ satisfying
\begin{align*}
\kappa_{s}\left(E_{+,s}\right)=E_{+,0}.
\end{align*}
\end{lemma}
\begin{proof}
The arguments of the proof are completely standard and we state these for the sake of completeness, c.f. \cite{kato2013perturbation,bachmann2012automorphic,nachtergaele2018quasi}. Take $\cal\equiv[0,1]$ and consider $\mathbf{H}\doteq\{H_{s}\}_{s\in\cal}\in\BL(\H)$ be a differentiable family of self--dual Hamiltonians. Fix $s\in\cal$ and let $E_{+,s}$ be the spectral projection of $H_{s}$ on $\Sigma_{s}$. Note that if the automorphism $\kappa_{s}\colon\BL(\H)\to\BL(\H)$ satisfying $\kappa_{s}\left(E_{+,s}\right)=E_{+,0}$ exists, this implies that it is unitarily implemented by a differentiable unitary operator $V_{s}$ defined by
\begin{equation}\label{eq:spectral proof1}
 \kappa_{s}\left(E_{+,s}\right)\doteq V_{s}^{*}E_{+,s}V_{s}, \quad\text{with}\quad V_{0}=\pm\1_{\H},
\end{equation}
and satisfying the differential equation
\begin{equation}\label{eq:spectral proof2}
 \p_{s}V_{s}=-\ii \mathfrak{D}_{\g,s}V_{s},
\end{equation}
where for the gap $\g$, as in Definition \ref{def: phase of the matter}, $\mathfrak{D}_{\g,s}\colon\cal\to\BL(\H)$ is a pointwise self--adjoint bounded operator. Here, $\p_{s}$ denotes the derivative with respect to $s\in\cal$. Now, for any $H_{s}$, we write its spectral projection on $\Sigma_{s}$ by
\begin{align}\label{eq: spec_pro_int}
E_{+,s}=\frac{1}{2\pi\ii}\oint_{\a_{s}}R_{\zeta}(H_{s})\d\zeta,
\end{align}
where, for any $s\in\cal$, $R_{\zeta}(H_{s})\in\BL(\H)$ is the resolvent set of $H_{s}$. In \eqref{eq: spec_pro_int}, for any $s\in\cal$, $\a_{s}$ is a  \emph{chain}, that is, $\a_{s}$ is a finite collection of closed rectifiable curves $\gamma_{s}$ in $\CP$. In particular, $\a_{s}$ surrounds $\Sigma_{s}$ and is in the complement of $\widetilde{\Sigma}_{s}$. By using the \emph{second} resolvent equation, i.e.,
\begin{equation}
\label{eq:sec_res}
R_{\zeta}(A)-R_{\zeta}(B)=R_{\zeta}(A)(B-A)R_{\zeta}(B),
\end{equation}
for any operators $A,B\in\BL(\H)$ and any $\zeta\notin\spec(A)\cap\spec(B)$, one can show that
\begin{equation}\label{eq: proj_resol}
 \p_{s}E_{+,s}=-\frac{1}{2\pi\ii}\oint_{\a_{s}}R_{\zeta}(H_{s})\left(\p_{s}H_{s}\right)R_{\zeta}(H_{s})\d\zeta,
\end{equation}
and it follows that the derivative $\p_{s}E_{+,s}$ is well--defined on $\cal$. A combination of \eqref{eq:spectral proof1}--\eqref{eq:spectral proof2} and $\kappa_{s}\left(E_{+,s}\right)=E_{+,0}$ yield us to
\begin{equation}\label{eq:spectral proof3}
\p_{s}E_{+,s}=-\ii[\mathfrak{D}_{\g,s},E_{+,s}].
\end{equation}
Additionally, since for any $s\in\cal$, $E_{+,s}$ is an orthogonal projection then
$$
E_{+,s}^{\perp}\left(\p_{s}E_{+,s}\right)E_{+,s}^{\perp}=E_{+,s}\left(\p_{s}E_{+,s}\right)E_{+,s}=0,
$$
where for any $s\in\cal$, $E_{+,s}^{\perp}$ denotes the orthogonal \emph{complement} of $E_{+,s}$, i.e., $E_{+,s}^{\perp}\doteq1-E_{+,s}$. From the latter identity we get the following one
$$
\p_{s}E_{+,s}=E_{+,s}\left(\p_{s}E_{+,s}\right)E_{+,s}^{\perp}+E_{+,s}^{\perp}\left(\p_{s}E_{+,s}\right)E_{+,s},
$$
and together with \eqref{eq: proj_resol} and the fact that $E_{+,s}, E_{-,s}$ are basis projections, see \eqref{eq:sum_basis_proj}, we arrive at
\begin{equation}\label{eq:projec_deriv}
\p_{s}E_{+,s}=-\frac{1}{\pi\ii}\Ree\left(\oint_{\a_{s}}\left(E_{+,s}R_{\zeta}(H_{s})\left(\p_{s}H_{s}\right)R_{\zeta}(H_{s})E_{-,s}\right)\d\zeta\right).
\end{equation}
Here, the self--adjoint operator $\Ree(A)\in\BL(\H)$ is the \emph{real part} of $A\in\BL(\H)$, given by $\Ree\left(A\right)\doteq \frac{1}{2}\left(A+A^{*}\right)$. Similarly, $\Imm(A)\in\BL(\H)$, the imaginary part of $A$, is the self--adjoint operator usually defined by $\Imm\left(A\right)\doteq \frac{1}{2\ii}\left(A-A^{*}\right)$.\\
Then, the existence of the automorphism $\kappa_{s}$ is equivalent to finding the operator $\mathfrak{D}_{\g,s}$ such that \eqref{eq:spectral proof3} and \eqref{eq:projec_deriv} are satisfied. This is precisely that is done in \cite{bachmann2012automorphic}, and in the present context we explicitly write $\p_{s}E_{+,s}$ as
$$
\p_{s}E_{+,s}=\int_{\Sigma_{s}}\int_{\Sigma_{s}}\frac{2}{\mu+\lambda}\Ree\left(\d E_{\mu,+,s}\left(\p_{s}H_{s}\right)\d E_{-\lambda,+,s}\right),
$$
where for any $s\in\cal$, $E_{\mu,+,s}$ is a \emph{resolution of the identity} supported on the positive (negative) part of $\spec(H_{s})$, i.e.,
\begin{equation}\label{eq:resol_iden}
E_{\pm,s}\doteq\int_{\pm\Sigma_{s}}\d E_{\lambda,\pm,s}.
\end{equation}
The next step is to verify that the self--adjoint bounded operator
\begin{equation}\label{eq: operator D}
\mathfrak{D}_{\g,s}\doteq \int_{\R}\e^{\ii tH_{s}}\left(\p_{s}H_{s}\right)\e^{-\ii tH_{s}}\mathfrak{W}_{\g}(t)\d t,\quad\text{for any}\quad s\in\cal,
\end{equation}
satisfies \eqref{eq:spectral proof3} and \eqref{eq:projec_deriv}. Here, $\mathfrak{W}_{\g}(t)\colon\R\to\R$ is an odd function on $L^{1}(\R)$ such that its Fourier transform, $\widehat{\mathfrak{W}}_{\g}\colon\R\to\R$ is given for $\mu\neq0$ by
$$
\widehat{\mathfrak{W}}_{\g}(\mu)\equiv-\frac{1}{\sqrt{2\pi}\mu}.
$$
For a complete description of the properties of $\mathfrak{W}_{\g}$ see \cite{bachmann2012automorphic,michalakis2013stability,nachtergaele2018quasi}. We now note that for any operator $B\in\BL(\H)$ and any orthogonal projection $P\in\BL(\H)$ we get
$$
-\ii[B,P]=\ii\left(PB(1-P)-(1-P)BP\right).
$$
In particular, note that by taking $B$ as $\mathfrak{D}_{\g,s}$ and $P$ as the spectral projection of $H_{s}$ on $\Sigma_{s}$, i.e., $E_{+,s}$, we have
\begin{eqnarray*}
-\ii[\mathfrak{D}_{\g,s},E_{+,s}]&=&\ii\int_{\R}\int_{\Sigma_{s}}\int_{-\Sigma_{s}}\e^{\ii t(\mu-\lambda)}\d E_{\mu,+,s}\left(\p_{s}H_{s}\right)\d E_{\lambda,-,s}\mathfrak{W}_{\g}(t)\d\lambda\d\mu\d t\\
&&-\ii\int_{\R}\int_{-\Sigma_{s}}\int_{\Sigma_{s}}\e^{\ii t(\lambda-\mu)}\d E_{\lambda,-,s}\left(\p_{s}H_{s}\right)\d E_{\mu,+,s}\mathfrak{W}_{\g}(t)\d\mu\d\lambda\d t\\
&=&
\ii\sqrt{2\pi}\int_{\Sigma_{s}}\int_{-\Sigma_{s}}\d E_{\mu,+,s}\left(\p_{s}H_{s}\right)\d E_{\lambda,-,s}\widehat{\mathfrak{W}}_{\g}(\lambda-\mu)\d\lambda\d\mu\\
& &-\ii\sqrt{2\pi}\int_{-\Sigma_{s}}\int_{\Sigma_{s}}\d E_{\lambda,-,s}\left(\p_{s}H_{s}\right)\d E_{\mu,+,s}\widehat{\mathfrak{W}}_{\g}(\mu-\lambda)\d\mu\d\lambda\\
&=&\int_{\Sigma_{s}}\int_{\Sigma_{s}}\frac{2}{\mu+\lambda}\Ree\left(\d E_{\mu,+,s}\left(\p_{s}H_{s}\right)\d E_{-\lambda,+,s}\right)
\end{eqnarray*}
where we have used \eqref{eq:resol_iden} and that $\widehat{\mathfrak{W}}_{\g}$ is an odd function.
\end{proof}
In particular, the unitary operator $V_{s}$ satisfying the differential equation \eqref{eq:spectral proof2} commutes with the involution $\a$, i.e., $\a V_{s}=V_{s}\a$. In fact, for any $s\in\cal$, let $C_{s}\in\BL(\H)$ be defined by
$C_{s}\doteq[\a,V_{s}]$, such that $C_{s}^{*}=[V_{s}^{*},\a]$. We would like to show that $C_{s}=0$. To do this, observe that the self--adjoint bounded operator $\mathfrak{D}_{\g,s}$ given by Expression \eqref{eq: operator D}, commutes with $\a$. Using \eqref{eq:spectral proof2}, after some calculations we have
$$
\p_{s}C_{s}=-\ii\mathfrak{D}_{\g,s}C_{s}\qquad\text{and}\qquad\p_{s}C_{s}^{*}=-\ii C_{s}^{*}\mathfrak{D}_{\g,s}.
$$
From the left hand side equation one has $\p_{s}C_{s}^{*}=\ii C_{s}^{*}\mathfrak{D}_{\g,s}$, which comparing with the right hand side equation, we obtain $C_{s}=0$. We have proven:
\begin{corollary}[Bogoliubov transformation]\label{corol:Vs}
For any $s\in\cal\equiv[0,1]$, the unitary operator $V_{s}$ satisfying the differential equation \eqref{eq:spectral proof2} commutes with the involution $\a$, i.e., $\a V_{s}=V_{s}\a$, then $V_{s}$ is a Bogoliubov transformation, see \eqref{Bogoliubov  automorphism}.
\end{corollary}
One primary consequence of Lemma \ref{lemma:impor_res} and Corollary \ref{corol:Vs} is the existence of a strongly continuous family of one--parameter (Bogoliubov) group $\Upsilon_{s}\doteq\left\{\Upsilon_{s}\right\}_{s\in\cal\in\R}$ of
$^{*}$--automorphisms of $\sCAR$, implemented by the Bogoliubov automorphisms $V_{s}$. To be precise, for the one--parameter unitary group $\{V_{s}\}_{s\in\cal}$ implementing the family of automorphism $\left\{\kappa_{s}\right\}_{s\in\cal}$ of Lemma \ref{lemma:impor_res} over the family of spectral projections $\{E_{+,s}\}_{s\in\cal}$ we are able to show that, for any $s$, the (Bogoliubov) $^{*}$--automorphism
\begin{equation}\label{eq:autom}
\Upsilon_{s}\left(\B\left(\varphi\right)\right)\equiv\chi_{V_{s}^{*}}\left(\B\left(\varphi\right)\right)=\B\left(V_{s}^{*}\varphi\right),
\end{equation}
exists, for any $\varphi\in\H$. The latter can be easily verified using bilinear elements, which are described in Definition \ref{def trace state copy(1)}. More generally, for any family of self--dual Hamiltonians $\{H_{s}\}_{s\in\cal}\in\BL(\H)$ as in Assumption \ref{assump:hamil}, we have an associated family of bilinear elements $\{\Bil{H_{s}}\}_{s\in\cal}\in\sCAR$ given by,
$$
\Upsilon_{s}\left(\Bil{H_{s}}\right)=\Bil{H_{0}},\qquad\text{for any}\qquad s\in\cal,
$$
(see Definition \ref{def trace state copy(1)}).\\
As stressed in comments around Expression \eqref{eq:top_index}, for any pair $P_{1},P_{2}\in\fp$ there exists a Bogoliubov transformation $U$ relating both, i.e., a unitary operator $U\in\BL(\H)$ so that $P_{2}=U^{*}P_{1}U$, with $U\a=\a U$. Thus, if $\det(U)=1$ ($\det(U)=-1$) we say that $U$ is in the positive (negative) connected component. Following \cite[Theo. 6.30 and Lemma 7.17]{EK98} for the special class of $U$ satisfying (i) $U\a=\a U$ and (ii) $\1_{\H}-U$ trace class, the topological index $\sigma(P,U^{*}PU)$ coincides with $\det(U)$. Note that Lemma \ref{lemma:impor_res} tells us about the existence of a family of unitary operators $\{V_{s}\}_{s\in\cal}$ which implements the family of automorphisms $\{\kappa_{s}\}_{s\in\cal}$ on $\BL(\H)$. However, we need to specify with which kind of Hamiltonians we are dealing. A wide class of fermion systems are those satisfying Lemma \ref{theor:unit_oper} and Proposition \ref{Combes-Thomas} below. More concretely, our results will permit to consider disordered fermions systems in which the spectral gap does not close. Note that a suitable control of the properties of $\{V_{s}\}_{s\in\cal}$ is closely related to the recently results found by Hastings in \cite{hastings2019stability}. Then, as already mentioned we invoke \cite[Theo. 6.30 and Lemma 7.17]{EK98} in order to distinguish different physical systems (see Definition \ref{def: phase of the matter}), which are classified by two components even in the interacting setting, see \cite{nachtergaele2018lieb}. In particular, we have:
\begin{corollary}\label{corol:sign}
Consider a family of self--dual Hamiltonians $\mathbf{H}\in\BL(\H)$ satisfying Assumption \ref{assump:hamil}. For any $s\in\cal\equiv[0,1]$, take the Bogoliubov transformation $V_{s}$ of Corollary \ref{corol:Vs}. Assume that for any $s\in\cal$, $\1_{\H}-V_{s}$ and $\mathfrak{D}_{\g,s}\in\BL(\H)$ are trace class, with $\mathfrak{D}_{\g,s}$ given by \eqref{eq:spectral proof2}. Then $\det\left(V_{s}\right)=\det\left(V_{0}\right)$.
\end{corollary}
\begin{proof}
Let $V_{0}=\pm\1_{\H}$, write $\det\left(V_{s}\right)-\det\left(V_{0}\right)=\int_{0}^{s}\p_{r}\left(\det\left(V_{r}\right)\right)\d r$, and apply the \emph{Jacobi's formula} of determinants for $V_{r}$: $\p_{r}\left(\det\left(V_{r}\right)\right)=\det\left(V_{r}\right)\tr_{\H}\left(V_{r}^{*}\left(\p_{r}V_{r}\right)\right)$ such that:
$$
\det\left(V_{s}\right)-\det\left(V_{0}\right)=-\ii\int_{0}^{s}\det\left(V_{r}\right)\tr_{\H}\left(\mathfrak{D}_{\g,s}\right)\d r,
$$
where we have used the differential equation \eqref{eq:spectral proof2} and the cyclic property of the trace. Since is $\mathfrak{D}_{\g,s}$ is trace class one can write
$$
\det\left(V_{s}\right)-\det\left(V_{0}\right)=-\ii\int_{0}^{s}\int_{\R}\det\left(V_{r}\right)\mathfrak{W}_{\g}(t)\tr_{\H}\left(\p_{s}H_{s}\right)\d t\d r,
$$
Now, by using that for each $s\in\cal$, $H_{s}$ is self--dual Hamiltonian, i.e.,  $H_{s}=-\a H_{s}\a$ we note that $\p_{s}H_{s}$ is also self--dual Hamiltonian, i.e.,  $\p_{s}H_{s}=-\a(\p_{s}H_{s})\a$. In particular, $\tr_{\H}(\p_{s}H_{s})=0$ (see Definition \ref{def one particle hamiltinian}), and the assertion follows.
\end{proof}
From now on, we will expose some issues about quasi--free ground states for $\g\in\R^{+}$, as in Definition \ref{def: phase of the matter}, and for $\g=0$. The former are called \emph{gapped quasi--free ground states} (see Expression \eqref{eq:GGS} and comments around it), and their uniqueness is guaranteed. Since the set $\states$ of ground states is metrizable in the weak$^{*}$--topology, we denote by $\states_{\g}\equiv(\states_{\g},\mathfrak{d}_{\g})$ and $\states_{0}\equiv(\states_{0},\mathfrak{d}_{0})$ the metric spaces in the weak$^{*}$--topology related to the quasi--free ground states for $\g\in\R^{+}$ and $\g=0$ respectively. In particular, one notes that $\states_{\g}$ and $\states_{0}$ are not homeomorphic since, as we will see in Corollary \ref{theorem:states_algebras}, the representations associated to $\states_{\g}$ are reducible whereas those associated to $\states_{0}$ are not. This is clear from the fact that there is no homeomorphism between one connected metric space and another one that is \emph{disconnected}\footnote{In particular $\states_{\g}$ is a weak$^{*}$--compact convex set metrizable in the weak$^{*}$--topology that can be written as $\states_{\g}=\states_{\g,\,-}\cup\states_{\g,\,+}$, for $\states_{\g,\,-}$ and $\states_{\g,\,+}$ nonempty and disjoint metrizable set in the weak$^{*}$--topology. Here, $\states_{\g,\,-}$ and $\states_{\g,\,+}$ are associated to the negative and positive components of the unitary operators respectively.}. Then the representations associated to $\states_{\g}$ and $\states_{0}$ are not physically equivalent as the intuition says. Instead, Corollary \ref{theorem:states_algebras} below claims that any two gapped quasi--free ground states associated to the quasi--free dynamics of two gapped Hamiltonians on the same $\g$--phase are unitarily equivalent. Thus, their irreducible representations also are.\\
In order to prove last statement, recall Expressions \eqref{symbol}, \eqref{symbolbis} and \eqref{eq:symbolbis2}, where for any $s\in\cal\equiv[0,1]$ one knows that for the positive (on $\Sigma_{s}$) spectral projection $E_{+,s}\in\BL(\H)$ associated to the self--dual Hamiltonian $H_{s}$ on $(\H,\a)$ there is a unique quasi--free ground state $\omega_{s}\in\states_{\g}$ such that
$$
\omega_{s}\left(\B\left(\varphi_{1}\right)\B\left(\varphi_{2}\right)^{*}\right)=\inner{\varphi_{1},E_{+,s}\varphi_{2}}_{\H},\qquad\varphi_{1},\varphi_{2}\in\H.
$$
Once again, consider $\mathbf{H}$, a family of self--dual Hamiltonians satisfying Assumption \ref{assump:hamil}. For any $s\in\cal$, $\omega_{s}$ is a gapped quasi--free ground state. By using the family of $^{*}$--automorphisms $\left\{\kappa_{s}\right\}_{s\in\cal}$ on $\BL(\H)$ of Lemma \ref{lemma:impor_res}, with $V_{s}$ a unitary operator implementing $\kappa_{s}$, we note that
\begin{equation}\label{eq:ground_relat}
 \omega_{s}=\omega_{0}\circ\Upsilon_{s}, \qquad s\in\cal.
\end{equation}
Here, $\Upsilon_{s}$ is the one--parameter (Bogoliubov) $^{*}$--automorphism of $\sCAR$ given by Expression \eqref{eq:autom}. Additionally, let $\mathbf{\omega}\doteq\{\omega_{s}\}_{s\in\cal}\subset\q\states$ be a family of gapped quasi--free ground states associated to self--dual Hamiltonians $\mathbf{H}$ on some self--dual Hilbert space $(\H,\a)$, with the same assumptions of Lemma \ref{lemma:impor_res}. The meaning of expression \eqref{eq:ground_relat} in terms of representations is that the associated (irreducible) GNS representation $(\H_{\mathbf{\omega}},\pi_{\mathbf{\omega}},\Omega_{\mathbf{\omega}})$ is unique (up to unitary equivalence): for all $A\in\sCAR$
$$
\mathbf{\omega}(A)=\inner{\Omega_{\mathbf{\omega}},\pi_{\mathbf{\omega}}(A)\Omega_{\mathbf{\omega}}}_{\H_{\mathbf{\omega}}},
$$
where the latter notation means that for any two states $\omega_{s_{1}},\omega_{s_{2}}\in\mathbf{\omega}$ there exists an isomorphism $\mathfrak{I}_{s_{1},s_{2}}$ from $\H_{\omega_{s_{1}}}$ to $\H_{\omega_{s_{2}}}$ satisfying
$$
\pi_{\omega_{s_{2}}}(A)=\mathfrak{I}_{s_{1},s_{2}}^{*}\pi_{\omega_{s_{1}}}(A)\mathfrak{I}_{s_{1},s_{2}},
$$
i.e., $\pi_{\omega_{s_{2}}}$ and $\pi_{\omega_{s_{2}}}$ are unitarily equivalent as well as their associated cyclic vectors $\Omega_{\omega_{s_{1}}}$ and $\Omega_{\omega_{s_{2}}}$. Additionally, following Definition \ref{def: ground state} and comments around it, there is a strongly continuous one--parameter unitary group $\left(\e^{\ii t\mathcal{L}_{\mathbf{\omega}}}\right)_{t\in\R}$ with generator $\mathcal{L}_{\mathbf{\omega}}=\mathcal{L}_{\mathbf{\omega}}^{*}\geq0$ satisfying $\e^{\ii t\mathcal{L}_{\mathbf{\omega}}}\Omega_{\mathbf{\omega} }=\Omega _{\mathbf{\omega}}$ such that for $t\in\R$, $\e^{\ii t\mathcal{L}_{\mathbf{\omega}}}\in\pi_{\mathbf{\omega}}(\sCAR)''$ and any $A\in\sCAR$
\begin{equation*}
\e^{\ii t\mathcal{L}_{\mathbf{\omega}}}\pi_{\mathbf{\omega}}(A)\Omega_{\mathbf{\omega}}=\pi_{\mathbf{\omega}}\left(\tau_{t}(A)\right)\Omega_{\mathbf{\omega}}.
\end{equation*}
We summarize the latter with the following Corollary:
\begin{corollary}\label{theorem:states_algebras}
Consider a family of self--dual Hamiltonians $\mathbf{H}\subset\BL(\H)$ satisfying Assumption \ref{assump:hamil}. Let $\mathbf{\omega}\in\q\states$ be a family of gapped quasi--free ground states associated to $\mathbf{H}$. The associated (irreducible) GNS representation $(\H_{\mathbf{\omega}},\pi_{\mathbf{\omega}},\Omega_{\mathbf{\omega}})$ is unique (up to unitary equivalence). In particular, any state $\omega_{s}\in\mathbf{\omega}$, $s\in\cal$, is related to $\omega_{0}\in\mathbf{\omega}$ by Expression \eqref{eq:ground_relat}, namely, $\omega_{s}=\omega_{0}\circ\Upsilon_{s}$.
\end{corollary}
\subsection{Dynamics, ground states and spectral flow automorphism in the Thermodynamic limit}\label{subsec: lattice fermion}
For $d\in\N$, let $\Z^{d}$ be the Cayley graph as defined in Appendix \ref{appendix: GAM}, see Expression \eqref{eq: Zd}, and let the \emph{spin} set $\spin$, such that $\L\doteq\Z^{d}\times\spin$. Since we are dealing with fermions, w.l.o.g., these can be treated as \emph{negatively} charged particles. The cases of particles positively charged can be treated by exactly the same methods. Then, in order to \emph{take} the \emph{thermodynamic limit} we define the Hilbert spaces $\H_{\spin}\doteq\ell^{2}\left(\spin\right)\oplus\ell^{2}\left(\spin\right)^{*}$ and $\H_{L}\doteq\ell^{2}\left(\Lambda _{L};\H_{\spin}\right)$ for all $L\in\R_{0}^{+}\cup\{\infty\}$, where $\Lambda_{L}$ for $L\in\R_{0}^{+}\cup\{\infty\}$ is defined by the increasing sequence of cubic boxes
\begin{equation}\label{eq:boxesl}
\Lambda _{L}\doteq\{(x_{1},\ldots,x_{d})\in\Z^{d}\colon|x_{1}|,\ldots,|x_{d}|\leq L\}\in\Pfz,
\end{equation}
of side length $\mathcal{O}(L)$. Note that such a sequence is a ``Van Hove net''\index{Van Hove net}, i.e., the volume of the boundaries\footnote{By fixing $m\geq1$, the boundary $\partial \Lambda$ of any $\Lambda\subset\Z^{d}$ is defined by $\partial\Lambda\doteq\{x\in \Lambda\;\colon\;\exists y\in\Z^{d}\backslash\Lambda \text{ with }d_{\epsilon}(x,y)\leq m\},$ where for $\epsilon\in(0,1]$, $d_{\epsilon}(x,y)\colon\Z^{d}\times\Z^{d}\to[0,\infty)$ is a well--defined pseudometric related to the distance between $x,y$ in the lattice $\Z^{d}$ \cite{BruPedra2}. W.l.o.g. we will take the $\epsilon$--Euclidean distance $d_{\epsilon}(x,y)\doteq|x-y|^{\epsilon}$.} $\partial \Lambda _{L}\subset \Lambda _{L}\in\Pfz$ is negligible w.r.t. the volume of $\Lambda_{L}$ for $L$ large enough: $\lim\limits_{L\to\infty}\frac{|\partial\Lambda_{L}|}{|\Lambda_{L}|}=0$.\\
We now fix any antiunitary involution $\a_{\spin}$ on $\H_{\spin}$. For any $L\in \R_{0}^{+}\cup \{\infty \}$, we define an antiunitary involution $\a_{L}$ on $\H_{L}$ by
\begin{equation}
\left( \a_{L}\varphi \right)\left(x\right)\doteq\a_{\spin}\left( \varphi \left( x\right) \right),\qquad x\in \Lambda_{L},\ \varphi \in \H_{L}\ .  \label{involution}
\end{equation}
Then, $(\H_{L},\a_{L})$ is a \emph{local} self--dual Hilbert space for any $L\in \R_{0}^{+}\cup \{\infty \}$. Note that $\H_{\spin}$ and $\H_{L}$ are finite--dimensional, with even dimension, whenever $L<\infty$: Let
\begin{equation}\label{eq: X_L}
\mathbb{X}_{L}\mathbb{\doteq }\Lambda _{L}\times \spin\times \{+,-\},\qquad L\in \R_{0}^{+}\cup\{\infty\}.
\end{equation}
The canonical orthonormal basis $\left\{\E_{\x}\right\}_{\x\in \mathbb{X}_{L}}$ of $\H_{L}$, $L\in \R_{0}^{+}\cup \{\infty \}$, now is defined by
\begin{equation}
\E_{\x}(y)\doteq \delta_{x,y}\mathfrak{f}_{\mathfrak{s},v},\qquad \x=(x,\mathfrak{s},v)\in\mathbb{X}_{L},\quad y\in\Lambda_{L},  \label{canonical onb1}
\end{equation}
where $\mathfrak{f}_{\mathfrak{s},+}\doteq \a_{\spin}\mathfrak{f}_{\mathfrak{s},-}\in\H_{\spin}$ and $\mathfrak{f}_{\mathfrak{s},-}(\mathrm{t})\doteq \delta _{\mathfrak{s},\mathrm{t}}$ for any $\mathfrak{s},\mathrm{t}\in \spin$.\\
Within the self--dual formalism, a lattice fermion system in infinite volume is defined by a self--dual Hamiltonian $H_{\infty }\in \mathcal{B}(\H_{\infty })$ on $(\H_{\infty },\a_{\infty})$, that is, $H_{\infty}=H_{\infty }^{*}=
-\a_{\infty}H_{\infty}\a_{\infty}$. See Definition \ref{def one particle hamiltinian} which is here extended to the infinite--dimensional case. For a fixed basis projection $P_{\infty}$ diagonalizing $H_{\infty }$, the operator $P_{\infty}H_{\infty}P_{\infty}$ is the so--called one--particle Hamiltonian associated with the system. To obtain the corresponding self--dual Hamiltonians in finite volume we use the orthogonal projector $P_{\H_{L}}\in \mathcal{B}(\H_{\infty })$ on $\H_{L}$ and define
\begin{equation}
H_{L}\doteq P_{\H_{L}}H_{\infty}P_{\H_{L}},\qquad L\in\R_{0}^{+}.  \label{definition finite volume hamiltonian}
\end{equation}
By construction, if $H_{\infty}$ is a self--dual Hamiltonian on $(\H_{\infty },\a_{\infty })$, then, for any $L\in\R_{0}^{+}$, $H_{L}$ is a self--dual Hamiltonian on $(\H_{L},\a_{L})$. Note that $P_{\H_{L}}$ strongly converges to $\1_{\H_{\infty}}$ as $L\to\infty$.\\
For the self--dual Hilbert space $(\H_{\infty},\a_{\infty})$, the self--dual $\CAR$ algebra associated is denoted by $\A_{\infty}\doteq\mathrm{sCAR}(\H_{\infty},\a_{\infty})$, with generator elements $\mathfrak{1}$ and $\{\B(\E_{\x})\}_{\x\in\mathbb{X}_{\infty}}$ satisfying $\CAR$ Expressions of Definition \ref{def Self--dual CAR Algebras}. The subalgebra of even elements of $\A_{\infty}$ (see \eqref{eq:even odd}) will be denoted by $\A_{\infty}^{+}$ in the sequel. For $\Lambda\in\Pfz$ and the finite--dimensional (one--particle) Hilbert space $\H_{\Lambda}\doteq\ell^{2}\left(\Lambda;\H_{\spin}\right)$ with involution given by \eqref{involution}, we identify the finite dimensional $\CAR$ $\C$--algebra
\begin{equation}\label{eq: local}
\A_{\Lambda}\doteq\mathrm{sCAR}(\H_{\Lambda},\a_{\Lambda}),\qquad \Lambda \in \mathcal{P}_{\text{f}}(\Z^{d}),
\end{equation}
with the $\C$--subalgebra generated by the unit $\mathfrak{1}$ and $\{\B(\E_{\x})\}_{\x\in\mathbb{X}_{\Lambda}}$. Then, we define by
\begin{equation}\label{local elements}
\A_{\infty}^{(0)}\doteq \underset{\Lambda \in \mathcal{P}_{\text{f}}(\Z^{d})}{\bigcup }\A_{\Lambda }\subset\A_{\infty},
\end{equation}
the normed $^{*}$--algebra of local elements, which is dense in $\A_{\infty}$.\\
From Definition \ref{def: ground state} one notes that existence of ground states strongly relies on the existence of the dynamics in the thermodynamical limit. The latter means that the sequence $\left\{\Lambda_{L}\right\}_{L\in\R_{0}^{+}\cup\{\infty\}}$, defined by \eqref{eq:boxesl}, eventually will contain all the finite subsets, $\Pfz$ of $\Z^{d}$ as $L\to\infty$. In fact, for any $H_{L}=H_{L}^{*}\in\BL(\H_{L})$ one can associate a quasi--free dynamics \eqref{eq:autoSCAR} defining a continuous group $\{\tau_{t}^{(L)}\}_{t\in {\R},L\in\R_{0}^{+}}$ of \emph{finite volume} $^{*}$--automorphisms of $\A_{L}\equiv\A_{\Lambda_{L}}$ by
$$
\tau_{t}^{(L)}(A)\doteq\e^{-\ii t\Bil{H_{L}}}A\e^{\ii t\Bil{H_{L}}},\qquad A\in \A_{\infty},\ t\in\R.
$$
See \eqref{definition finite volume hamiltonian} and \eqref{local elements}. The associated \emph{finite volume} generator or \emph{finite symmetric derivation} is given by \eqref{eq:gener_sCAR}, namely,
\begin{equation}\label{eq:fin_gen}
\delta^{(L)}(A)=-\ii[\Bil{H_{L}},A],\qquad A\in\A_{\infty}^{(0)},
\end{equation}
while, the \emph{infinite volume} generator or \emph{symmetric derivation} is
\begin{equation}\label{eq:infin_gen}
\delta(A)=-\ii[\Bil{H_{\infty}},A],\qquad A\in\A_{\infty}^{(0)}.
\end{equation}
For $L\in\R_{0}^{+}$ and $\Lambda_{L}\in\Pfz$, denote by $\Lambda_{L}^{\text{c}}\equiv\Z^{d}\setminus\Lambda_{L}$ the complement of $\Lambda_{L}$. Then, $\A_{\Lambda^{\text{c}}}\doteq\mathrm{sCAR}(\H_{\Lambda^{\text{c}}},\a_{\Lambda^{\text{c}}})$, will be the $\C$--subalgebra generated by the unit $\mathfrak{1}$ and $\{\B(\E_{\x})\}_{\x\in\mathbb{X}_{\Lambda_{L}^{\text{c}}}}$. The bilinear elements associated to the (border) terms on $\Lambda_{L}$ and $\Lambda_{L}^{\text{c}}$ are (cf. Definition \ref{def trace state copy(1)}):
\begin{align*}
\Bil{\p H_{L}\,}=\sum\limits_{\x_{1},\x_{2}\in \mathbb{X}_{\infty}}\inner{\E_{\x_{2}},\p H_{L}\,\E_{\x_{1}}}_{\H_{\infty}}\B\left(\E_{\x_{1}}\right)\B\left(\E_{\x_{2}}\right)^{*},
\end{align*}
with $\H_{L}^{\text{c}}\equiv\H_{\Lambda_{L}^{\text{c}}}$ and
\begin{equation}\label{eq:border_terms}
\p H_{L}\doteq P_{\H_{L}}H_{\infty}P_{\H_{L}^{\text{c}}}+P_{\H_{L}^{\text{c}}}H_{\infty}P_{\H_{L}},
\end{equation}
where for any $\Lambda_{L}\in\Pfz$, $P_{\H_{L}}\in\BL(\H_{\infty})$ is the orthogonal projector on $\H_{L}$, see Expression \eqref{definition finite volume hamiltonian}.
\begin{theorem}[Infinite volume dynamics]\label{theo:inf_dyn}
Assume that the sequence $\{H_{L}\}_{L\in\R_{0}^{+}}$ of self--dual Hamiltonians $H_{L}\in\BL(\H_{L})$ strongly converges to $H_{\infty}\in\BL(\H_{\infty})$ so that
$$
\sum_{\x\in\mathbb{X}_{\infty}}\left|\inner{\E_{0},H_{\infty}\E_{\x}}_{\H_{\infty}}\right|\in\R_{0}^{+}.
$$
Then, for $L\in\R_{0}^{+}$, the continuous group $\{\tau_{t}^{(L)}\}_{t\in {\R}}$ with generator $\delta^{(L)}$ converges strongly to a continuous group $\{\tau_{t}\}_{t\in {\R}}$ with generator $\delta$ as $L\to\infty$. $\delta$ is a conservative closed symmetric derivation.
\end{theorem}
\begin{proof} The proof of the statements is completely standard. We present it here for the sake of completeness. This is split in a set of parts:
\begin{enumerate}
\item We can combine Expressions \eqref{eq:autoSCAR} and \eqref{eq:dyn_bogou} such that for any self--dual Hamiltonian $H_{\infty}\in\BL(\H_{\infty})$ we have
$$
\tau_{t}^{(L)}\left(\B(\varphi)\right)=\B\left(\left(U_{t}^{(L)}\right)^{*}\varphi\right)\qquad\text{and}\qquad\tau_{t}\left(\B(\varphi)\right)=\B\left(U_{t}^{*}\varphi\right).
$$
Here, for $L\in\R_{0}^{+}$, $\tau_{t}^{(L)}\doteq\chi_{\e^{\ii t H_{L}}}$ and $\tau_{t}\doteq\chi_{\e^{\ii t H_{\infty}}}$ so that
$$\left\{U_{t}^{(L)}\doteq\e^{\ii tH_{L}}\right\}_{t\in\R}\quad\text{and}\quad\left\{U_{t}\equiv U_{t}^{(\infty)}\doteq\e^{\ii tH_{\infty}}\right\}_{t\in\R}
$$
are the \emph{one--parameter unitary groups} on $(\H_{\infty},\a_{\infty})$ associated to the \emph{finite} and \emph{infinite} dynamical systems, respectively. Note that for any $\varphi\in\H_{\infty}$, $\B(\varphi)$ is bounded (see Definition \ref{def Self--dual CAR Algebras}). Then, using that $\Vert\B(\varphi)\Vert_{\A_{\infty}}\leq\Vert\varphi\Vert_{\H_{\infty}}$, for any $L_{1},L_{2}\in\R_{0}^{+}$, with $L_{2}\geq L_{1}$, we have
$$
\left\Vert\tau_{t}^{(L_{2})}\left(\B(\varphi)\right)-\tau_{t}^{(L_{1})}\left(\B(\varphi)\right)\right\Vert_{\A_{\infty}}\leq\left\Vert\left(U_{t}^{(L_{2})}-U_{t}^{(L_{1})}\right)\varphi\right\Vert_{\BL(\H_{\infty})}.
$$
We can write
$$
U_{t}^{(L_{2})}-U_{t}^{(L_{1})}=\int_{0}^{t}\p_{s}\left(U_{t-s}^{(L_{1})}U_{s}^{(L_{2})}\right)\d s=\ii\int_{0}^{t}U_{t-s}^{(L_{1})}\left(H_{L_{2}}-H_{L_{1}}\right)U_{s}^{(L_{2})}\d s,
$$
so that
$$
\left\Vert\tau_{t}^{(L_{2})}\left(\B(\varphi)\right)-\tau_{t}^{(L_{1})}\left(\B(\varphi)\right)\right\Vert_{\A_{\infty}}\leq|t|\left\Vert\left(H_{L_{2}}-H_{L_{1}}\right)\varphi\right\Vert_{\BL(\H_{\infty})}.
$$
Since the sequence $\{H_{L}\}_{L\in\R_{0}^{+}}$ strongly converges to $H_{\infty}$ as $L\to\infty$, the last expression shows that it is a Cauchy sequence of self--adjoint operators. Therefore, the continuous group of $^{*}$--automorphisms $\left\{\tau_{t}^{(L)}\right\}_{t\in\R}$, $L\in\R_{0}^{+}$, strongly converges to $\left\{\tau_{t}\right\}_{t\in\R}$ for all $t\in\R$.
\item In order to show the existence of the generator, take $\Lambda\in\Pfz$ and $A\in\A_{\Lambda}$ in \eqref{eq:fin_gen}. By the estimate
\begin{align*}
\left\Vert\delta(A)\right\Vert_{\A_{\infty}}&\leq2\left\Vert A\right\Vert_{\A}\sum_{\x_{1}\in\mathbb{X}_{\infty},\x_{2}\in\Lambda}\left|\inner{\E_{\x_{2}},H_{\infty}\E_{\x_{1}}}_{\H_{\infty}}\right|\\
&\leq2\left\Vert A\right\Vert_{\A}\left|\Lambda\right|\sum_{\x_{1}\in\mathbb{X}_{\infty}}\left|\inner{\E_{0},H_{\infty}\E_{\x_{1}}}_{\H_{\infty}}\right|\in\R_{0}^{+},
\end{align*}
and the hypothesis of the Theorem note that the infinite volume generator, given by Equation \eqref{eq:infin_gen}, is absolutely convergent. By \eqref{eq:fin_gen}, it follows that
$$
\delta(A)=\lim_{L\to\infty}\delta^{(L)}(A), \qquad A\in\A_{\infty}^{(0)}.
$$
\item Moreover, for any fixed $\zeta\not\in\spec(\delta^{(L)})$, the resolvent $R_{\zeta}\left(\delta^{(L)}\right)\doteq\left(\zeta\mathfrak{1}-\delta^{(L)}\right)^{-1}$ of $ \delta^{(L)}$ converges strongly to $R_{\zeta}\left(\delta\right)\doteq\left(\zeta\mathfrak{1}-\delta\right)^{-1}$, the resolvent of $\delta$. In fact, take $L\in\R_{0}^{+}$ and let $\rho(\delta^{(L)})\doteq\CP\setminus\spec(\delta^{(L)})$ be the \emph{resolvent set} of $\delta^{(L)}$. We know that for any $\zeta\in\rho(\delta^{(L)})$ the resolvent $R_{\zeta}\left(\delta^{(L)}\right)\doteq\left(\zeta\mathfrak{1}-\delta^{(L)}\right)^{-1}$ satisfies the identity \cite{EngelNagel}
$$
R_{\zeta}\left(\delta^{(L)}\right)(A)=\int_{0}^{\infty}\e^{-\zeta t}\tau_{t}^{(L)}(A)\d t,\qquad A\in\A_{\infty}^{(0)}.
$$
Note that for any generators $\delta^{(L_{1})},\delta^{(L_{2})}\colon\A_{\infty}^{(0)}\to\A_{\infty}$ and $\zeta\in \rho(\delta^{(L_{1})})\cap\rho(\delta^{(L_{2})})$, with $L_{1},L_{2}\in\R_{0}^{+}\cup\{\infty\}$ and $L_{2}>L_{1}$,  we have
\begin{equation}\label{eq:resol}
R_{\zeta}\left(\delta^{(L_{2})}\right)-R_{\zeta}\left(\delta^{(L_{1})}\right)=\int_{0}^{\infty}\e^{-\zeta t}\left(\tau_{t}^{(L_{2})}(A)-\tau_{t}^{(L_{1})}(A)\right)\d t.
\end{equation}
By linearity, we can take $A\equiv\B(\varphi)$, with $\varphi\in\H_{\infty}$, use first part of the current Theorem and the \emph{Lebesgue's dominated convergence theorem} in order to conclude that $R_{\zeta}\left(\delta^{(L_{2})}\right)-R_{\zeta}\left(\delta^{(L_{1})}\right)$ is a Cauchy sequence. By above Expression it follows that
$$
\lim_{L\to\infty}\left\Vert\left(R_{\zeta}\left(\delta^{(\infty)}\right)-R_{\zeta}\left(\delta^{(L)}\right)\right)(A)\right\Vert_{\A_{\infty}}=0,
$$
as desired.
\item By taking into account proof of Theorem 4.8 in \cite{brupedraLR}, and the first and second Trotter--Kato approximation theorems \cite[Chap. III, Sect. 4.8 and 4.9]{EngelNagel} we claim that $\A_{\infty}^{(0)}\subset\mathrm{ran}\left(R_{\zeta}\left(\delta^{(L)}\right)\right)$ is dense in $\A_{\infty}$ and that $\delta$ is conservative.
\end{enumerate}
We finally remark, that the proof considered here for the operator $\delta$ on $\A$ with dense domain $\mathcal{D}(\delta)=\A_{\infty}^{(0)}$ also works if $\delta$ is \emph{unbounded}.
\end{proof}
In order to study quasi--free ground states at infinite volume we use:
\begin{proposition}\label{lemma:strongly_limits}
Let $\{H_{L}\}_{L\in\R_{0}^{+}}\in\BL(\H_{\infty})$ be a sequence of self--dual Hamiltonians on $(\H_{\infty},\a_{\infty})$ strongly convergent to $H_{\infty}\in\BL(\H_{\infty})$. For any $L\in\R_{0}^{+}\cup\{\infty\}$, $E_{+,\,L}$ will denote the spectral projection on $\R^{+}$ associated to the self--dual Hamiltonian $H_{L}$. If zero is not an eigenvalue of $H_{\infty}$, then $E_{+}$ will be
the strong limit of the sequence $\{E_{+,L}\}_{L\in\R_{0}^{+}}$, i.e., $\lim\limits_{L\to\infty}E_{+,L}=E_{+}$.
\end{proposition}
\begin{proof}
The proof is found in \cite[Lemma 3.3.]{araEva}
\end{proof}\par
For any $L\in\R_{0}^{+}$ let us define the set of \emph{local} quasi--free ground states by  $\q\states^{(L)}\subset\q\states^{(\infty)}$ on $\A_{L}\subset\A_{\infty}$. See Definition \ref{def:qfgs}. To be explicit, for any
self--dual Hamiltonian $H_{\infty}\in\BL(\H_{\infty})$ on $(\H_{\infty},\a_{\infty})$ and any orthogonal projection $P_{\H_{L}}\in\BL(\H_{\infty})$ on $\H_{L}$ the local Hamiltonian is given by \eqref{definition finite volume hamiltonian}, namely,
$$
H_{L}\doteq P_{\H_{L}}H_{\infty}P_{\H_{L}},
$$
which has an associated \emph{local} Gibbs state defined by
$$
\varrho_{\Lambda_{L}}\left(\B(\varphi_{1,\,\Lambda_{L}})\B(\varphi_{2,\,\Lambda_{L}})^{*}\right)\doteq\inner{\varphi_{1,\,\Lambda_{L}},E_{+,L}\varphi_{2,\,\Lambda_{L}}}_{\H_{L}},
$$
for $\varphi_{j,\Lambda_{L}}\in\H_{L}$, $j=\{1,2\}$, where $E_{+,L}$ denotes the sequence of spectral projections of Proposition \ref{lemma:strongly_limits}. Then, using Expressions \eqref{ass O0-00bis}--\eqref{Pfaffian} the local quasi--free ground state $\omega_{\Lambda_{L}}\in\q\states^{(L)}$ is found to be
\begin{multline}\label{eq:quasi_gibbs}
\omega_{\Lambda_{L}}\left(\B(\varphi_{1,\Lambda_{L}})\B(\varphi_{2,\Lambda_{L}})^{*}\B(\varphi_{3,\Lambda_{L}^{\text{c}}})\B(\varphi_{4,\Lambda_{L}^{\text{c}}})^{*}\right)=\\
\varrho_{\Lambda_{L}}\left(\B(\varphi_{1,\,\Lambda_{L}})\B(\varphi_{2,\,\Lambda_{L}})^{*}\right)\omega_{\Lambda_{L}}\left(\B(\varphi_{3,\Lambda_{L}^{\text{c}}})\B(\varphi_{4,\Lambda_{L}^{\text{c}}})^{*}\right)
\end{multline}
where $\varphi_{j,\Lambda_{L}^{\text{c}}}\in\H_{\Lambda_{L}^{\text{c}}}\equiv\H_{L}^{\text{c}}$, $j=\{3,4\}$, cf. \cite[Section 7.5]{Araki-Moriya}. By linearity, for any two \emph{even} elements $A\in\A_{L}^{+}$ and $B\in\A_{L}^{+,\,\text{c}}$, see \eqref{eq:even odd}, we get from \eqref{eq:quasi_gibbs}:
$$
\omega_{\Lambda_{L}}\left(AB\right)=\varrho_{\Lambda_{L}}\left(A\right)\omega_{\Lambda_{L}}\left(B\right),
$$
see again \cite[Section 7.5]{Araki-Moriya}. In particular, for $B=\mathfrak{1}\in\A_{\infty}$ we have
$$
\omega_{\Lambda_{L}}\left(A\right)=\varrho_{\Lambda_{L}}\left(A\right).
$$
We now state:
\begin{theorem}[Quasi--free ground states]\label{theo:quasi_free}
The local quasi--free ground state $\omega_{\Lambda_{L}}$ converges to
$$
\omega\left(\B(\varphi_{1})\B(\varphi_{2})^{*}\right)=\inner{\varphi_{1},E_{+}\varphi_{2}}_{\H_{\infty}},
$$
in the weak$^{*}$--topology, where $E_{+}\in\BL(\H_{\infty})$ is the spectral projection on $\R^{+}$ associated to the self--dual Hamiltonian $H_{\infty}\in\BL(\H_{\infty})$, and $\varphi_{1},\varphi_{2}\in\H_{\infty}$.
\end{theorem}
\begin{proof}
For the sake of clarity, for any $L\in\R_{0}^{+}$ denote $\H_{\Lambda_{L}}\equiv\H_{L}$, and $E_{\Lambda_{L}}\equiv E_{L}$. See Expressions \eqref{eq:boxesl}--\eqref{canonical onb1} and comments around it. Take $L_{1},L_{2}\in\R_{0}^{+}$, with $L_{2}\geq L_{1}$ such that $\Lambda_{L_{2}}\supsetneq\Lambda_{L_{1}}$. Thus, we analyze the following difference:
$$
D_{\omega_{L_{2}},\,\omega_{L_{1}}}\doteq\omega_{L_{2}}\left(\B(\varphi_{1,\,L_{2}})\B(\varphi_{2,\,L_{2}})^{*}\right)-\omega_{L_{1}}\left(\B(\varphi_{1,\,L_{1}})\B(\varphi_{2,\,L_{1}})^{*}\right).
$$
Here, in the way that the set of boxes $\Lambda_{L}$ was defined \eqref{eq:boxesl}, for $j=\{1,2\}$ we canonically identify $\varphi_{j,\Lambda_{L_{1}}}\in\H_{\Lambda_{L_{1}}}$ with the element $\varphi_{j,\Lambda_{L_{1}}}\oplus0_{\Lambda_{L_{2}}\setminus\Lambda_{L_{1}}}\in\H_{\Lambda_{L_{2}}}$. The spectral projections on $\R^{+}$ are related by
$E_{+,\Lambda_{L_{2}}}=E_{+,\Lambda_{L_{1}}}\oplus E_{+,\Lambda_{L_{2}}\setminus\Lambda_{L_{1}}}\in\BL(\H_{\Lambda_{L_{2}}})$.
 Straightforward calculations yield us to note that $\lim\limits_{L_{1}\to\infty}\lim\limits_{L_{2}\to\infty}D_{\omega_{L_{2}},\,\omega_{L_{1}}}$ equals zero.
\end{proof}\par
We are now in a position to prove the properties of the family of automorphisms $\kappa_{s}\colon\BL(\H)\to\BL(\H)$, for any $s\in\cal$, given by Assumption \ref{assump:hamil} and Lemma \ref{lemma:impor_res}, which are associated to a differentiable family of self--dual Hamiltonians $\mathbf{H}\in\BL(\H_{\infty})$, see Definition \ref{def: phase of the matter}. Observe that the existence of such $\kappa_{s}$ is closely related to the existence of a differentiable unitary operator $V_{s}$ satisfying the \emph{non}--autonomous differential equation, Expression \eqref{eq:spectral proof2}:
$$
\p_{s}V_{s}=-\ii \mathfrak{D}_{\g,s}V_{s},\quad\text{with}\quad V_{0}=\pm\1_{\H},
$$
where $\{\mathfrak{D}_{\g,s}\}_{s\in\cal}\in\BL(\H_{\infty})$ is a family of self--adjoint operators that is found to be
$$
\mathfrak{D}_{\g,s}\doteq \int_{\R}\e^{\ii tH_{s}}\left(\p_{s}H_{s}\right)\e^{-\ii tH_{s}}\mathfrak{W}_{\g}(t)\d t,
$$
with $\mathfrak{W}_{\g}\colon\R\to\R$ an integrable odd function the properties of which are summarized in \cite{bachmann2012automorphic,michalakis2013stability} and references therein. In the sequel, for any $s\in\cal$, $V_{s},H_{s}$ and $\p_{s}H_{s}$ have to be understood as the strong limit of the sequences $\{V_{s}^{(L)}\}_{L\in\R_{0}^{+}},\{H_{s,L}\}_{L\in\R_{0}^{+}}$ and $\{\p_{s}H_{s,L}\}_{L\in\R_{0}^{+}}$ respectively. We formulate:
\begin{lemma}\label{theor:unit_oper}
Take $\cal\equiv[0,1]$, fix $s\in\cal$ and consider the family of operators satisfying Assumption \ref{assump:hamil}. Then, the sequence of automorphisms $\{\kappa_{s}^{(L)}\}_{L\in\R_{0}^{+}}\colon\BL(\H_{L})\to\BL(\H_{L})$ of Lemma \ref{lemma:impor_res} on the local self--dual Hilbert space $(\H_{L},\a_{L})$ strongly converges on $\cal$ to $\kappa_{s}\colon\BL(\H_{\infty})\to\BL(\H_{\infty})$. More precisely, for any $\Lambda\in\Pfz$, $B\in\BL(\H_{\Lambda})$ and $L\in\R_{0}^{+}$ such that $\Lambda\subset\Lambda_{L}$ we have
\begin{align*}
\lim_{L\to\infty}\left\Vert\kappa_{s}(B)-\kappa_{s}^{(L)}(B)\right\Vert_{\BL(\H_{\infty})}=0, \qquad\text{for any}\qquad s\in\cal.
\end{align*}
\end{lemma}
\begin{proof}
Fix $\Lambda\in\Pfz$ and take $L_{1},L_{2}\in\R_{0}^{+}$, with $L_{2}\geq L_{1}$ such that $\Lambda_{L_{2}}\supsetneq\Lambda_{L_{1}}\supset\Lambda$. We proceed in a similar way as in \cite[Lemma 4.4]{brupedraLR}. Note that with a few modifications of the proof we can arrive at a result that works even in the \emph{interparticle} case \cite{AMR2}.\par
For any $L\in\R_{0}^{+}$, let $V_{s}^{(L)}$ be the unitary operator satisfying the differential equation \eqref{eq:spectral proof2}, with $V_{0}^{(L)}=\pm\1_{\H}$. For $s,r\in\cal$, one defines the unitary element
\begin{equation}\label{eq:unit_ele_sr}
U_{L}(s,r)\doteq V_{s}^{(L)}\left(V_{r}^{(L)}\right)^{*},
\end{equation}
which satisfies $U_{L}(s,s)=\1_{\H}$ for all $s\in\cal$ while
\begin{equation}\label{eq:unit_ele_sr2}
\p_{s}U_{L}(s,r)=-\ii \mathfrak{D}_{\g,s}^{(L)}U_{L}(s,r)\qquad\text{and}\qquad
\p_{r}U_{L}(s,r)=\ii U_{L}(s,r)\mathfrak{D}_{\g,r}^{(L)}.
\end{equation}
Note that for $B\in\BL(\H_{\Lambda})$ one can write
$$
\kappa_{s}^{(L_{2})}(B)-\kappa_{s}^{(L_{1})}(B)=\int_{0}^{s}\p_{r}(U_{L_{2}}(0,r)U_{L_{1}}(r,s)BU_{L_{1}}(s,r)U_{L_{2}}(r,0))\d r.
$$
Straightforward calculations show us that the derivative inside the integral is
\begin{equation}\label{eq:diff_Ds}
 \ii U_{L_{2}}(0,r)\left[\left(\mathfrak{D}_{\g,r}^{(L_{2})}-\mathfrak{D}_{\g,r}^{(L_{1})}\right),U_{L_{1}}(r,s)BU_{L_{1}}(s,r)\right]U_{L_{2}}(r,0)
\end{equation}
with $s,r\in\cal$, and for $L\in\R_{0}^{+}$, and $\Lambda\in\Pfz$, $\Lambda\subset\Lambda_{L}$.\\
On the other hand, for any $s\in\cal$, $t\in\R^{+}$, $\Lambda\in\Pfz$ and $L\in\R_{0}^{+}$ such that $\Lambda_{L}\supset\Lambda$, define the $s$--automorphism $\widetilde{\tau}_{s,\,t}^{(L)}\colon\BL(\H_{L})\to\BL(\H_{L})$ by
$$
\widetilde{\tau}_{s,\,t}^{(L)}(B)\doteq\e^{\ii tH_{s,\,L}}B\e^{-\ii tH_{s,\,L}},
$$
with $H_{s,\,L}$ a self--dual Hamiltonian on $(\H_{L},\a_{L})$. Then, for $L_{1},L_{2}\in\R_{0}^{+}$ one can write the following
\begin{eqnarray*}
\widetilde{\tau}_{s,\,t}^{(L_{2})}(B)-\widetilde{\tau}_{s,\,t}^{(L_{1})}(B)&=&\int_{0}^{t}\p_{u}\left(\widetilde{\tau}_{s,\,u}^{(L_{2})}\circ\widetilde{\tau}_{s,\,t-u}^{(L_{1})}(B)\right)\d u\\
&=&\ii\int_{0}^{t}\widetilde{\tau}_{u}^{(L_{2})}\left(\left[H_{s,\,L_{2}}-H_{s,\,L_{1}},\widetilde{\tau}_{s,\,t-u}^{(L_{1})}(B)\right]\right)\d u,
\end{eqnarray*}
where, for a fix $s\in\cal$, the difference $H_{s,\,L_{2}}-H_{s,\,L_{1}}\in\BL(\H_{\infty})$ is given by
$$
H_{s,\,L_{2}}-H_{s,\,L_{1}}=P_{\H_{L_{1}}}H_{s,\,\infty}P_{\H_{L_{2}}\setminus\H_{L_{1}}}+P_{\H_{L_{2}}\setminus\H_{L_{1}}}H_{s,\,\infty}P_{\H_{L_{1}}}+P_{\H_{L_{2}}\setminus\H_{L_{1}}}H_{s,\,\infty}P_{\H_{L_{2}}\setminus\H_{L_{1}}},
$$
where $P_{\H_{L_{2}}\setminus\H_{L_{1}}}\equiv P_{\H_{L_{2}}}-P_{\H_{L_{1}}}\in\BL(\H_{\infty})$ is the orthogonal projector on $\H_{L_{2}}\setminus\H_{L_{1}}$. Here, $H_{s,\,\infty}$ is the self--dual Hamiltonian on $(\H_{\infty},\a_{\infty})$ at infinite volume. It follows that,
\begin{eqnarray}\label{eq:diff_auto}
\left\Vert
\widetilde{\tau}_{s,\,t}^{(L_{2})}(B)-\widetilde{\tau}_{s,\,t}^{(L_{1})}(B)\right\Vert_{\BL(\H_{\infty})}&\leq&\int_{0}^{t}\left\Vert\left[H_{s,\,L_{2}}-H_{s,\,L_{1}},\widetilde{\tau}_{s,\,t-u}^{(L_{1})}(B)\right]\right\Vert\d u \nonumber\\
&\leq&2|t|\left\Vert B\right\Vert_{\BL(\H_{\infty})}\left\Vert H_{s,\,L_{2}}-H_{s,\,L_{1}}\right\Vert_{\BL(\H_{\infty})}.
\end{eqnarray}
By Assumption \ref{assump:hamil}, $\{H_{s,\,L}\}_{L\in\R_{0}^{+}}$ is a sequence of operators which converges in norm to $H_{s,\,\infty}$ as $L\to\infty$, then last expression is a Cauchy sequence of self--adjoint operators. Hence, for all $t\in\R$, $\widetilde{\tau}_{s,\,t}^{(L)}$ converges strongly on $\BL(\H_{L})$ to $\widetilde{\tau}_{s,\,t}$, as $L\to\infty$.\\
What is important to stress is that the difference $\mathfrak{D}_{\g,r}^{(L_{2})}-\mathfrak{D}_{\g,r}^{(L_{1})}$ in Expression \eqref{eq:diff_Ds} can be written as follows
\begin{align*}
\mathfrak{D}_{\g,r}^{(L_{2})}-\mathfrak{D}_{\g,r}^{(L_{1})}
&=\int_{\R}\left(\widetilde{\tau}_{r,\,t}^{(L_{2})}(\p_{r}\left\{H_{r,L_{2}}\right\})-\widetilde{\tau}_{r,\,t}^{(L_{2})}(\p_{r}\left\{H_{r,L_{1}}\right\})\right)\mathfrak{W}_{\g}(t)\d t+\\
&\int_{\R}\left(\widetilde{\tau}_{r,\,t}^{(L_{2})}(\p_{r}\left\{H_{r,L_{1}}\right\})-\widetilde{\tau}_{r,\,t}^{(L_{1})}(\p_{r}\left\{H_{r,L_{1}}\right\})\right)\mathfrak{W}_{\g}(t)\d t.
\end{align*}
From which one has
\begin{multline}\label{eq:kappas}
\left\Vert\kappa_{s}^{(L_{2})}(B)-\kappa_{s}^{(L_{1})}(B)\right\Vert_{\BL(\H_{\infty})}\leq 2\left\Vert B\right\Vert_{\BL(\H_{\infty})}|s|\\
\sup\limits_{r\in\cal}\left(\int_{\R}\left\Vert\p_{r}\left\{H_{r,L_{2}}\right\}-\p_{r}\left\{H_{r,L_{1}}\right\}\right\Vert_{\BL(\H_{\infty})}\left|\mathfrak{W}_{\g}(t)\right|\d t\right.\\
\left.+\int_{\R}\left\Vert\left(\widetilde{\tau}_{r,\,t}^{(L_{2})}-\widetilde{\tau}_{r,\,t}^{(L_{1})}\right)\p_{r}\left\{H_{r,L_{1}}\right\}\right\Vert_{\BL(\H_{\infty})}\left|\mathfrak{W}_{\g}(t)\right|\d t\right).
\end{multline}
Hence, for a fixed $s\in\cal$, by Assumption \ref{assump:hamil} and Inequality \eqref{eq:diff_auto} one notes that the right hand side of the last inequality vanishes as $L_{2}\to\infty$ and $L_{1}\to\infty$. Thus, $\kappa_{s}^{(L)}$ is a pointwise Cauchy sequence as $L\to\infty$ and hence the family of automorphism $\kappa_{s}^{(L)}$ converges strongly on $\BL(\H_{L})$ to $\kappa_{s}$ as $L\to\infty$.
\end{proof}
As a consequence we have:
\begin{corollary}\label{corol:V_cauchy}
Make the same assumptions as in Lemma \ref{theor:unit_oper}. Then for any $s\in\cal$, the sequence of unitary operators $V_{s}^{(L)}$ converges in norm, as $L\to\infty$, to some $V_{s}$.
\end{corollary}
\begin{proof} As is usual, it is enough to show that the sequence $V_{s}^{(L)}$ is a Cauchy sequence. Note that for any $s\in\cal$ and $L_{1},L_{2}\in\R_{0}^{+}$ with $L_{2}\geq L_{1}$, we can write:
$$
\left(V_{s}^{(L_{2})}\right)^{*}-\left(V_{s}^{(L_{1})}\right)^{*}=\int_{0}^{s}\p_{r}\left(U_{L_{2}}(0,r)U_{L_{1}}(r,s)\right)\d r,
$$
where for any $s,r\in\cal$, $U_{L}(s,r)$ is the unitary element defined by  \eqref{eq:unit_ele_sr}--\eqref{eq:unit_ele_sr2}. Straightforward calculations yield to
$$
\left(V_{s}^{(L_{2})}\right)^{*}-\left(V_{s}^{(L_{1})}\right)^{*}=\ii\int_{0}^{s}U_{L_{2}}(0,r)\left(\mathfrak{D}_{\g,r}^{(L_{2})}-\mathfrak{D}_{\g,r}^{(L_{1})}\right)U_{L_{1}}(r,s)\d r.
$$
Proceeding as in \eqref{eq:kappas} we arrive at the desired result. We omit the details.
\end{proof}
\begin{lemma}[Uniformity of the determinant]\label{lemma:VI_cauchy}
Make the same assumptions as in Lemma \ref{theor:unit_oper} and suppose that for any $s\in\cal$ and $L\in\R_{0}^{+}\cup\{\infty\}$, $\1_{\H_{L}}-V_{s}^{(L)}$ and $\mathfrak{D}_{\g,s}^{(L)}$ ($\mathfrak{D}_{\g,s}^{(L)}$ given by \eqref{eq:spectral proof2}) are trace class on $\H_{L}$. Then, for any $s\in\cal$, the family of determinants $\sigma^{(L)}(s)\doteq\det\left(V_{s}^{(L)}\right)\in\{-1,1\}$ is uniform for $L\in\R_{0}^{+}\cup\{\infty\}$. Moreover, the sequence $\sigma^{(L)}\in\left\{-1,1\right\}$ converges uniformly on $\mathcal{C}$.
\end{lemma}
\begin{proof}
For $L\in\R_{0}^{+}\cup\{\infty\}$ and any $s,r\in\cal$ consider $U_{L}(s,r)$, the unitary element defined by \eqref{eq:unit_ele_sr}--\eqref{eq:unit_ele_sr2}. By the Jacobi's formula of determinants for $U_{L}(s,r)$ we have for $L_{1},L_{2}\in\R_{0}^{+}$ with $L_{2}\geq L_{1}$ that
\begin{multline*}
\left|\det\left(V_{s}^{(L_{2})}\right)-\det\left(V_{s}^{(L_{1})}\right)\right|=\left|\int_{0}^{s}\p_{r}\left(\det\left(U_{L_{2}}(r,0)\right)\det\left(U_{L_{1}}(s,r)\right)\right)\d r\right|\\
=\left|\int_{0}^{s}\det\left(U_{L_{2}}(0,r)\right)\det\left(U_{L_{1}}(s,r)\right)
\left(\tr_{\H_{L_{2}}}\left(\mathfrak{D}_{\g,r}^{(L_{2})}\right)-\tr_{\H_{L_{1}}}\left(\mathfrak{D}_{\g,r}^{(L_{1})}\right)\right)\d r\right|
\end{multline*}
Now, similar to the proof of Corollary \ref{corol:sign}, since $H_{s,\,L}$ is self--adjoint for $s\in\cal$ and $L\in\R_{0}^{+}\cup\{\infty\}$, $\p_{s}H_{s,\,L}$ also is. Thus, by Expression \eqref{eq: operator D} and the cyclic property of the trace, it follows that $\tr_{\H_{L}}\left(\mathfrak{D}_{\g,s}^{(L)}\right)=0$, for $L\in\R_{0}^{+}$. Now, for any $s\in\cal$, note that the sequence of functions $\left\{\det\left(V_{s}^{(L)}\right)\right\}_{L\in\N}$ is equicontinuous and pointwise bounded. By the \emph{Ascoli--Arzel\`{a}} Theorem there exists a uniform convergent subsequence $\left\{\det\left(V_{s}^{(L^{(n)})}\right)\right\}_{n\in\N}$ such that the map $s\mapsto\det\left(V_{s}^{(L^{(n)})}\right)$ converges uniformly for $s\in\cal$. By Corollary \ref{corol:sign}, for any $s\in\cal$ and $n\in\N$, $\sigma_{s}^{(L^{(n)})}\doteq\det\left(V_{s}^{(L^{(n)})}\right)=\det\left(V_{0}^{(L^{(n)})}\right)=\pm1$.
\end{proof}
\subsection{Decay estimates of correlations and gapped quasi--free ground states}
Fix $\epsilon\in(0,1]$ and let $(\H_{\infty},\a_{\infty})$ be the self--dual Hilbert space as defined in subsection \ref{subsec: lattice fermion}. Moreover, consider the family of self--adjoint operators $\{A_{s}\}_{s\in\cal}\in\BL(\H_{\infty})$. Thus, for any $s\in\cal$ we define the constants
\begin{equation*}
\mathbf{S}(A_{s},\mu)\doteq\sup_{\x_{1}\in\mathbb{X}_{\infty}}\sum_{\x_{2}\in\mathbb{X}_{\infty}}\left(\e^{\mu|x_{1}-x_{2}|^{\epsilon}}-1\right)\left\vert\inner{\E_{\x_{1}},A_{s}\E_{\x_{2}}}_{\H_{\infty}}\right\vert\in\R_{0}^{+}\cup \left\{ \infty \right\},
\end{equation*}
for $\mu\in \R_{0}^{+}$ and
\begin{equation*}
\Delta (A_{s},z)\doteq \inf\left\{ \left\vert z-\lambda \right\vert\colon\lambda
\in\spec(A_{s})\right\},\qquad z\in \mathbb{C},
\end{equation*}
as the distance from the point $z$ to the spectrum of $A_{s}$. $\mathbb{X}_{\infty}$ is defined by \eqref{eq: X_L}. Here, $\mu$ is not necessarily the same for two different operators $A_{s_{1}},A_{s_{2}}\in\{A_{s}\}_{s\in\cal}$, but in the sequel w.l.o.g. we will assume this. Since the function $x\mapsto(\e^{xr}-1)/x$ is increasing on $\R^{+}$ for any fixed $r\geq0$, it follows that
\begin{equation}
\mathbf{S}(A_{s},\mu_{1})\leq\frac{\mu _{1}}{\mu _{2}}\mathbf{S}(A_{s},\mu _{2}),\qquad \mu_{2}\geq \mu _{1}\geq 0.\label{inequality combes easy}
\end{equation}
We have the following Combes--Thomas estimates:
\begin{proposition}[Combes--Thomas]\label{Combes-Thomas}
Let $\epsilon \in (0,1],\cal\doteq[0,1]$, the family of self--adjoint operators $\{A_{s}\}_{s\in\cal}\in\mathcal{B}(\H_{\infty})$, $\mu\in \R_{0}^{+}$ and $z\in \mathbb{C}$. If $\Delta (A_{s},z)>\mathbf{S}(A_{s},\mu )$ then for any $s\in\cal$ and $\x=(x,\mathfrak{s},v),\,\y=(y,\mathfrak{t},w)\in\mathbb{X}_{\infty}$
\begin{align*}
\left\vert\inner{\E_{\x},(z-A_{s})^{-1}\E_{\y}}_{\H_{\infty}}\right\vert\leq\sup_{s\in\cal}\left\{\frac{\e^{-\mu|x-y|^{\epsilon }}}{\Delta (A_{s},z)-\mathbf{S}(A_{s},\mu)}\right\}.
\end{align*}
\end{proposition}
For a proof see \cite[Theorem 10.5]{AW16}. Some immediate consequences are summarized as follows:
\begin{corollary}\label{Lemma AG98}
Let $\epsilon \in(0,1],\cal\doteq[0,1]$, the family of self--adjoint operators $\{A_{s}\}_{s\in\cal}\in\mathcal{B}(\H_{\infty})$, $\mu\in \R_{0}^{+}$ and all $\x=(x,\mathfrak{s},v),\y=(y,\mathfrak{t},w)\in \mathbb{X}_{\infty }$.
Then,
\begin{enumerate}
\item[(a)] Let $\eta\in\R^{+}$ such that $\sup\limits_{s\in\cal}\left\{\mathbf{S}(A_{s},\mu )\right\}\leq \eta/2, u\in\R$ and $s\in\cal$,
\begin{align*}
&\left\vert\left\langle\E_{\x},((A_{s}-u)^{2}+\eta^{2})^{-1}\E_{\y}\right\rangle_{\H_{\infty}} \right\vert \\
&\leq D_{\ref{Lemma AG98},(a)}\e^{-\mu |x-y|^{\epsilon }}\sup\limits_{s\in\cal}\left\{\left\langle \E_{\x},((A_{s}-u)^{2}+\eta^{2})^{-1}\E_{\x}\right\rangle_{\H_{\infty}}^{1/2}\left\langle \E_{\y},((A_{s}-u)^{2}+\eta ^{2})^{-1}\E_{\y}\right\rangle_{\H_{\infty}}^{1/2}\right\}.
\end{align*}
Moreover, for any function $G(z)\colon\CP\to\CP$ analytic on $|\Imm(z)|\leq\eta$ and uniformly bounded by $\Vert G\Vert_{\infty}$ we have
$$
\inner{\E_{\x},G(A_{s})\E_{\y}}_{\H_{\infty}}\leq D_{\ref{Lemma AG98},(b)}\Vert G\Vert_{\infty}\e^{-\mu\min\left\{1,\inf\limits_{s\in\cal}\left\{\frac{\eta}{4S(A_{s},\mu)}\right\}\right\}|x-y|^{\epsilon}}.
$$
\item[(b)] (Gapped Case) For $z\in \mathbb{C}$ such that $\inf\limits_{s\in\cal}\Delta (A_{s},z)\geq \g/2>0$, with $\g$ as in Definition \ref{def: phase of the matter}:
\begin{equation}\label{combes 0}
\left\vert \left\langle \E_{\x},(z-A_{s})^{-1}\E_{\y}\right\rangle_{\H_{\infty}}\right\vert \leq 4\g^{-1}\exp\left(-\mu \min \left\{ 1,\inf\limits_{s\in\cal}\left\{\frac{\g}{4S(A_{s},\mu)}\right\}\right\}
|x-y|^{\epsilon }\right).
\end{equation}
Moreover, for $\eta\in\left(0,\g/2\right]$, and any function $G(z)\colon\CP\to\CP$ analytic on $z\in \R_{0}^{+}+\eta +\ii\eta \left[-1,1\right]$ and uniformly bounded by $\Vert G\Vert_{\infty}$ we have
$$
\inner{\E_{\x},E_{+}G(A_{s})E_{+}\E_{\y}}_{\H_{\infty}}\leq D_{\ref{Lemma AG98},(c)}\Vert G\Vert_{\infty}\e^{-\mu\min\left\{1,\inf\limits_{s\in\cal}\left\{\frac{\g}{4S(A_{s},\mu)}\right\}\right\}|x-y|^{\epsilon}}.
$$
\end{enumerate}
In all inequalities, the numbers $D_{\ref{Lemma AG98},(a)},D_{\ref{Lemma AG98},(b)},D_{\ref{Lemma AG98},(c)}\in\R^{+}$ are suitable constants.
\end{corollary}
\begin{proof}
(a) is proven as in \cite[Theorem 3 and Lemma 3]{AG98}. (b) The first part is a consequence of Proposition \ref{Combes-Thomas} together with Inequality \eqref{inequality combes easy}. On the other hand, we use Cauchy's integral formula to write, for all real $E\in \R\backslash \{\eta \}$,
\begin{equation*}
\chi_{(\eta,\infty)} G\left( E\right) =\frac{1}{2\pi\ii}\int_{\eta
}^{\infty}\left(\frac{G\left(u-\ii\eta\right) }{u-E-\ii\eta }-\frac{G\left(
u+\ii\eta\right)}{u-E+\ii\eta}\right) \d u-\frac{1}{2\pi }\int_{-\eta
}^{\eta}\frac{G\left(\eta+\ii u\right) }{\eta -E+\ii u}\d u\ ,
\end{equation*}
which yields
\begin{eqnarray*}
\chi_{(\eta,\infty)}G\left(E\right)&=&\frac{\eta}{\pi}\int_{\eta }^{\infty}\frac{G\left(u-\ii\eta\right)+G\left(u+\ii\eta \right)}{\left( u-E\right) ^{2}+\eta^{2}}\d u-\frac{2\eta }{\pi}\int_{\eta}^{\infty }\frac{G\left( u\right) }{\left(u-E\right) ^{2}+4\eta ^{2}}\d u\\
&&+\frac{1}{2\pi}\int_{0}^{\eta }\frac{G\left( \eta -\ii u\right) }{\eta
-\ii u-E+2\ii\eta}\d u+\frac{1}{2\pi}\int_{0}^{\eta}\frac{G\left( \eta
+\ii u\right)}{\eta+\ii u-E-2\ii\eta}\d u\\
&&-\frac{1}{2\pi }\int_{-\eta}^{\eta }\frac{G\left(\eta+\ii u\right)}{\eta
-E+\ii u}\d u.
\end{eqnarray*}
By spectral calculus, together the last equality, part (a) of this Corollary, Inequality \eqref{combes 0} and the Cauchy--Schwarz inequality, the result follows. For further details see \cite[Lemma 5.12]{LD1}.
\end{proof}
At this point it is useful to introduce the normalized \emph{trace per unit volume} as
$$
\Tr(\cdot)\doteq\lim\limits_{L\to\infty}\frac{1}{\dim(\H_{L})}\tr_{\H_{L}}(\cdot).
$$
We are able to state the following:
\begin{lemma}\label{lemma:unit_oper2}
Take $\cal\equiv[0,1]$ and consider the family of operators satisfying assumptions of Corollary \ref{Lemma AG98} for $\{\p_{s}H_{s}^{(L)}\}_{s\in\cal}\in\BL(\H_{L})$, $L\in\R_{0}^{+}\cup\{\infty\}$. Consider the pointwise sequence $V_{s}^{(L)}\colon\cal\to\BL(\H)$, $L\in\R_{0}^{+}\cup\{\infty\}$, of unitary operators satisfying \eqref{eq:spectral proof2}. Then, the sequence $\left\{\1_{\H_{L}}-V_{s}^{(L)}\right\}_{L\in\R_{0}^{+}\cup\{\infty\}}$ is trace--class per unit volume. Thus, for $L\in\R_{0}^{+}\cup\{\infty\}$, the family of one--parameter (Bogoliubov) group $\left\{\Upsilon_{s}^{(L)}\right\}_{s\in\cal\in\R}$ of $^{*}$--automorphisms on $\A_{L}$ (see \eqref{eq: local}), given by \eqref{eq:autom}, is inner\footnote{For $U\in\BL(\H)$, a Bogoliubov transformation, the Bogoliubov $^{*}$--automorphism $\chi_{U}$ on $\sCAR$ is inner if and only if $\1_{\H_{\infty}}-U$ is trace class and $\det(U)=\pm1$, see \cite[Theorem 4.1]{A87}.}.
\end{lemma}
\begin{proof}
For $s\in\cal$ and $L\in\R_{0}^{+}$, let $W_{s}^{(L)}\in\BL(\H_{L})$ be the partial isometry arising from the polar decomposition of $\1_{\H_{L}}-V_{s}^{(L)}$
$$
\1_{\H_{L}}-V_{s}^{(L)}=W_{s}^{(L)}\left|\1_{\H_{L}}-V_{s}^{(L)}\right|.
$$
From this one can calculate the trace of $\left|\1_{\H_{L}}-V_{s}^{(L)}\right|$ as follows
$$
\tr_{\H_{L}}\left|\1_{\H_{L}}-V_{s}^{(L)}\right|=\sum_{\x\in\mathbb{X}_{L}}\inner{\E_{\x},\left(W_{s}^{(L)}\right)^{*}\left(\1_{\H_{L}}-V_{s}^{(L)}\right)\E_{\x}}_{\H_{L}}.
$$
Note that for the unitary bounded operator $V_{s}^{(L)}$ on $\H_{L}$ we can write $\1_{\H_{L}}-V_{s}^{(L)}=-\ii\int_{0}^{s}\mathfrak{D}_{\g,r}^{(L)}V_{r}^{(L)}\d r$. Then, by combining the explicit form of $\mathfrak{D}_{\g,r}^{(L)}$ given by \eqref{eq: operator D}, Cauchy--Schwarz inequality, Corollary \ref{Lemma AG98}, and other simple arguments we arrive at
\begin{multline*}
\frac{1}{\left|\Lambda_{L}\right|}\left|\tr_{\H_{L}}\left|\1_{\H_{L}}-V_{s}^{(L)}\right|\right|\leq D_{\text{Lem. } \ref{lemma:unit_oper2}}|s|\left|\mathfrak{S}\right|
\int_{\R}\left|\mathfrak{W}_{\g}(t)\right|\d t
\\
\sum_{x\in\Z^{d}}\e^{-\mu\min\left\{1,\inf\limits_{r\in\cal}\left\{\frac{\eta}{4S(\p_{r}H_{r},\mu)}\right\},\inf\limits_{r\in\cal}\left\{\frac{\g}{4S(H_{r},\mu)}\right\}\right\}|x|^{\epsilon}}.
\end{multline*}
Then, $\1_{\H_{L}}-V_{s}^{(L)}$ is trace class per unit volume, and $V_{s}^{(L)}\in\BL(\H_{L})$ is a Bogoliubov transformation such that $\det\left(V_{s}^{(L)}\right)=\pm1$, for $L\in\R_{0}^{+}\cup\{\infty\}$. See also Lemma \ref{lemma:VI_cauchy}. It follows from \cite[Theorem 4.1]{A87} that the $^{*}$--automorphism $\Upsilon_{s}^{(L)}$ on $\A_{L}$ is inner.
\end{proof}
A combination of Corollary \ref{corol:V_cauchy} and Lemma \ref{lemma:unit_oper2} yields to:
\begin{corollary}\label{corol: automorphismY}
Take same assumptions of Lemma \ref{theor:unit_oper}. Then, the one--parameter (Bogoliubov) group $\Upsilon_{s}^{(L)}$ on $\A_{\infty}^{(0)}$ converges uniformly for $s\in\cal$ as $L\to\infty$ to the one--parameter (Bogoliubov) group $\Upsilon_{s}$ on $\A_{\infty}$, thus defining a strongly continuous group on $\A_{\infty}$\footnote{Recall that $\A_{\infty}$ is the completeness of the normed $^{*}$--algebra $\mathfrak{A}_{\infty}^{(0)}$ given by \eqref{local elements}.}. Moreover, $\left(\Upsilon_{s}^{(L)}\right)^{-1}$ exists and strongly converges to $\Upsilon_{s}^{-1}$.
\end{corollary}
\begin{proof}
Note that the sequence of one--parameter (Bogoliubov) group $\Upsilon_{s}^{(L)}$ on $\A_{\infty}^{(0)}$ is Cauchy for any $B\in\A_{\infty}^{(0)}$. We omit the details. Existence of $\left(\Upsilon_{s}^{(L)}\right)^{-1}$ is a straight conclusion from Corollary \ref{theor:unit_oper}, its convergence is immediate. We also omit the details.
\end{proof}
In regard to the unitary operator $U_{L}(s,r)$ defined by \eqref{eq:unit_ele_sr}--\eqref{eq:unit_ele_sr2} for any $r,s\in\cal$ and $L\in\R_{0}^{+}\cup\{\infty\}$ we have the following:
\begin{lemma}\label{lemma:unit_oper3}
Take same assumptions of Lemma \ref{theor:unit_oper} and consider the unitary operator $U_{L}(s,r)$ defined by \eqref{eq:unit_ele_sr}--\eqref{eq:unit_ele_sr2}. For fixed $r,s\in\cal$ we have: (a) The sequence $\left\{\1_{\H_{L}}-U_{L}(s,r))\right\}_{L\in\R_{0}^{+}\cup\{\infty\}}$ is trace--class per unit volume. (b) $U_{L}(s,r)$ commutes with the involution $\a_{L}$ for any $L\in\R_{0}^{+}\cup\{\infty\}$. (c) $\det\left(U_{L}(s,r)\right)=1$.
\end{lemma}
\begin{proof}
(a) Similar to proof of Lemma \ref{lemma:unit_oper2} for $r,s\in\cal$ and $L\in\R_{0}^{+}$, let $W_{L}(r,s)\in\BL(\H_{L})$ be the partial isometry arising from the polar decomposition of $\1_{\H_{L}}-U_{L}(r,s)$, in such a way that we write
\begin{eqnarray*}
\tr_{\H_{L}}\left|\1_{\H_{L}}-U_{L}(r,s)\right|&=&\sum_{\x\in\mathbb{X}_{L}}\inner{\E_{\x},\left(W_{L}(r,s)\right)^{*}\left(\1_{\H_{L}}-U_{L}(r,s)\right)\E_{\x}}_{\H_{L}}\\
&=&\sum_{\x\in\mathbb{X}_{L}}\inner{\E_{\x},\left(W_{L}(r,s)\right)^{*}\left(V_{r}^{(L)}-V_{s}^{(L)}\right)\left(V_{r}^{(L)}\right)^{*}\E_{\x}}_{\H_{L}},
\end{eqnarray*}
where we have used \eqref{eq:unit_ele_sr}. Note that we can write $V_{r}^{(L)}-V_{s}^{(L)}=-\ii\int_{s}^{r}\mathfrak{D}_{\g,q}^{(L)}V_{q}^{(L)}\d q$. Then, by combining the explicit form of $\mathfrak{D}_{\g,r}^{(L)}$ given by \eqref{eq: operator D}, Cauchy--Schwarz inequality, Corollary \ref{Lemma AG98}, and other simple arguments we arrive at
\begin{multline*}
\frac{1}{\left|\Lambda_{L}\right|}\left|\tr_{\H_{L}}\left|\1_{\H_{L}}-U_{L}(s,r)\right|\right|\leq D_{\text{Lem. } \ref{lemma:unit_oper3}}|r-s|\left|\mathfrak{S}\right|
\int_{\R}\left|\mathfrak{W}_{\g}(t)\right|\d t
\\
\sum_{x\in\Z^{d}}\e^{-\mu\min\left\{1,\inf\limits_{r\in\cal}\left\{\frac{\eta}{4S(\p_{r}H_{r},\mu)}\right\},\inf\limits_{r\in\cal}\left\{\frac{\g}{4S(H_{r},\mu)}\right\}\right\}|x|^{\epsilon}}.
\end{multline*}
Then, $\1_{\H_{L}}-U_{L}(s,r)$ is trace--class per unit volume. Part (b) is straightforward from Corollary \ref{corol:Vs} applied for $r,s\in\cal$ and $L\in\R_{0}^{+}\cup\{\infty\}$. Part (c) follows from parts (a) and (b) and taking into account Corollary \ref{corol:sign} and the uniformity of the determinants of Lemma \ref{lemma:VI_cauchy}: $\det\left(V_{s}^{(L)}\right)=\det\left(V_{r}^{(L)}\right)=\pm1$.
\end{proof}
For any $L\in\R_{0}^{+}$, $\q\states^{(L,\infty)}\subset\q\states^{(\infty)}$ denotes the \emph{local} quasi--free ground states on $\A_{L}\subset\A_{\infty}$. We postulate:
\begin{theorem}[Gapped quasi--free ground states]\label{lemma: supporting states2}
Take $\cal\equiv[0,1]$ and consider the family of self--dual Hamiltonians satisfying Assumption \ref{assump:hamil} (b). Fix $L\in\R_{0}^{+}$, and let $\left\{\omega_{s}^{(L)}\right\}_{s\in\cal}\subset\q\states^{(L,\infty)}$ be the family of gapped quasi--free ground states associated to the family of Hamiltonians $\left\{H_{s}^{(L)}\right\}_{s\in\cal}\in\BL(\H_{L})$. Then,
\begin{enumerate}
 \item[(1)] $\omega_{s}^{(L)}=\omega_{0}^{(L)}\circ\Upsilon_{s}^{(L)}$, for all $s\in\cal,$ where $\Upsilon_{s}^{(L)}$ is the finite--volume Bogoliubov $^{*}$--automorphism on  $\A_{L}$ of Corollary \ref{corol: automorphismY}.
\item[(2)] Let $\omega_{s}\in\q\states^{(\infty)}$ be the weak$^{*}$--limit of $\omega_{s}^{(L)}\in\q\states^{(L,\infty)}$ and $\Upsilon_{s}$ the infinite--volume Bogoliubov $^{*}$--automorphism on $\A_{\infty}$ associated to the sequence $\Upsilon_{s}^{(L)}$ of Corollary \ref{corol: automorphismY}. With respect to the weak$^{*}$--topology, the following three statements are equivalent: (a) $\lim\limits_{L\to\infty}\omega_{s}^{(L)}=\omega_{s}$. (b) $\lim\limits_{L\to\infty}\omega_{s}^{(L)}\circ\Upsilon_{s}=\omega_{s}\circ\Upsilon_{s}$. (c) $\lim\limits_{L\to\infty}\omega_{s}^{(L)}\circ\Upsilon_{s}^{(L)}=\omega_{s}\circ\Upsilon_{s}$.
\end{enumerate}
\end{theorem}
\begin{proof}
(1) follows from Corollary \ref{theorem:states_algebras} and Lemma \ref{lemma:unit_oper2}. (2) Fix $s\in\cal$. Note that the existence of the weak$^{*}$--limit $\omega_{s}$ is consequence of Theorem \ref{theo:quasi_free} while the existence of the Bogoliubov $^{*}$--automorphism $\Upsilon_{s}$ is a consequence of Corollary \ref{corol: automorphismY}. Now, take any $A\in\A_{\infty}$ and note that (a) $\Rightarrow$ (b) because
$$
\left|\omega_{s}^{(L)}\circ\Upsilon_{s}(A)-\omega_{s}\circ\Upsilon_{s}(A)\right|\leq\left|\omega_{s}^{(L)}-\omega_{s}\right|\Vert A\Vert_{\A_{\infty}}.
$$
(b) $\Rightarrow$ (c) follows by recognizing $\omega_{s}^{(L)}$ and $\omega_{s}$ as states and writing
$$
\left|\omega_{s}^{(L)}\circ\Upsilon_{s}^{(L)}(A)-\omega_{s}\circ\Upsilon_{s}(A)\right|\leq\left|\omega_{s}^{(L)}-\omega_{s}\right|\Vert A\Vert_{\A_{\infty}}+\left|\omega_{s}^{(L)}\right|\left\Vert\Upsilon_{s}^{(L)}(A)-\Upsilon_{s}(A)\right\Vert_{\A_{\infty}},
$$
and we have that the left hand side of last inequality is zero. Finally, we note that
$$
\left|\omega_{s}^{(L)}(A)-\omega_{s}(A)\right|\leq
\left|\omega_{s}^{(L)}\circ\Upsilon_{s}^{(L)}-\omega_{s}\circ\Upsilon_{s}\right|\Vert A\Vert+\left|\omega_{s}\circ\Upsilon_{s}\right|\left\Vert\left(\Upsilon_{s}^{(L)}\right)^{-1}(A)-\Upsilon_{s}^{-1}(A)\right\Vert_{\A_{\infty}},
$$
and from Corollary \ref{corol: automorphismY}, the right hand side of last inequality is zero, thus (c) $\Rightarrow$ (a).
\end{proof}
\appendix
\section{Disordered models on general graphs}\label{appendix: GAM}
Consider the \emph{graph} $\mathfrak{G}\doteq\mathfrak{V}\times\mathfrak{E}$, where $\mathfrak{V}$ is the so--called set of \emph{vertices} and $\mathfrak{E}$ is called set of \emph{edges}. A graph has the following basic properties:
\begin{enumerate}
 \item For any, $v,w\in\mathfrak{V}$, and  $(v,w)\in\mathfrak{V}\times\mathfrak{V}$, $v$ and $w$ are called the \emph{endpoints} of $(v,w)\in\mathfrak{E}$.
 \item For $v,w\in\mathfrak{V}$, the vertices $\mathfrak{E}$ set does not contain element of the form $(v,v)$.
 \item Unless otherwise indicated, the edges set is \emph{not--oriented}: $(v,w)\in\mathfrak{E}$ iff $(w,v)\in\mathfrak{E}$.
 \item For simplicity, the element $g\in\mathfrak{G}$ is written as $g\equiv(v,e)$ for some $v\in\mathfrak{V}$ and $e\in\mathfrak{E}$.
 \item For $\epsilon\in(0,1]$ and any $v,w\in\mathfrak{V}$, one can endow with $\mathfrak{G}$ of a pseudometric $\mathfrak{d}_{\epsilon}\colon\mathfrak{G}\times\mathfrak{G}\to\R_{0}^{+}\cup\{\infty\}$, that is, an equivalence relation satisfying the metric properties on $\mathfrak{G}$, except that $\mathfrak{d}_{\epsilon}(v,w)=0$ does not implies that $v=w$. $\mathfrak{d}_{\epsilon}$ is closely related to the size of the \emph{path} with the minimum number of edges \emph{joining} the vertices $v$ and $w$.
 \item For $\mathfrak{G}$, $\mathcal{P}_{\text{f}}(\mathfrak{G})\subset2^{\mathfrak{G}}$ will denote the set of all finite subsets of $\mathfrak{G}$.
\end{enumerate}
We refer the reader to \cite{lyons2017probability} for a complete discussion about graphs.\\
Take $d\in\N$. Among the graphs that physicists consider, the \emph{$d$--dimensional cubic lattice} or \emph{crystal} $\Z^{d}$ is taken as a subset of $\R^{d}$ in the following way\footnote{Because of its spatial symmetric properties: translations, rotations.}:
\begin{equation}\label{eq: Zd}
\Z^{d}\doteq\{(x_{1},\ldots,x_{d})\in\R^{d}\colon x_{j}\in\Z\text{ for any }1\leq j\leq d\}.
\end{equation}
However, one can take more general assumptions by considering \emph{Cayley graphs}, which are defined via the group $\mathfrak{V}\equiv(\mathfrak{V},\cdot)$ generated by the subset $\mathfrak{v}\equiv(\mathfrak{v},\cdot)$. Then, we associate to any element of $\mathfrak{V}$ a vertex of the Cayley graph $\mathfrak{G}$ and the set of edges is defined by
$$
\mathfrak{E}\doteq\left\{(v,w)\in\mathfrak{V}^{2}\colon v^{-1}w\in\mathfrak{v}\right\}.
$$
In the $\Z^{d}$ case, the group $\mathfrak{V}\equiv(\mathfrak{V},+)$ is the so--called \emph{translation} group.\\
From the physical point of view, mobility or confinement of particles \emph{embedded} in a graph $\mathfrak{G}\doteq\mathfrak{V}\times\mathfrak{E}$ will rely on the \emph{impurities} of the material, crystal lattice defects (as in the $\Z^{d}$ case), etc., which usually are modeled (in the simplest case) by random (one--site) external potentials on the set of vertices $\mathfrak{V}$ as follows: We take the probability space $(\Omega,\mathfrak{A}_{\Omega},\mathfrak{a}_{\Omega})$, where $\Omega\doteq[-1,1]^{\mathfrak{V}}$. For any $v\in\mathfrak{V}$, $\Omega_{v}$ is an arbitrary element of the Borel $\sigma$--algebra $\mathfrak{A}_{v}$ of the Borel set $[-1,1]$ w.r.t. the usual metric topology. Then, $\mathfrak{A}_{\mathfrak{V}}$ is the $\sigma$--algebra generated by the cylinder sets $\prod\limits_{v\in\mathfrak{V}}\Omega_{v}$, where $\Omega_{v}=[-1,1]$ for all but finitely many $v\in\mathfrak{V}$. Additionally, we assume that the distribution $\mathfrak{a}_{\Omega}$ is an arbitrary \emph{ergodic} probability measure on the measurable space $(\Omega ,\mathfrak{A}_{\Omega})$. I.e., it is invariant under the action
$$
\rho\longmapsto\chi_{v}^{(\Omega)}(\rho)\doteq\chi_{v}^{(\mathfrak{V})}(\rho),\qquad v\in\mathfrak{V},
$$
of the group $\mathfrak{V}\equiv(\mathfrak{V},\cdot)$ on $\Omega$ and $\mathfrak{a}_{\Omega}(\mathcal{O})\in \{0,1\}$ whenever $\mathcal{O}\in\mathfrak{A}_{\Omega }$ satisfies $\chi _{e}^{(\Omega )}\left(\mathcal{O}\right)=\mathcal{O}$ for all $v\in\mathfrak{V}$. Here, for any $\rho\in \Omega$, $v\in\mathfrak{V}$ and $w\in\mathfrak{V}$
$$
\chi _{v}^{(\mathfrak{E})}\left(\rho\right)\left(w\right)\doteq\rho\left(v^{-1}w\right).
$$
As is usual, $\mathbb{E}\left[\,\cdot\,\right]$ denotes the expectation value associated with $\mathfrak{a}_{\Omega }$.\\
For the Cayley graph $\mathfrak{G}\doteq\mathfrak{V}\times\mathfrak{E}$, $\h\doteq\ell^{2}(\mathfrak{G},\CP)$ will denote a separable Hilbert space associated to $\mathfrak{G}$ with scalar product $\inner{\cdot,\cdot}_{\h}$ and canonical orthonormal basis denoted by $\{\E_{v}\}_{v\in\mathfrak{V}}$, which is defined by $\E_{v}(w)=\delta_{v^{-1}w,\1_{\mathfrak{V}}}$ for all $v, w \in\mathfrak{v}$, with $\mathfrak{v}$ the generator set of $\mathfrak{V}$ and $\1_{\mathfrak{V}}$ the unit on $\mathfrak{V}$. For any, $\rho\in\Omega$, one introduces the external potential $V_{\rho}\in\BL(\h)$ as the self--adjoint multiplication operator operator $V_{\rho}\colon\mathfrak{V}\to[-1,1]$. On the other hand, one defines for the compact set $\cal\doteq[0,1]$, the family of graph Laplacians $\{\Delta_{\mathfrak{G},s}\}_{s\in\cal}$ defined for any $s\in\cal$ by
\begin{eqnarray}
\lbrack\Delta_{\mathfrak{V},s}(\psi)](v)\doteq\mathrm{deg}_{\mathfrak{V}}(v)\psi(v)-s\sum_{p\in\mathfrak{V}\colon d_{\epsilon}(v,w)=1}\psi(v^{-1}w),\quad v\in\mathfrak{V},\,\psi\in\h\label{equation sup}
\end{eqnarray}
where for any $\epsilon\in(0,1]$, $d_{\epsilon}\colon\mathfrak{V}\times\mathfrak{V}\to\R_{0}^{+}\cup\{\infty\}$ is a pseudometric on $\mathfrak{G}$. In \eqref{equation sup}, on the right hand side, $\mathrm{deg}_{\mathfrak{V}}(v)$ is the number of nearest neighbors to vertex $v$, or \emph{degree} of $v$. If $\left\{\mathrm{deg}_{\mathfrak{V}}(v)\right\}_{v\in\mathfrak{V}}\in\N$ is the same for all $v\in\mathfrak{V}$, we say that the graph is regular.\\
The random tight--binding (Anderson) model is the one--particle Hamiltonian defined by
\begin{equation}
h_{\mathfrak{V},s}^{(\rho)}\doteq\Delta_{\mathfrak{V},s}+\lambda V_{\rho},\qquad \rho\in\Omega,\,\lambda\in\R_{0}^{+}.  \label{eq:Ham_lap_pot}
\end{equation}
See \cite{AW16} for further details. In \cite{ABdS3}, we consider a more general setting such that hopping disorder is present, i.e., we associate to particles a hopping probability on the non--oriented edges $\mathfrak{E}$. In this case, one deals with \emph{hopping amplitudes} and the probability space $(\Omega,\mathfrak{A}_{\Omega},\mathfrak{a}_{\Omega})$ is properly implemented.
\section{Fermionic Fock space and parity of the vacuum vector}\label{appendix: fock}
Let $(\H,\a)$ be a self--dual Hilbert space as defined in Section \ref{Self--dual CAR Algebras} and take $P\in\fp$, a basis projection, with range $\mathrm{ran}(P)=\h_{P}$. For $n\in\N_{0}$, let $\h_{P}^{0}\doteq\CP$ and for $n\in\N$ define
$$
\h_{P}^{n}\doteq\mathrm{lin}\{\varphi_{1}\otimes\cdots\otimes \varphi_{n}\colon \varphi_{1},\ldots,\varphi_{n}\in\h_{P}\}.
$$
The set $\varphi_{1},\ldots,\varphi_{n}\in\h_{P}$ denotes $n$ \emph{state vectors} of a single particle. Thus, the element $\varphi_{1}\otimes\cdots\otimes \varphi_{n}\in\h_{P}^{n}$ associate the state of the particle 1 in the state $\varphi_{1}$, the particle 2 in the state $\varphi_{2}$, and so on \cite{attal2006open}. Then, the Fock (Hilbert) space is nothing but
$$
 \Fs(\h_{P})\doteq\bigoplus_{n\geq0}\h_{P}^{n},
$$
where, as always, this infinite direct sum of Hilbert spaces is the subspace of the product space $\prod\limits_{n=0}^{\infty }\h_{P}^{n}$, with elements zero for all but for a finite number of these. An element $\Upsilon\in\Fs(\h_{P})$ is the sequence of functions $\{Y_{n}\}_{n\geq0}$ such that $Y_{0}\in\CP$ and $Y_{n}\in\h_{P}^{n}$ for $n\in\N$ \cite{reed1981functional}:
\begin{equation}\label{eq:Psi_seq}
\Upsilon\doteq\{Y_{0},Y_{1}(\phi_{1}^{*}),Y_{2}(\phi_{1}^{*},\phi_{2}^{*}),\ldots\},
\end{equation}
with $Y_{n}\doteq \varphi_{1}\otimes\cdots\otimes \varphi_{n}\in\h_{P}^{n}$ and
\begin{equation}\label{eq:f_otimes_f}
(\varphi_{1}\otimes\cdots\otimes \varphi_{n})(\phi_{1}^{*},\ldots,\phi_{n}^{*})\doteq \phi_{1}^{*}(\varphi_{1})\cdots \phi_{n}^{*}(\varphi_{n}),\qquad \phi_{1}^{*},\ldots,\phi_{n}^{*}\in\h_{P}^{*}.
\end{equation}
Naturally, the inner product on $\Fs(\h_{P})$ is given by
$$
\inner{\Upsilon,\Phi}_{\Fs(\h_{P})}\doteq\sum_{n\geq0}\inner{\varphi_{n},\phi_{n}}_{\h_{P}^{n}}.
$$
for $\Upsilon,\Phi\in\Fs(\h_{P})$. We then define the completely antisymmetric $n$--linear form $\varphi_{1}\wedge\cdots\wedge \varphi_{n}\in\wedge^{n}\h_{P}$ as
\begin{equation}\label{eq:linearform}
\varphi_{1}\wedge\cdots\wedge \varphi_{n}\doteq\sum_{\pi\in\mathcal{S}_{n}}\varepsilon_{\pi}\varphi_{\pi(1)}\otimes\cdots\otimes \varphi_{\pi(n)},
\end{equation}
where $\mathcal{S}_{n}$ denotes the set of all permutations of $n\in\N$ elements, with $\varepsilon_{\pi}$ equals $+1$ or $-1$ if the permutation is even or odd respectively. Note that, for any permutation $\varepsilon_{\pi}$ of $n\in\N$ elements we have
$$
\varphi_{1}\wedge\cdots\wedge \varphi_{n}=\varepsilon_{\pi}\varphi_{\pi(1)}\wedge\cdots\wedge \varphi_{\pi(n)},\qquad \varphi_{1},\ldots,\varphi_{n}\in\h_{P}.
$$
For $n\in\N_{0}$ we use that $\wedge^{0}\h_{P}\doteq\CP$ and for $n\in\N$ we define
$$
\wedge^{n}\h_{P}\doteq\mathrm{lin}\{\varphi_{1}\wedge\cdots\wedge \varphi_{n}\colon \varphi_{1},\ldots,\varphi_{n}\in\h_{P}\}.
$$
Note that by \eqref{eq:f_otimes_f}, \eqref{eq:linearform} and the \emph{Leibniz formula} for determinants we are able to write
$$
(\varphi_{1}\wedge\cdots\wedge \varphi_{n})(\phi_{1}^{*},\ldots,\phi_{n}^{*})=\det\left(\left(\phi_{i}^{*}(\varphi_{j})\right)_{i,j}^{n}\right),\qquad \phi_{1}^{*},\ldots,\phi_{n}^{*}\in\h_{P}^{*}.
$$
Then, for the Hilbert space $\h_{P}$, here we define the fermionic Fock space by
$$
\wedge\h_{P}\doteq\bigoplus_{n\geq0}\wedge^{n}\h_{P}.
$$
Note that the subspace of $\wedge\h_{P}$ generated by monomials $\varphi_{1},\ldots,\varphi_{n}$ of even order $n\in2\N_{0}$ forms a commutative subalgebra, the even subalgebra of $\wedge\h_{P}$, and it is denoted by $\wedge_{+}\h_{P}$. Then, for $H\in\BL(\H)$ be a self--dual element, $H^{*}=-\a H\a$,  and an orthonormal basis $\{\psi _{j}\}_{j\in J}$ of $\h_{P}$, the \emph{bilinear} element $\inner{\h_{P},H\h_{P}}$ on the fermionic Fock space $\wedge\h_{P}$ is defined by
$$
\inner{\h_{P},H\h_{P}}\doteq \sum\limits_{i,j\in J}\inner{\psi_{i},H\psi _{j}}_{\H}\left(\a\psi _{j}\right)\wedge \psi _{i}\in\wedge_{+}\h_{P},
$$
and hence we use the exponent function in $\wedge\h_{P}$
$$
\e^{\zeta}\doteq\1+\sum_{k=1}^{2\dim\h_{P}}\frac{\zeta^{n}}{k!}, \qquad \zeta\in\wedge\h_{P},
$$
in order to define the \emph{Gaussian element} $\e^{\inner{\h_{P},H\h_{P}}}\in\wedge_{+}\h_{P}$.\\
The \emph{vacuum} vector denoted by $\Omega\in\wedge\h_{P}$ is such that $[\Omega]_{0}\doteq1\in\h_{P}^{0}$ and $[\Omega]_{n}\doteq0\in\h_{P}^{n}$ for $n\geq1$, thus, physically $\Omega$ is associated to the state $(1,0,0,\ldots)$ without fermions. See \eqref{eq:Psi_seq}. The maps $a\colon\wedge^{n}\h_{P}\to\wedge^{n-1}\h_{P}$ and $a^{*}\colon\wedge^{n}\h_{P}\to\wedge^{n+1}\h_{P}$ are the so--called  ``annihilation'' and ``creation'' operators, respectively. For $\varphi,\varphi_{1},\ldots,\varphi_{n}\in\h_{P}$ they are defined by
\begin{align}\label{eq:anni_crea}
a(\varphi)(\varphi_{1}\wedge\cdots\wedge \varphi_{n})&=\sum_{k=1}^{n}(-1)^{k-1}\inner{\varphi,\varphi_{k}}_{\h_{P}}\varphi_{1}\wedge\cdots\wedge\breve{\varphi}_{k}\wedge\cdots\wedge \varphi_{n},\\
a^{*}(\varphi)(\varphi_{1}\wedge\cdots\wedge \varphi_{n})&=\varphi\wedge \varphi_{1}\wedge\cdots\wedge \varphi_{n}\nonumber
\end{align}
where the symbol $\breve{\phantom{\cdot}}$ means that the corresponding coordinate $\varphi_{k}$ was omitted. Hence, $a(\varphi)\Omega=0$ and $a^{*}(\varphi)\Omega=\varphi$ for all $\varphi\in\h_{P}$. Here, for $\varphi\in\h_{P}$, the involution of $a(\varphi)\in\BL(\wedge\h_{P})$, namely $a(\varphi)^{*}\in\BL(\wedge\h_{P})$, is canonically identified with $a^{*}(\varphi)$, i.e., $a^{*}(\varphi)\equiv a(\varphi)^{*}$. Then for $n\in\N$, $P\in\fp$ and $\varphi_{1},\ldots,\varphi_{n}\in\h_{P}$ we can define the element of size $n$ in $\wedge\h_{P}$ by
$$
\varphi_{1}\wedge\cdots\wedge\varphi_{n}=a^{*}(\varphi_{1})\cdots a^{*}(\varphi_{n})\Omega \in\wedge\h_{P}.
$$
Additionally, we can show that the canonical anticommutation relations hold
$$
a(\varphi_{1})a^{*}(\varphi_{2})+a^{*}(\varphi_{2})a(\varphi_{1})=\inner{\varphi_{1},\varphi_{2}}_{\h_{P}}\1_{\wedge\h_{P}},\quad a(\varphi_{1})a(\varphi_{2})+a(\varphi_{2})a(\varphi_{1})=0.
$$
Hence, the family of operators $\{a(\varphi)\}_{\varphi\in\h_{P}}$ and $\1_{\wedge\h_{P}}$ generate a $\CAR$ $\C$--algebra. By \cite[Theorem 5.2.5]{BratteliRobinson} there is an injective homomorphism between the self--dual $\CAR$ algebra $\sCAR$ and the space of bounded operators acting on the fermionic Fock space $\pi_{P}\colon\sCAR\to\BL(\wedge\h_{P})$, which is \emph{Fock representation} of the $\CAR$ algebra. In the finite dimension situation, this homomorphism is even a $^{*}$--isomorphism of $\C$--algebras. Explicitly, for any $P\in\fp$ and $\varphi\in\H$ we have
\begin{equation}\label{eq:fock_rep}
\pi_{P}(\B(\varphi))\doteq a(P\varphi)+a^{*}(\a P^{\perp}\varphi),\qquad \varphi\in\H,
\end{equation}
c.f., Expression \eqref{map iodiote}. The \emph{Fock state} is canonically defined by
\begin{equation}\label{eq:focK_state}
\omega_{P}\left(A\right)=\inner{\Omega,\pi_{P}(A)\Omega}_{\wedge\h_{P}},\qquad A\in\sCAR.
\end{equation}
Let now, $\{\psi_{j}\}_{j\in J}$ be an orthogonal basis of $\h_{P}$, we define the element of size $|J|$ in $\wedge\h_{P}$ as
\begin{equation}\label{eq:vacuum_P}
\Omega_{P}\doteq\psi_{1}\wedge\cdots\wedge\psi_{|J|}=a^{*}(\psi_{1})\cdots a^{*}(\psi_{|J|})\Omega,
\end{equation}
where $\Omega$ is the vacuum vector in $\wedge\h_{P}$ while $a^{*}(\cdot)$ is the associated creation operator on $\wedge\h_{P}$, see Expressions \eqref{eq:anni_crea}. The GNS construction associated to $\Omega_{P}$ in \eqref{eq:vacuum_P} is written by
$$
(\H_{\omega_{P}},\pi_{\omega_{P}},\Omega_{\omega_{P}})\equiv
(\H_{\omega_{P}},\pi_{\omega_{P}},\Omega_{P}),
$$
with Fock state given by (c.f. \eqref{eq:focK_state})
$$
\omega_{P}\left(A\right)=\inner{\Omega_{P},\pi_{\omega_{P}}(A)\Omega_{P}}_{\wedge\h_{P}},\qquad A\in\A,
$$
with (see Equation \eqref{eq:GNS_state} and comments around it) \cite[Exer. 6.10]{EK98}. Now, take the even and odd parts $\sCAR^{\pm}\subset\sCAR$ of the self--dual $\CAR$ $\C$--algebra associated to the self--dual Hilbert space $(\H,\a)$, see Expressions \eqref{eq:even odd}, and let $\pi_{P}$ be the fermionic Fock representation associated to $P$ given by \eqref{eq:fock_rep}. Note that $\pi_{P}$ can be decomposed as two disjoint irreducible representations \cite{A87}: $\pi_{P}=\pi_{P}^{+}\oplus\pi_{P}^{-}$, where $\pi_{P}^{\pm}$ is the restriction of $\omega_{P}$ to $\sCAR^{\pm}$, and coincides with the restriction of $\pi_{P}(\A_{\pm})$ to the clousure $\wedge_{\pm}\h_{P}$ of $\pi_{\omega_{P}}(\sCAR^{\pm})\Omega$. In this way, $\left(\pi_{\omega_{P}}\right)_{\pm}$ is identified with $\pi_{P}^{\pm}$ or $\pi_{P}^{\mp}$ depending if $|J|$ in Expression \eqref{eq:vacuum_P} is even or odd \cite{EK98}.\par
Observe that we can study the parity of the vacuum vector $\Omega$ by using \emph{Clifford algebra} tools. In this framework, instead we need to consider \emph{orthogonal complex structures} $\mathcal{J}\in\BL(\H)$, that is, a linear endomorphism on $(\H,\a)$ satisfying $\mathcal{J}^{2}=-\mathbf{1}_{\H}$ and $\mathcal{J}^{*}=-\mathcal{J}$, as well as use suitable isomorphisms between self--dual $\CAR$--algebra $\sCAR$ and the Clifford algebra $\mathbb{C}\ell(\Ree(\H))$ generated by $\Ree(\H)$ \cite{Varilly}. Here, $\Ree(\H)\doteq\{\varphi\in\H\colon\a\varphi=\varphi\}$, and $\mathcal{J}\doteq\ii(2P-\mathbf{1}_{\H})_{\Ree(\H)}$, for any $P\in\fp$. We also endow to $\mathbb{C}\ell(\Ree(\H))$ with the inner product $\inner{\cdot,\cdot}_{\Ree(\H)}$, which is defined by a symmetric non--degenerated bilinear form $\mathcal{S}\colon\Ree(\H)\times\Ree(\H)\to\R$, so that $\mathcal{S}(\varphi_{1},\varphi_{2})\doteq\inner{\varphi_{1},\varphi_{2}}_{\Ree(\H)}$ for any $\varphi_{1},\varphi_{2}\in\Ree(\H)$. In \cite{CalReyes}  the parity of $\Omega$ was studied via orthogonal complex structures, and it is completely equivalent to that presented in this paper. In fact, one defines a $\Z_{2}$ topological index via two orthogonal complex structures $\mathcal{J}_{1},\mathcal{J}_{2}$ as follows
$$
\Sigma(\mathcal{J}_{1},\mathcal{J}_{2})\doteq(-1)^{\frac{1}{2}\dim\ker(\mathcal{J}_{1}+\mathcal{J}_{2})},
$$
which coincides with the one of Expression \eqref{eq:top_index} \cite{EK98}.
\section{\eqt{$\CAR$ $\C$}--Algebra}\label{appendix: CAR}
Let $\h$ be a separable Hilbert space with $\dim\h\in\N_{0}$, and consider the direct sum
$$
\H\doteq \h\oplus\h^{*}.
$$
Compare with Equation \eqref{definition H bar}. The scalar product on $\H$ is
% [Andrés, 03.03.2022]: el uso de \mathrm en la ecuación de abajo generó un problema que indicó Ling en las proofs...voy a quitar esos \mathrm
$$
\left\langle \varphi ,\tilde{\varphi}\right\rangle _{\H}\doteq
\left\langle \varphi_{1},\tilde{\varphi}_{1}\right\rangle
_{\h}+\left\langle \tilde{\varphi}_{2},\varphi_{2}\right\rangle_{\h},\qquad \varphi\doteq(\varphi_{1},%
\varphi_{2}^{*}),\, \tilde{\varphi}\doteq(\tilde{\varphi}_{1},
\tilde{\varphi}_{2}^{*})\in \H.
$$
% Esta era la versión problemática (con \mathrm):
%$$
%\left\langle \varphi ,\tilde{\varphi}\right\rangle _{\H}\doteq
%\left\langle \mathrm{\varphi }_{1},\mathrm{\tilde{\varphi}}_{1}\right\rangle
%_{\h}+\left\langle \mathrm{\tilde{\varphi}}_{2},\mathrm{\varphi }%
%_{2}\right\rangle_{\h},\qquad \varphi\doteq(\mathrm{\varphi }_{1},%
%\mathrm{\varphi }_{2}^{*}),\, \tilde{\varphi}\doteq(\mathrm{\tilde{\varphi}}%
%_{1},\mathrm{\tilde{\varphi}}_{2}^{*})\in \H.
%$$
Here, $\varphi^{*}$ denotes the element of the dual $\h^{*}$ of the Hilbert space $\h$ which is related to $\varphi $ via the Riesz representation. We define the canonical antiunitary involution $\Gamma$ of $\H$ by
$$
\Gamma\left(\varphi_{1},\varphi_{2}^{*}\right) \doteq \left( \varphi_{2},\varphi_{1}^{*}\right),\qquad \varphi \doteq(\varphi_{1},\varphi_{2}^{*})\in\H.  \label{antiunitary simple}
$$
Note that $\varphi ^{* }=\Gamma\varphi $ for any $\varphi\in\h\subset \mathcal{\H}$. Then, the $\CAR$ algebra $\CAR(\h)$ and the self--dual $\CAR$ algebra $\sCAR$ are the same $\C$--algebra, by defining
\begin{equation}
\mathrm{B}\left(\varphi \right) \equiv \mathrm{B}_{P_{\h}}\left(
\varphi \right) \doteq a(\varphi_{1})+a(\varphi_{2})^{* },\qquad \varphi =
(\varphi_{1},\varphi_{2}^{* })\in\H,  \label{map iodiote2}
\end{equation}%
with $P_{\h}\in \mathcal{B}(\mathcal{\H})$ being the basis projection of $(\H,\Gamma)$ with range $\h$. See \eqref{CAR Grassmann III} and \eqref{map iodiote}.\\
We now consider the $\CAR$ $\C$--algebra generated by the identity $\mathfrak{1}$ and the elements $\{a(\psi )\}_{\psi \in \mathfrak{h}}$ satisfying the canonical anticommutation relations ($\CAR$): For all $\psi ,\varphi \in \mathfrak{h}$,
$$
a(\psi )a(\varphi )=-a(\varphi )a(\psi ),\quad a(\psi )a(\varphi )^{\ast
}+a(\varphi )^{\ast }a(\psi )=\left\langle \psi ,\varphi \right\rangle _{%
\mathfrak{h}}\mathfrak{1}.
$$
As is usual, $a(\psi )$ and $a(\psi )^{\ast }$ are called, respectively, annihilation and creation operators of a fermion in the state $\varphi \in \mathfrak{h}$. One usually consider fermionic Hamiltonians of the form
\begin{equation}
\mathbf{H}=\mathrm{d}\Phi(h)+\mathrm{d}\Upsilon (g)+W,\qquad W=W^{\ast}\in \mathrm{CAR}(\mathfrak{h}),  \label{form}
\end{equation}
where, for any $h=h^{\ast}\in \mathcal{B}(\mathfrak{h})$ and \emph{antilinear} operator $g=-g^{\ast }$ on $\mathfrak{h}$,
\begin{align*}
\mathrm{d}\Phi (h) &\doteq \sum\limits_{i,j\in J}\left\langle \psi
_{i},h\psi _{j}\right\rangle _{\mathfrak{h}}a\left( \psi _{i}\right) ^{\ast
}a\left( \psi _{j}\right), \\
\mathrm{d}\Upsilon (g) &\doteq \frac{1}{2}\sum\limits_{i,j\in J}\left(
\left\langle \psi _{i},g\psi _{j}\right\rangle _{\mathfrak{h}}a\left( \psi
_{i}\right) ^{\ast }a\left( \psi _{j}\right) ^{\ast }+\overline{\left\langle
\psi _{i},g\psi _{j}\right\rangle _{\mathfrak{h}}}a\left( \psi _{j}\right)
a\left( \psi _{i}\right) \right).
\end{align*}
Here, $\{\psi _{j}\}_{j\in J}$ is any orthonormal basis of $\mathfrak{h}$. In \eqref{form}, $\mathrm{d}\Phi(h)=\mathrm{d}\Phi(h)^{*}\in \mathrm{CAR}\left(\mathfrak{h}\right) $ is the \emph{second quantization} of the one--particle Hamiltonian $h\in \mathcal{B}(\mathfrak{h})$. It represents a \emph{gauge} invariant model of free fermions. On the other hand, $\mathrm{d}\Upsilon (g)=\mathrm{d}\Upsilon (g)^{*}\in \mathrm{CAR}\left( \mathfrak{h}\right)$ represents the \emph{non--gauge} invariant quadratic part of $\mathbf{H}$. Finally, $W$ encodes the interparticle interaction of the fermion system.\par
Let now $P\in\fp$ be a basis projection with $\text{ran}(P)=\h_{P}$. For any $h=h^{*}\in \mathcal{B}(\h_{P})$ and antilinear operator $g=-g^{*}$ on $\h_{P}$ are defined the maps $\varkappa$ and $\tilde{\varkappa}$ by:
\begin{equation*}
\varkappa\left(h\right)\doteq\frac{1}{2}\left(PhP-\a PhP\a\right)\in\BL(\h_{P})
\qquad\text{and}\qquad
\tilde{\varkappa}\left(g\right) \doteq \frac{1}{2}\left( PgP\Gamma-\Gamma PgP\right) \in \mathcal{B}(\H).
\end{equation*}
Observe that
\begin{equation*}
\varkappa\left(h\right)^{*}=\varkappa\left(h\right)=-\Gamma \varkappa\left(h\right)\Gamma
\qquad\text{and}\qquad
\tilde{\varkappa}\left(g\right)^{*}=\tilde{\varkappa}\left(g\right)=-\Gamma\tilde{\varkappa}\left(g\right)\Gamma,
\end{equation*}
thus $\varkappa$ and $\tilde{\varkappa}$ provide self--dual Hamiltonians, see Definition \ref{def one particle hamiltinian}. By \eqref{map iodiote2} we can write
$$
\d\Phi (h)+\d\Upsilon (g)=-\inner{\B,\left[\varkappa \left( h\right) +\tilde{\varkappa}\left( g\right) \right]\B} +\frac{1}{2}\tr_{\h_{P}}\left( h\right) \mathfrak{1}\in\mathrm{CAR}(\h_{P}),
$$
that is, a bilinear element of a self--dual Hamiltonian plus a constant term. For further details see \cite[Sects. 1 and 6.]{LD1}. For $W=0$ in \eqref{form}, the dynamics is provided by the continuous group $\{\tau_{t}\}_{t\in\R}$ of $^{*}$--automorphisms of $\CAR(\h_{P})$ defined by
$$
\tau_{t}(A)\doteq\e^{\ii t\mathbf{H}}A\e^{-\ii t\mathbf{H}},\qquad A\in\CAR(\h_{P}),
$$
which, by above comments, can be written as
\begin{align*}
\tau_{t}(A)&=\e^{\ii t\left(\d\Phi (h)+\d\Upsilon (g)\right)}A\e^{-\ii t\left(\d\Phi (h)+\d\Upsilon (g)\right)}\\
&=\e^{-\ii t\inner{\B,\left[\varkappa \left( h\right) +\tilde{\varkappa}\left( g\right) \right]\B}}A\e^{\ii t\inner{\B,\left[\varkappa \left( h\right) +\tilde{\varkappa}\left( g\right) \right]\B}}.
\end{align*}
The latter equation provides exactly the quasi--free dynamics given by Equation \eqref{eq:autoSCAR}.\\
Now, instead of considering $\mathbf{H}$, observe that for any basis projection $P\in\fp$ equivalently
\begin{equation}\label{eq: free HG1}
-\inner{\B,\left[F+G\right]\B}+\tr_{\h_{P}}\left(PFP\right)\mathfrak{1}\in\mathrm{CAR}(\h_{P}),
\end{equation}
gives us the description of \emph{all} free--fermion systems, where $F\in\BL(\h_{P})$ and $G\in\mathcal{L}(\h_{P})$ are self--dual Hamiltonians on $\H$. As mentioned in \eqref{kappabisbiskappabisbis}, $F_{P}\doteq2PFP$ is the so--called one--particle Hamiltonian, then, w.l.o.g. we can \emph{remove} the term $\tr_{\h_{P}}\left(PFP\right)\mathfrak{1}$, by writing \eqref{eq: free HG1} as
$$
-\inner{\B,\left[\widetilde{F}+G\right]\B}\in\mathrm{CAR}(\h_{P}),
$$
for $\widetilde{F}\doteq F-\frac{1}{|I|}\tr_{\h_{P}}(PFP)\varkappa(\1_{\h_{P}})$, with $|I|$ the cardinality of the Hilbert space $\H$. For $H\doteq\widetilde{F}+G$, a self--dual Hamiltonian, define $h\doteq2P_{\h}HP_{\h}$ and $g\doteq2P_{\h}H\a P_{\h}$, in order to describe \emph{any} quadratic Fermionic Hamiltonian. In fact, given $P\in\fp$ with $\text{ran}(P)=\h_{P}$ and the self--dual Hamiltonian $H\in\BL(\H)$ the operators
$$
h\doteq 2PHP\in\BL(\h_{P})\qquad\text{and}\qquad g\doteq2PHP\a\in\mathcal{L}(\H),
$$
provide \emph{all} the possible free--fermion models.

%\bigskip
%\noindent \textbf{Data availability statement:} Data sharing not applicable to this article as no datasets were generated or analysed during the current study.

\bigskip
\noindent \textit{Acknowledgments:} We are very grateful to the anonymous Referee, whose feedback and constructive criticism helped us to significantly improve our work. We are also very grateful to Y. Ogata for hints and references. Financial support from the Faculty of Sciences of \emph{Universidad de los Andes} through project INV-2019-84-1833 is gratefully acknowledged.
%\bibliography{books}

\begin{thebibliography}{TKNdN82}

\bibitem[ABPM21]{LD1}
N.~J.~B. {Aza}, J.-B. {Bru}, de~Siqueira~W. {Pedra}, and L.~C. P. A.~M.
  {Müssnich}, \emph{{Large Deviations in Weakly Interacting Fermions:
  Generating Functions as Gaussian Berezin Integrals and Bounds on Large
  Pfaffians.}}, Reviews in Mathematical Physics (2021).

\bibitem[ABPR19]{ABdS3}
N.~J.~B. {Aza}, {J–B} {Bru}, de~Siqueira~W. {Pedra}, and
  A~{Ratsimanetrimanana}, \emph{{Accuracy of Classical Conductivity Theory at
  Atomic Scales for Free Fermions in Disordered Media}}, Journal de
  Mathématiques Pures et Appliquées (2019).

\bibitem[AE83]{araEva}
H.~Araki and D.~E. Evans, \emph{{On a C*–algebra approach to phase transition
  in the two–dimensional Ising model}}, Communications in Mathematical
  Physics \textbf{91} (1983), no.~4, 489–503.

\bibitem[AG98]{AG98}
M.~Aizenman and G.~M. Graf, \emph{{Localization bounds for an electron gas}},
  Journal of Physics A: Mathematical and General \textbf{31} (1998), no.~32,
  6783.

\bibitem[AJP06]{attal2006open}
S.~Attal, A.~Joye, and C.A. Pillet, \emph{{Open Quantum Systems I: The
  Hamiltonian Approach}}, {Lecture Notes in Mathematics}, Springer, 2006.

\bibitem[AM03]{Araki-Moriya}
H.~Araki and H.~Moriya, \emph{{Equilibrum Statistical Mechanics of Fermion
  Lattice Systems}}, Reviews in Mathematical Physics \textbf{15} (2003),
  no.~02, 93–198.

\bibitem[AMR]{AMR2}
N.~J.~B. {Aza}, L.~C. P. A.~M. {Müssnich}, and A.~F. {Reyes-Lega}, \emph{{A
  $\mathbb{Z}_{2}$ Topological Index for Interacting Fermions Systems}}, To
  appear.

\bibitem[Ara68]{A68}
H.~Araki, \emph{{On the diagonalization of a bilinear Hamiltonian by a
  Bogoliubov transformation}}, Publications of the Research Institute for
  Mathematical Sciences, Kyoto University. Ser. A \textbf{4} (1968), no.~2,
  387–412.

\bibitem[Ara71]{A70}
\bysame, \emph{{On quasifree states of CAR and Bogoliubov automorphisms}},
  Publications of the Research Institute for Mathematical Sciences \textbf{6}
  (1971), no.~3, 385–442.

\bibitem[Ara87]{A87}
\bysame, \emph{{Bogoliubov Automorphisms and Fock Representations of Canonical
  Anticommutation Relations}}, Contemp. Math \textbf{62} (1987), 23–141.

\bibitem[Ara88]{A88}
\bysame, \emph{{Schwinger terms and cyclic cohomology}}, {Quantum Theories and
  Geometry}, Springer, 1988, p.~1–22.

\bibitem[AT85]{araki1985ground}
H.~Araki and Matsui T., \emph{{Ground states of the XY–model}},
  Communications in Mathematical Physics \textbf{101} (1985), no.~2, 213–245.

\bibitem[AW15]{AW16}
M.~Aizenman and S.~Warzel, \emph{{Random operators: Disorder effects on quantum
  spectra and dynamics, volume 168 of}}, Graduate Studies in Mathematics
  (2015).

\bibitem[AZ97]{Altland}
A.~Altland and M.~R. Zirnbauer, \emph{{Nonstandard symmetry classes in
  mesoscopic normal-superconducting hybrid structures}}, Physical Review B
  \textbf{55} (1997), no.~2, 1142–1161.

\bibitem[BCR16]{Bourne}
C.~Bourne, A.~L. Carey, and A.~Rennie, \emph{{A non–commutative framework for
  topological insulators}}, Reviews in Mathematical Physics \textbf{28} (2016),
  no.~02, 1650004.

\bibitem[BDF18]{bachRoeckFrass}
S.~Bachmann, W.~{De Roeck}, and M.~Fraas, \emph{{The adiabatic theorem and
  linear response theory for extended quantum systems}}, Communications in
  Mathematical Physics (2018), no.~3, 997–1027.

\bibitem[BMNS12]{bachmann2012automorphic}
S.~Bachmann, S.~Michalakis, B.~Nachtergaele, and R.~Sims, \emph{{Automorphic
  equivalence within gapped phases of quantum lattice systems}}, Communications
  in Mathematical Physics \textbf{309} (2012), no.~3, 835–871.

\bibitem[BP13]{BruPedra2}
{J–B} {Bru} and de~Siqueira~W {Pedra}, \emph{{Non–cooperative equilibria of
  Fermi systems with long range interactions}}, vol. 224, American Mathematical
  Soc., 2013.

\bibitem[BP16a]{brupedraLR}
\bysame, \emph{{Lieb–Robinson Bounds for Multi–Commutators and Applications
  to Response Theory}}, Springer \textbf{13} (2016).

\bibitem[BP16b]{universal}
\bysame, \emph{{Universal Bounds for Large Determinants from Non Commutative
  Hölder Inequalities in Fermionic Constructive Quantum Field Theory}},
  Preprint: mp arc 16-16 (2016).

\bibitem[BPH14]{bru2014heat}
{J–B} {Bru}, W~de~Siqueira {Pedra}, and C.~{Hertling}, \emph{{Heat production
  of Noninteracting fermions subjected to electric fields}}, Communications on
  Pure and Applied Mathematics (2014).

\bibitem[BR03a]{BratteliRobinsonI}
O.~Bratteli and D.W. Robinson, \emph{{Operator Algebras and Quantum Statistical
  Mechanics 1: C*– and W*–Algebras. Symmetry Groups. Decomposition of
  States}}, 2 ed., {Operator Algebras and Quantum Statistical Mechanics},
  Springer, 2003.

\bibitem[BR03b]{BratteliRobinson}
\bysame, \emph{{Operator Algebras and Quantum Statistical Mechanics:
  Equilibrium States. Models in Quantum Statistical Mechanics}}, 2 ed.,
  {Operator Algebras and Quantum Statistical Mechanics 2: Equilibrium States.
  Models in Quantum Statistical Mechanics}, Springer, 2003.

\bibitem[BSB20]{BouSBal}
C.~Bourne and H.~Schulz-Baldes, \emph{{On $\mathbf{Z}_{2}$-indices for ground
  states of fermionic chains}}, Reviews in Mathematical Physics (2020), no.~0,
  2050028.

\bibitem[BO21]{ogata2021}
C.~Bourne and Y.~Ogata, \emph{{The classification of symmetry protected topological phases of one-dimensional
fermion systems}}, Forum of Mathematics, Sigma, vol 9, Cambridge University Press, p. e25. (2021).
 
 
\bibitem[BvES94]{Bellisard}
J.~Bellissard, A.~van Elst, and H.~Schulz–Baldes, \emph{{The noncommutative
  geometry of the quantum Hall effect}}, Journal of Mathematical Physics
  \textbf{35} (1994), no.~10, 5373–5451.

\bibitem[CGRL18]{CalReyes}
J.~S. Calderón-García and A.~F. Reyes-Lega, \emph{{Majorana fermions and
  orthogonal complex structures}}, Modern Physics Letters A (2018), no.~14,
  1840001.

\bibitem[CHM{\etalchar{+}}06]{Carey1}
A.~L. Carey, K.~C. Hannabuss, V.~Mathai, et~al., \emph{{Quantum Hall effect and
  noncommutative geometry}}, Journal of Geometry and Symmetry in Physics
  \textbf{6} (2006), 16–37.

\bibitem[Cho]{Choquet1966}
G.~Choquet, Academic Press.

\bibitem[CNN18]{cha2018complete}
M.~Cha, P.~Naaijkens, and B.~Nachtergaele, \emph{{The complete set of infinite
  volume ground states for Kitaev’s abelian quantum double models}},
  Communications in Mathematical Physics \textbf{357} (2018), no.~1, 125–157.

\bibitem[DS19]{deRoecK_Salm}
W.~{De Roeck} and M.~Salmhofer, \emph{{Persistence of exponential decay and
  spectral gaps for interacting fermions}}, Communications in Mathematical
  Physics \textbf{365} (2019), no.~2, 773–796.

\bibitem[Dys62]{Dyson}
F.~J. Dyson, \emph{{The Threefold Way. Algebraic Structure of Symmetry Groups
  and Ensembles in Quantum Mechanics}}, Journal of Mathematical Physics
  \textbf{3} (1962), no.~6, 1199–1215.

\bibitem[EBN{\etalchar{+}}06]{EngelNagel}
K.J. Engel, S.~Brendle, R.~Nagel, T.~Hahn, G.~Metafune, G.~Nickel, D.~Pallara,
  C.~Perazzoli, A.~Rhandi, et~al., \emph{{One–Parameter Semigroups for Linear
  Evolution Equations}}, {Graduate Texts in Mathematics}, Springer New York,
  2006.

\bibitem[EK98]{EK98}
D.E. Evans and Y.~Kawahigashi, \emph{{Quantum symmetries on operator
  algebras}}, {Oxford Mathematical Monographs, ISSN 0964-9174}, Clarendon
  Press, 1998.

\bibitem[FMP16]{Fiorenza2016}
D.~{Fiorenza}, D.~{Monaco}, and G.~{Panati}, \emph{{Z(2) Invariants of
  Topological Insulators as Geometric Obstructions}}, Communications in
  Mathematical Physics (2016), no.~3, 1115–1157.

\bibitem[GBVF01]{Varilly}
J.~G. Gracia-Bondía, J.~C. Várilly, and H.~Figueroa, \emph{{Elements of
  Noncommutative Geometry}}, {Birkhäuser advanced texts. Basler Lehrbücher},
  Birkhäuser, 2001.

\bibitem[GJ16]{giuliani2016ground}
A.~{Giuliani} and I.~{Jauslin}, \emph{{The ground state construction of bilayer
  graphene}}, Reviews in Mathematical Physics \textbf{28} (2016), no.~08,
  1650018.

\bibitem[GMP16]{GMP}
A.~Giuliani, V.~Mastropietro, and M.~Porta, \emph{{Universality of the Hall
  Conductivity in Interacting Electron Systems}}, Communications in
  Mathematical Physics (2016), 1–55.

\bibitem[Has19]{hastings2019stability}
M.~B. Hastings, \emph{{The stability of free Fermi Hamiltonians}}, Journal of
  Mathematical Physics \textbf{60} (2019), no.~4, 042201.

\bibitem[HL11]{Hastings2011}
M.~B. {Hastings} and T.~A. {Loring}, \emph{{Topological insulators and
  $C^{*}$-algebras: Theory and numerical practice}}, Annals of Physics (2011),
  no.~7, 1699–1759, July 2011 Special Issue.

\bibitem[Kat13]{kato2013perturbation}
T.~Kato, \emph{{Perturbation Theory for Linear Operators}}, vol. 132, Springer
  Science \& Business Media, 2013.

\bibitem[Kit01]{kitaev2001}
A.~Y. Kitaev, \emph{{Unpaired Majorana fermions in quantum wires}},
  Physics-Uspekhi (2001), no.~10S, 131.

\bibitem[Kit09]{Kitaev}
A.~Kitaev, \emph{{Periodic table for topological insulators and
  superconductors}}, American Institute of Physics Conference Series
  \textbf{1134} (2009), 22–30.

\bibitem[KK18]{Katsura2018}
H.~Katsura and T.~Koma, \emph{{The noncommutative index theorem and the
  periodic table for disordered topological insulators and superconductors}},
  Journal of Mathematical Physics (2018), no.~3, 031903.

\bibitem[LH10]{Loring2010}
T.~A. {Loring} and M.~B. {Hastings}, \emph{{Disordered topological insulators
  via $C^{*}$–algebras}}, {EPL} (Europhysics Letters) (2010), no.~6, 67004.

\bibitem[LP17]{lyons2017probability}
R.~Lyons and Y.~Peres, \emph{{Probability on Trees and Networks}}, {Cambridge
  Series in Statistical and Probabilistic Mathematics}, Cambridge University
  Press, 2017.

\bibitem[{Mat}13]{matsui2013}
T.~{Matsui}, \emph{{Boundedness of entanglement entropy and split property of
  quantum spin chains}}, Reviews in Mathematical Physics \textbf{25} (2013),
  no.~09, 1350017.

\bibitem[Mat20]{matsui20}
T.~Matsui, \emph{{Split Property and Fermionic String Order}}, arXiv preprint
  arXiv:2003.13778 (2020).

\bibitem[MZ13]{michalakis2013stability}
S.~Michalakis and J.~P. Zwolak, \emph{{Stability of Frustration-Free
  Hamiltonians}}, Communications in Mathematical Physics \textbf{322} (2013),
  no.~2, 277–302.

\bibitem[NSY18a]{nachtergaele2018lieb}
B.~Nachtergaele, R.~Sims, and A.~Young, \emph{{Lieb-Robinson bounds, the
  spectral flow, and stability of the spectral gap for lattice fermion
  systems}}, Mathematical Problems in Quantum Physics \textbf{117} (2018), 93.

\bibitem[NSY18b]{nachtergaele2018quasi}
\bysame, \emph{{Quasi–Locality Bounds for Quantum Lattice Systems. Part I.
  Lieb–Robinson Bounds, Quasi–Local Maps, and Spectral Flow
  Automorphisms}}, arXiv preprint arXiv:1810.02428 (2018).

\bibitem[Oga20]{ogata2020}
Y.~Ogata, \emph{{A Z(2)-Index of Symmetry Protected Topological Phases with
  Time Reversal Symmetry for Quantum Spin Chains}}, Communications in
  Mathematical Physics \textbf{374} (2020), no.~2, 705–734.

\bibitem[PS16]{Prodan}
E.~{Prodan} and H.~{Schulz-Baldes}, \emph{{Bulk and Boundary Invariants for
  Complex Topological Insulators}}, Springer, 2016.

\bibitem[RS81]{reed1981functional}
M.~Reed and B.~Simon, \emph{{I: Functional Analysis}}, {Methods of Modern
  Mathematical Physics}, Elsevier Science, 1981.

\bibitem[RSFL10]{Ryu}
Shinsei Ryu, Andreas~P Schnyder, Akira Furusaki, and Andreas W~W Ludwig,
  \emph{{Topological insulators and superconductors: tenfold way and
  dimensional hierarchy}}, New Journal of Physics \textbf{12} (2010), no.~6,
  065010.

\bibitem[Rud91]{rudin}
W.~Rudin, \emph{{Functional Analysis}}, {International series in pure and
  applied mathematics}, McGraw–Hill, 1991.

\bibitem[TKNdN82]{TKNN}
D.~J. Thouless, M.~Kohmoto, M.~P. Nightingale, and Md. den Nijs,
  \emph{{Quantized Hall conductance in a two–dimensional periodic
  potential}}, Physical Review Letters \textbf{49} (1982), no.~6, 405.

\end{thebibliography}
%\bibliographystyle{amsalpha}

\newcommand{\etalchar}[1]{$^{#1}$}
\providecommand{\bysame}{\leavevmode\hbox to3em{\hrulefill}\thinspace}
\providecommand{\MR}{\relax\ifhmode\unskip\space\fi MR }
\providecommand{\MRhref}[2]{%
  \href{http://www.ams.org/mathscinet-getitem?mr=#1}{#2}
}
\providecommand{\href}[2]{#2}

\end{document}